\documentclass[11pt]{article}
\usepackage{graphicx,amsmath,amsfonts,bm,cite,epsfig,epsf,url}
\usepackage{fullpage}
\usepackage[small,bf]{caption}
\newtheorem{theorem}{Theorem}[section]
\newtheorem{lemma}[theorem]{Lemma}
\newtheorem{corollary}[theorem]{Corollary}

\newtheorem{definition}[theorem]{Definition}

\newcommand{\R}{\mathbb{R}}

\newcommand{\<}{\langle}
\renewcommand{\>}{\rangle}

\newcommand{\goto}{\rightarrow}
\newcommand{\sgn}{\textrm{sgn}}
\renewcommand{\P}{\operatorname{\mathbb{P}}}
\newcommand{\E}{\operatorname{\mathbb{E}}}
\newcommand{\norm}[1]{{\left\lVert{#1}\right\rVert}}

\newcommand{\e}{\mathrm{e}}

\newcommand{\PO}{{\cal P}_{\Omega}}
\newcommand{\PT}{{\cal P}_T}
\newcommand{\PTc}{{\cal P}_{T^\perp}}

\newcommand{\cAOT}{\mathcal{A}_{\Omega T}}
\newcommand{\eab}{\vct{e}_{a}\vct{e}_{b}^*}
\newcommand{\eabone}{\vct{e}_{a_1}\vct{e}_{b_1}^*}
\newcommand{\eabtwo}{\vct{e}_{a_2}\vct{e}_{b_2}^*}
\newcommand{\eabthree}{\vct{e}_{a_3}\vct{e}_{b_3}^*}
\newcommand{\eabp}{\vct{e}_{a'} \vct{e}_{b'}^*}
\newcommand{\dab}{\delta_{ab}}

\newcommand{\xiab}{\xi_{ab}}
\newcommand{\xiabp}{\xi_{a'b'}}
\newcommand{\xiabpp}{\xi'_{a'b'}}

\newcommand{\vct}[1]{\bm{#1}}
\newcommand{\mtx}[1]{\bm{#1}}

\newcommand{\lspan}[1]{\operatorname{span}{#1}}

\newcommand{\rank}{\operatorname{rank}}

\newcommand{\trace}{\operatorname{trace}}

\newcommand{\OpId}{\mathcal{I}}

\numberwithin{equation}{section}

\def \endprf{\hfill {\vrule height6pt width6pt depth0pt}\medskip}
\newenvironment{proof}{\noindent {\bf Proof} }{\endprf\par}

\title{Exact Matrix Completion via 
Convex Optimization}

\author{Emmanuel J. Cand\`es$^{\dagger}$ and Benjamin Recht$^{\sharp}$\\
  \vspace{-.1cm}\\
  $\dagger$ Applied and Computational Mathematics,
  Caltech, Pasadena, CA 91125\\
  \vspace{-.3cm}\\
  $\sharp$ Center for the Mathematics of Information, Caltech,
  Pasadena, CA 91125}

\date{May 2008}

\begin{document}

\maketitle

\vspace{-0.3in}

\begin{abstract}
  We consider a problem of considerable practical interest: the
  recovery of a data matrix from a sampling of its entries.
  Suppose that we observe $m$ entries selected uniformly at random
  from a matrix $M$. Can we complete the matrix and recover the
  entries that we have not seen?

  We show that one can perfectly recover most low-rank matrices
  from what appears to be an incomplete set of entries. We prove that
  if the number $m$ of sampled entries obeys
\[
m \ge C \, n^{1.2} r \log n
\]
for some positive numerical constant $C$, then with very high probability,
most $n \times n$ matrices of rank $r$ can be perfectly recovered by
solving a simple convex optimization program. This program finds the
matrix with minimum nuclear norm that fits the data.  The condition
above assumes that the rank is not too large. However, if one replaces
the 1.2 exponent with 1.25, then the result holds for all values of
the rank. Similar results hold for arbitrary rectangular matrices as
well.
Our results are connected with the recent
literature on compressed sensing, and show that objects other than
signals and images can be perfectly reconstructed from very limited
information.

\end{abstract}


\section{Introduction}
\label{sec:intro}

In many practical problems of interest, one would like to recover a
matrix from a sampling of its entries. As a motivating example,
consider the task of inferring answers in a partially filled out
survey.  That is, suppose that questions are being asked to a
collection of individuals.  Then we can form a matrix where the rows index
each individual and the columns index the questions. We collect data
to fill out this table but unfortunately, many questions are left
unanswered. Is it possible to make an educated guess about what the
missing answers should be? How can one make such a guess? Formally, we
may view this problem as follows. We are interested in recovering a
data matrix $\mtx{M}$ with $n_1$ rows and $n_2$ columns but only get
to observe a number $m$ of its entries which is comparably much
smaller than $n_1 n_2$, the total number of entries. Can one recover
the matrix $\mtx{M}$ from $m$ of its entries? In general, everyone
would agree that this is impossible without some additional
information.

In many instances, however, the matrix we wish to recover is known
to be structured in the sense that it is low-rank or approximately
low-rank.  (We recall for completeness that a matrix with $n_1$ rows
and $n_2$ columns has rank $r$ if its rows or columns span an
$r$-dimensional space.) Below are two examples of practical
scenarios where one would like to be able to recover a low-rank
matrix from a sampling of its entries.
\begin{itemize}
\item {\em The Netflix problem.} In the area of recommender systems,
  users submit ratings on a subset of entries in a database, and the
  vendor provides recommendations based on the user's
  preferences~\cite{Rennie05,SrebroThesis}. Because users only rate a
  few items, one would like to infer their preference for unrated
  items.

  A special instance of this problem is the now famous Netflix
  problem~\cite{NetflixPrize}. Users (rows of the data matrix) are
  given the opportunity to rate movies (columns of the data matrix)
  but users typically rate only very few movies so that there are very
  few scattered observed entries of this data matrix. Yet one would
  like to complete this matrix so that the vendor (here Netflix) might
  recommend titles that any particular user is likely to be willing to
  order.  In this case, the data matrix of all user-ratings may be
  approximately low-rank because it is commonly believed that only a
  few factors contribute to an individual's tastes or preferences. 

\item {\em Triangulation from incomplete data.}  Suppose we are given
 partial information about the distances between objects and would
 like to reconstruct the low-dimensional geometry describing their
 locations. For example, we may have a network of low-power
 wirelessly networked sensors scattered randomly across a
 region. Suppose each sensor only has the ability to construct
 distance estimates based on signal strength readings from its
 nearest fellow sensors. From these noisy distance estimates, we can
 form a partially observed distance matrix. We can then estimate the
 true distance matrix whose rank will be equal to two if the sensors
 are located in a plane or three if they are located in three
 dimensional space~\cite{Linial95,So07}.  In this case, we
 only need to observe a few distances per node to have enough
 information to reconstruct the positions of the objects.
\end{itemize}
These examples are of course far from exhaustive and there are many
other problems which fall in this general category. For instance, we
may have some very limited information about a covariance matrix of
interest. Yet, this covariance matrix may be low-rank or approximately
low-rank because the variables only depend upon a comparably smaller
number of factors.

\subsection{Impediments and solutions}

Suppose for simplicity that we wish to recover a square $n \times n$
matrix $\mtx{M}$ of rank $r$.\footnote{We emphasize that there is
nothing special about $\mtx{M}$ being square and all of our
discussion would apply to arbitrary rectangular matrices as well.
The advantage of focusing on square matrices is a simplified
exposition and reduction in the number of parameters of which we
need to keep track.} Such a matrix $\mtx{M}$ can be represented by
$n^2$ numbers, but it only has $(2 n - r ) r$ degrees of freedom.
This fact can be revealed by counting parameters in the singular
value decomposition (the number of degrees of freedom associated
with the description of the singular values and of the left and
right singular vectors).  When the rank is small, this is
considerably smaller than $n^2$. For instance, when $\mtx{M}$
encodes a 10-dimensional phenomenon, then the number of degrees of
freedom is about $20\, n$ offering a reduction in dimensionality by
a factor about equal to $n/20$.  When $n$ is large (e.g.~in the
thousands or millions), the data matrix carries much less
information than its ambient dimension suggests. The problem is now
whether it is possible to recover this matrix from a sampling of its
entries without having to probe all the $n^2$ entries, or more
generally collect $n^2$ or more measurements about $\mtx{M}$.

\subsubsection{Which matrices?}

In general, one cannot hope to be able to recover a low-rank matrix
from a sample of its entries. Consider the rank-1 matrix $\mtx{M}$
equal to
\begin{equation}
  \label{eq:problem1}
  \mtx{M} = \vct{e}_1 \vct{e}_n^* = \begin{bmatrix}
  0 & 0 & \cdots & 0 & 1 \\
 0 & 0 & \cdots & 0 & 0 \\
 \vdots & \vdots &  \vdots & \vdots & \vdots \\
 0 & 0 & \cdots & 0 & 0
\end{bmatrix},
\end{equation}
where here and throughout, $\vct{e}_i$ is the $i$th canonical basis
vector in Euclidean space (the vector with all entries equal to 0
but the $i$th equal to 1). This matrix has a 1 in the top-right
corner and all the other entries are 0. Clearly this matrix cannot
be recovered from a sampling of its entries unless we pretty much
see all the entries. The reason is that for most sampling sets, we
would only get to see zeros so that we would have no way of guessing
that the matrix is not zero. For instance, if we were to see 90\% of
the entries selected at random, then 10\% of the time we would only
get to see zeroes.

It is therefore impossible to recover {\em all} low-rank matrices from
a set of sampled entries but can one recover {\em most} of them?
To investigate this issue, we introduce a simple model of low-rank
matrices. Consider the singular value decomposition (SVD) of a
matrix $\mtx{M}$
\begin{equation}
  \label{eq:svd}
  \mtx{M} = \sum_{k = 1}^r \sigma_k \vct{u}_k \vct{v}_k^*,
\end{equation}
where the $\vct{u}_k$'s and $\vct{v}_k$'s are the left and right
singular vectors, and the $\sigma_k$'s are the singular values (the
roots of the eigenvalues of $\mtx{M}^*\mtx{M}$). Then we could think
of a {\em generic} low-rank matrix as follows: the family
$\{\vct{u}_k\}_{1 \le k \le r}$ is selected uniformly at random among
all families of $r$ orthonormal vectors, and similarly for the the
family $\{\vct{v}_k\}_{1 \le k \le r}$. The two families may or may
not be independent of each other. We make no assumptions about the
singular values $\sigma_k$. In the sequel, we will refer to this model
as the {\em random orthogonal model}.  This model is convenient in the
sense that it is both very concrete and simple, and useful in the
sense that it will help us fix the main ideas. In the sequel, however,
we will consider far more general models. The question for now is
whether or not one can recover such a generic matrix from a sampling
of its entries.

\subsubsection{Which sampling sets?}

Clearly, one cannot hope to reconstruct any low-rank matrix
$\mtx{M}$---even of rank $1$---if the sampling set avoids any column
or row of $\mtx{M}$. Suppose that $\mtx{M}$ is of rank 1 and of the
form $\vct{x} \vct{y}^*$, $\vct{x}, \vct{y} \in \R^n$ so that the
$(i,j)$th entry is given by
\[
M_{ij} = x_i y_j.
\]
Then if we do not have samples from the first row for example, one
could never guess the value of the first component $x_1$, by any
method whatsoever; no information about $x_1$ is observed. There is of
course nothing special about the first row and this argument extends
to any row or column. To have any hope of recovering an unknown
matrix, one needs at least one observation per row and one observation
per column.

We have just seen that if the sampling is adversarial, e.g. one
observes all of the entries of $\mtx{M}$ but those in the first row,
then one would not even be able to recover matrices of rank $1$.  But
what happens for most sampling sets? Can one recover a low-rank matrix
from almost all sampling sets of cardinality $m$? Formally, suppose
that the set $\Omega$ of locations corresponding to the observed
entries ($(i,j) \in \Omega$ if $M_{ij}$ is observed) is a set of
cardinality $m$ sampled uniformly at random. Then can one recover a
generic low-rank matrix $M$, perhaps with very large probability, from
the knowledge of the value of its entries in the set $\Omega$?

\subsubsection{Which algorithm?}

If the number of measurements is sufficiently large, and if the
entries are sufficiently uniformly distributed as above, one might
hope that there is only {\em one} low-rank matrix with these
entries. If this were true, one would want to recover the data
matrix by solving the optimization problem
\begin{equation}
 \label{eq:l0}
  \begin{array}{ll}
\textrm{minimize}   & \quad \rank(\mtx{X})\\
\textrm{subject to} & \quad X_{ij} = M_{ij} \quad (i,j) \in \Omega,
 \end{array}
\end{equation}
where $\mtx{X}$ is the decision variable and $\rank(\mtx{X})$ is
equal to the rank of the matrix $\mtx{X}$.  The program
\eqref{eq:l0} is a common sense approach which simply seeks the
simplest explanation fitting the observed data. If there were only
one low-rank object fitting the data, this would recover $\mtx{M}$.
This is unfortunately of little practical use because this
optimization problem is not only NP-hard, but all known algorithms
which provide exact solutions require time doubly exponential in the
dimension $n$ of the matrix in both theory and practice \cite{Grigoriev84}.

If a matrix has rank $r$, then it has exactly $r$ nonzero singular
values so that the rank function in \eqref{eq:l0} is simply the
number of nonvanishing singular values. In this paper,
we consider an alternative which minimizes the sum of the singular
values over the constraint set. This sum is called the \emph{nuclear
  norm},
\begin{equation}
  \label{eq:nuclear}
  \|\mtx{X}\|_* = \sum_{k = 1}^n \sigma_k(\mtx{X})
\end{equation}
where, here and below, $\sigma_k(\mtx{X})$ denotes the $k$th largest
singular value of $\mtx{X}$.
The heuristic optimization is then given by
\begin{equation}
 \label{eq:sdp}
  \begin{array}{ll}
\textrm{minimize}   & \quad \|\mtx{X}\|_*\\
\textrm{subject to} & \quad   X_{ij} = M_{ij} \quad (i,j) \in \Omega.
 \end{array}
\end{equation}
Whereas the rank function counts the number of nonvanishing singular
values, the nuclear norm sums their amplitude and in some sense, is
to the rank functional what the convex $\ell_1$ norm is to the
counting $\ell_0$ norm in the area of sparse signal recovery.  The
main point here is that the nuclear norm is a convex function and,
as we will discuss in Section~\ref{sec:background} can be optimized
efficiently via semidefinite programming.

\subsubsection{A first typical result}

Our first result shows that, perhaps unexpectedly, this heuristic
optimization recovers a generic $\mtx{M}$ when the number of randomly
sampled entries is large enough.  We will prove the following:
\begin{theorem}
  \label{teo:main}
  Let $\mtx{M}$ be an $n_1 \times n_2$ matrix of rank $r$ sampled from
  the random orthogonal model, and put $n = \max(n_1, n_2)$. Suppose
  we observe $m$ entries of $\mtx{M}$ with locations sampled uniformly
  at random. Then there are numerical constants $C$ and $c$ such that if
\begin{equation}
\label{eq:main1} m \ge C \, n^{5/4} r \log n\, ,
\end{equation}
the minimizer to the problem \eqref{eq:sdp} is unique and equal to
$\mtx{M}$ with probability at least $1 - c n^{-3}$; that is to say,
the semidefinite program $\eqref{eq:sdp}$ recovers all the entries of
$\mtx{M}$ with no error.  In addition, if $r \leq n^{1/5}$, then the
recovery is exact with probability at least $1 - c n^{-3}$ provided
that
\begin{equation}
\label{eq:main2} m \ge C \,  n^{6/5}r \log n\,.
\end{equation}
\end{theorem}
The theorem states that a surprisingly small number of entries are
sufficient to complete a generic low-rank matrix. For small values of
the rank, e.g. when $r = O(1)$ or $r = O(\log n)$, one only needs to
see on the order of $n^{6/5}$ entries (ignoring logarithmic factors)
which is considerably smaller than $n^2$---the total number of entries
of a squared matrix. The real feat, however, is that the recovery
algorithm is tractable and very concrete. Hence the contribution is
twofold:
\begin{itemize}
\item Under the hypotheses of Theorem \ref{teo:main}, there is a
  unique low-rank matrix which is consistent with the observed
  entries.
\item Further, this matrix can be recovered by the convex
optimization  \eqref{eq:sdp}. In other words, for most problems, the
  nuclear norm relaxation is {\em formally equivalent} to the
  combinatorially hard rank minimization problem \eqref{eq:l0}.
\end{itemize}

Theorem \ref{teo:main} is in fact a special instance of a far more
general theorem that covers a much larger set of matrices $\mtx{M}$.
We describe this general class of matrices and precise recovery
conditions in the next section.

\subsection{Main results}
\label{sec:main}

As seen in our first example \eqref{eq:problem1}, it is impossible
to recover a matrix which is equal to zero in nearly all of its
entries unless we see all the entries of the matrix. To recover a
low-rank matrix, this matrix cannot be in the null space of the
sampling operator giving the values of a subset of the entries. Now
it is easy to see that if the singular vectors of a matrix $\mtx{M}$
are highly concentrated, then $\mtx{M}$ could very well be in the
null-space of the sampling operator. For instance consider the
rank-2 symmetric matrix $\mtx{M}$ given by
\[
\mtx{M} = \sum_{k = 1}^2 \sigma_k \vct{u}_k \vct{u}_k^*, \quad
\begin{array}{ll}
  \vct{u}_1 & = (\vct{e}_1 + \vct{e}_2)/\sqrt{2},\\
  \vct{u}_2 & = (\vct{e}_1 - \vct{e}_2)/\sqrt{2},
 \end{array}
\]
where the singular values are arbitrary. Then this matrix vanishes
everywhere except in the top-left $2 \times 2$ corner and one would
basically need to see all the entries of $\mtx{M}$ to be able to
recover this matrix exactly by any method whatsoever. There is an
endless list of examples of this sort.  Hence, we arrive at the
notion that, somehow, the singular vectors need to be sufficiently
spread---that is, uncorrelated with the standard basis---in order to
minimize the number of observations needed to recover a low-rank
matrix.\footnote{Both the left and right singular vectors need to be
  uncorrelated with the standard basis. Indeed, the matrix $\vct{e}_1
  \vct{v}^*$ has its first row equal to $\vct{v}$ and all the others
  equal to zero. Clearly, this rank-1 matrix cannot be recovered unless
  we basically see all of its entries.} This motivates the following
definition.
\begin{definition}
\label{def:coherence} Let $U$ be a subspace of $\mathbb{R}^n$ of
dimension $r$ and $\mtx{P}_U$ be the orthogonal projection onto $U$.
Then the \emph{coherence} of $U$ (vis-\`a-vis the standard basis
$(\vct{e}_i)$) is defined to be
\begin{equation}
\label{eq:coherence} \mu(U) \equiv \frac{n}{r} \max_{1 \le i \le n}
\|\mtx{P}_U
  \vct{e}_i\|^2.
\end{equation}
\end{definition}
Note that for any subspace, the smallest $\mu(U)$ can be is $1$,
achieved, for example, if $U$ is spanned by vectors whose entries all
have magnitude $1/\sqrt{n}$. The largest possible value for $\mu(U)$
is $n/r$ which would correspond to any subspace that contains a
standard basis element.  We shall be primarily interested in subspace
with low coherence as matrices whose column and row spaces have low
coherence cannot really be in the null space of the sampling
operator. For instance, we will see that the random subspaces
discussed above have nearly minimal coherence.

To state our main result, we introduce two assumptions about an
$n_1\times n_2$ matrix $\mtx{M}$ whose SVD is given by $\mtx{M} =
\sum_{1 \le k \le r} \sigma_k \vct{u}_k \vct{v}_k^*$ and with column
and row spaces denoted by $U$ and $V$ respectively.
\begin{description}
\item[{A0}] The coherences obey $\max(\mu(U), \mu(V)) \le \mu_0$ for
  some positive $\mu_0$.
\item[{A1}] The $n_1 \times n_2$ matrix $\sum_{1 \le k \le r}
  \vct{u}_k \vct{v}_k^*$ has a maximum entry bounded by $\mu_1
  \sqrt{r/(n_1 n_2)}$ in absolute value for some positive $\mu_1$.
\end{description}
The $\mu$'s above may depend on $r$ and $n_1, n_2$. Moreover, note
that {\bf A1} always holds with $\mu_1 = \mu_0 \, \sqrt{r}$ since
the $(i,j)$th entry of the matrix $\sum_{1\le k \le r} \vct{u}_k
\vct{v}_k^*$ is given by $\sum_{1\le k \le r} u_{ik} v_{jk}$ and by
the Cauchy-Schwarz inequality,
\[
\left|\sum_{1\le k \le r} u_{ik} v_{jk}\right| \le \sqrt{\sum_{1\le k
    \le r} |u_{ik}|^2}\, \sqrt{\sum_{1\le k \le r} |v_{jk}|^2} \le
\frac{\mu_0 r}{\sqrt{n_1 n_2}}.
\]
Hence, for sufficiently small ranks, $\mu_1$ is comparable to $\mu_0$.
As we will see in Section \ref{sec:whichlow}, for larger ranks, both
subspaces selected from the uniform distribution and spaces
constructed as the span of singular vectors with bounded entries are
not only incoherent with the standard basis, but also obey {\bf A1}
with high probability for values of $\mu_1$ at most logarithmic in
$n_1$ and/or $n_2$. Below we will assume that $\mu_1$ is greater than
or equal to 1.

We are in the position to state our main result: if a matrix has row
and column spaces that are incoherent with the standard basis, then
nuclear norm minimization can recover this matrix from a random
sampling of a small number of entries.
\begin{theorem}
  \label{teo:main2} Let $\mtx{M}$ be an $n_1 \times n_2$ matrix of
  rank $r$ obeying {\bf A0} and {\bf A1} and put $n = \max(n_1,n_2)$.
  Suppose we observe $m$ entries of $\mtx{M}$ with locations sampled
  uniformly at random. Then there exist constants $C$, $c$ such that
  if
\begin{equation}\label{eq:main3}
  m \ge  \, C
  \max(\mu_1^2, \mu_0^{1/2} \mu_1, \mu_0 n^{1/4}) \, n r (\beta \log n)
\end{equation}
for some $\beta>2$, then the minimizer to the problem \eqref{eq:sdp}
is unique and equal to $\mtx{M}$ with probability at least $1 - c
n^{-\beta}$.  For $r \le \mu_0^{-1} n^{1/5}$ this estimate can be
improved to
\begin{equation}
  \label{eq:main4} m \ge   C \, \mu_0
  \, n^{6/5} r (\beta \log n)
\end{equation}
with the same probability of success.
\end{theorem}
Theorem \ref{teo:main2} asserts that if the coherence is low, few
samples are required to recover $\mtx{M}$. For example, if $\mu_0 =
O(1)$ and the rank is not too large, then the recovery is exact with
large probability provided that
\begin{equation}
\label{eq:incoherent-teo} m \ge C \, n^{6/5} r \log n\,.
\end{equation}
We give two illustrative examples of matrices with incoherent column
and row spaces. This list is by no means exhaustive.
\begin{enumerate}
\item The first example is the random orthogonal model. For values of
  the rank $r$ greater than $\log n$, $\mu(U)$ and $\mu(V)$ are
  $O(1)$, $\mu_1 = O(\log n)$ both with very large probability.
  Hence, the recovery is exact provided that $m$ obeys
  \eqref{eq:main1} or \eqref{eq:main2}. Specializing Theorem
  \ref{teo:main2} to these values of the parameters gives Theorem
  \ref{teo:main}. Hence, Theorem \ref{teo:main} is a special case of
  our general recovery result.

\item The second example is more general and, in a nutshell, simply
  requires that the components of the singular vectors of $\mtx{M}$
  are small.  Assume that the $\vct{u}_j$ and $\vct{v}_j$'s obey
\begin{equation}
  \label{eq:muB}
  \max_{ij} |\<\vct{e}_i, \vct{u}_j\>|^2 \le \mu_B/n, \quad  \max_{ij} |\<\vct{e}_i, \vct{v}_j\>|^2
  \le \mu_B/n,
\end{equation}
for some value of $\mu_B=O(1)$. Then the maximum coherence is at
most $\mu_B$ since $\mu(U) \le \mu_B$ and $\mu(V) \le \mu_B$.
Further, we will see in Section \ref{sec:whichlow} that ${\bf A1}$
holds most of the time with $\mu_1 = O(\sqrt{\log n})$. Thus, for
matrices with singular vectors obeying \eqref{eq:muB}, the recovery
is exact provided that $m$ obeys \eqref{eq:incoherent-teo} for
values of the rank not exceeding $\mu_B^{-1} n^{1/5}$.
\end{enumerate}

\subsection{Extensions}
\label{sec:extensions}

Our main result (Theorem~\ref{teo:main2}) extends to a variety of
other low-rank matrix completion problems beyond the sampling of
entries. Indeed, suppose we have two orthonormal bases
$\vct{f}_1,\ldots, \vct{f}_{n}$ and $\vct{g}_1,\ldots, \vct{g}_{n}$
of $\R^n$, and that we are interested in solving the rank
minimization problem
\begin{equation}\label{eq:rank-min-prob-general}
  \begin{array}{ll}
\textrm{minimize}   & \quad \textrm{rank}(\mtx{X})\\
\textrm{subject to} & \quad \vct{f}_i^* \mtx{X} \vct{g}_j =
\vct{f}_i^* \vct{M} \vct{g}_j, \quad (i,j) \in \Omega,
 \end{array}\,.
\end{equation}
This comes up in a number of applications. As a motivating example,
there has been a great deal of interest in the machine learning
community in developing specialized algorithms for the
\emph{multiclass} and \emph{multitask} learning problems (see, e.g.,
\cite{Abernethy06, Argyriou07, Amit07}).  In multiclass learning,
the goal is to build multiple classifiers with the same training
data to distinguish between more than two categories. For example,
in face recognition, one might want to classify whether an image
patch corresponds to an eye, nose, or mouth.  In multitask learning,
we have a large set of data, but have a variety of different
classification tasks, and, for each task, only partial subsets of
the data are relevant.  For instance, in activity recognition, we
may have acquired sets of observations of multiple subjects and want
to determine if each observed person is walking or running. However,
a different classifier is to be learned for each individual, and it
is not clear how having access to the full collection of
observations can improve classification performance. Multitask
learning aims precisely to take advantage of the access to the full
database to improve performance on the individual tasks.

In the abstract formulation of this problem for linear classifiers,
we have $K$ classes to distinguish and are given training examples
$\vct{f}_1, \ldots, \vct{f}_n$.  For each example, we are given
partial labeling information about which classes it belongs or does
not belong to. That is, for each example $\vct{f}_j$ and class $k$,
we may either be told that $\vct{f}_j$ belongs to class $k$, be told
$\vct{f}_j$ does not belong to class $k$, or provided no information
about the membership of $\vct{f}_j$ to class $k$.  For each class $1
\le k \le K$, we would like to produce a linear function $\vct{w}_k$
such that $\vct{w}_k^* \vct{f}_i
> 0$ if $\vct{f}_i$ belongs to class $k$ and $\vct{w}_k^* \vct{f}_i
< 0$ otherwise.  Formally, we can search for the vector $\vct{w}_k$
that satisfies the equality constraints $\vct{w}_k^* \vct{f}_i =
y_{ik}$ where $y_{ik}=1$ if we are told that $\vct{f}_i$ belongs to
class $k$, $y_{ik}=-1$ if we are told that $\vct{f}_i$ does not
belong to class $k$, and $y_{ik}$ unconstrained if we are not
provided information. A common hypothesis in the multitask setting
is that the $\vct{w}_k$ corresponding to each of the classes
together span a very low dimensional subspace with dimension
significantly smaller than $K$~\cite{Abernethy06, Argyriou07,
Amit07}.  That is, the basic assumption is that
\[
    \mtx{W} = [\vct{w}_1, \ldots, \vct{w}_K]
\]
is low-rank. Hence, the multiclass learning problem can be cast as
\eqref{eq:rank-min-prob-general} with observations of the form
$\vct{f}_i^* \mtx{W} \vct{e}_j$.

To see that our theorem provides conditions under which
\eqref{eq:rank-min-prob-general} can be solved via nuclear norm
minimization, note that there exist unitary transformations
$\mtx{F}$ and $\mtx{G}$ such that $\vct{e}_j = \mtx{F} \vct{f}_j$
and $\vct{e}_j = \mtx{G} \vct{g}_j$ for each $j = 1,\ldots, n$.
Hence,
\begin{equation*}
  \vct{f}_i^* \mtx{X} \vct{g}_j = \vct{e}_i^* (\mtx{F} \mtx{X} \mtx{G}^*) \vct{e}_j.
\end{equation*}
Then if the conditions of Theorem~\ref{teo:main2} hold for the
matrix $\mtx{F} \mtx{X} \mtx{G}^*$, it is immediate that nuclear
norm minimization finds the unique optimal solution of
\eqref{eq:rank-min-prob-general} when we are provided a large enough
random collection of the inner products $\vct{f}_i^* \mtx{M}
\vct{g}_j$. In other words, all that is needed is that the column
and row spaces of $\mtx{M}$ be respectively incoherent with the
basis $(\vct{f}_i)$ and $(\vct{g}_i)$.

From this perspective, we additionally remark that our results likely
extend to the case where one observes a small number of arbitrary
linear functionals of a hidden matrix $\mtx{M}$. Set $N= n^2$ and
$\mtx{A}_1,\ldots, \mtx{A}_{N}$ be an orthonormal basis for the linear
space of $n \times n$ matrices with the usual inner product
$\<\mtx{X}, \mtx{Y}\> = \trace(\mtx{X}^* \mtx{Y})$.  Then we expect
our results should also apply to the rank minimization problem
\begin{equation}
\label{eq:generalrank}
  \begin{array}{ll}
    \textrm{minimize}   & \quad \textrm{rank}(\mtx{X})\\
    \textrm{subject to} & \quad \<\mtx{A}_k, \mtx{X}\> =
\<\mtx{A}_k, \mtx{M}\> \quad k \in \Omega,
 \end{array}
\end{equation}
where $\Omega \subset\{1,\ldots, N\}$ is selected uniformly at
random. 
In fact, \eqref{eq:generalrank} is \eqref{eq:l0} when the orthobasis
is the canonical basis $(\vct{e}_i \vct{e}_j^*)_{1 \le i, j \le
  n}$. Here, those low-rank matrices which have small inner product
with all the basis elements $\mtx{A}_k$ may be recoverable by nuclear
norm minimization. To avoid unnecessary confusion and notational
clutter, we leave this general low-rank recovery problem for future
work.

\subsection{Connections, alternatives and prior art}
\label{sec:background}

Nuclear norm minimization is a recent heuristic introduced by Fazel
in~\cite{FazelThesis}, and is an extension of the trace heuristic
often used by the control community, see
e.g.~\cite{Beck98,Mesbahi97}. Indeed, when the matrix variable is
symmetric and positive semidefinite, the nuclear norm of $\mtx{X}$
is the sum of the (nonnegative) eigenvalues and thus equal to the
trace of $\mtx{X}$. Hence, for positive semidefinite unknowns,
\eqref{eq:sdp} would simply minimize the trace over the constraint
set:
\begin{equation*}
  \begin{array}{ll}
\textrm{minimize}   & \quad  \trace(\mtx{X})\\
\textrm{subject to} & \quad   X_{ij} = M_{ij} \quad (i,j) \in \Omega\\
& \quad \mtx{X} \succeq 0
 \end{array}.
\end{equation*}
This is a semidefinite program.  Even for the general matrix $\mtx{M}$
which may not be positive definite or even symmetric, the nuclear norm
heuristic can be formulated in terms of semidefinite programming as,
for instance, the program \eqref{eq:sdp} is equivalent to
\begin{equation*}
  \begin{array}{ll}
\textrm{minimize}   & \quad  \trace(\mtx{W}_1) + \trace(\mtx{W}_2)\\
\textrm{subject to} & \quad   X_{ij} = M_{ij} \quad (i,j) \in \Omega\\
& \quad \begin{bmatrix} \mtx{W}_1 & \mtx{X}\\
\mtx{X}^* & \mtx{W}_2
\end{bmatrix} \succeq 0
 \end{array}
\end{equation*}
with optimization variables $\mtx{X}$, $\mtx{W}_1$ and $\mtx{W}_2$,
(see, e.g.,~\cite{FazelThesis,Vandenberghe96}). There are many
efficient algorithms and high-quality software available for solving
these types of problems.

Our work is inspired by results in the emerging field of {\em
compressive sampling} or {\em compressed sensing}, a new paradigm
for acquiring information about objects of interest from what
appears to be a highly incomplete set of
measurements~\cite{CRT06,OptimalRecovery,DonohoCS}. In practice,
this means for example that high-resolution imaging is possible with
fewer sensors, or that one can speed up signal acquisition time in
biomedical applications by orders of magnitude, simply by taking far
fewer specially coded samples. Mathematically speaking, we wish to
reconstruct a signal $\vct{x}\in \R^n$ from a small number
measurements $\vct{y} = \Phi \vct{x}$, $\vct{y} \in \R^m$, and $m$
is much smaller than $n$; i.e.~we have far fewer equations than
unknowns.  In general, one cannot hope to reconstruct $\vct{x}$ but
assume now that the  object we wish to recover is known to be
structured in the sense that it is sparse (or approximately sparse).
This means that the unknown object depends upon a smaller number of
unknown parameters. Then it has been shown that $\ell_1$
minimization allows recovery of sparse signals from remarkably few
measurements: supposing $\Phi$ is chosen randomly from a suitable
distribution, then with very high probability, all sparse signals
with about $k$ nonzero entries can be recovered from on the order of
$k \log n$ measurements. For instance, if $\vct{x}$ is $k$-sparse in
the Fourier domain, i.e.~$\vct{x}$ is a superposition of $k$
sinusoids, then it can be perfectly recovered with high
probability---by $\ell_1$ minimization---from the knowledge of about
$k \log n$ of its entries sampled uniformly at random \cite{CRT06}.

From this viewpoint, the results in this paper greatly extend the
theory of compressed sensing by showing that other types of
interesting objects or structures, beyond sparse signals and images,
can be recovered from a limited set of measurements. Moreover, the
techniques for proving our main results build upon ideas from the
compressed sensing literature together with probabilistic tools such
as the powerful techniques of Bourgain and of Rudelson for bounding
norms of operators between Banach spaces.

Our notion of incoherence generalizes the concept of the same name
in compressive sampling.  Notably, in~\cite{CR07}, the authors
introduce the notion of the incoherence of a unitary transformation.
Letting $\mtx{U}$ be an $n \times n$ unitary matrix, the
\emph{coherence} of $\mtx{U}$ is given by
\[
\mu(\mtx{U}) = n\max_{j,k} |U_{jk}|^2.
\]
This quantity ranges in values from $1$ for a unitary transformation
whose entries all have the same magnitude to $n$ for the identity
matrix.  Using this notion, \cite{CR07} showed that with high
probability, a $k$-sparse signal could be recovered via linear
programming from the observation of the inner product of the signal
with $m = \Omega(\mu(\mtx{U}) k\, \log n)$ randomly selected columns
of the matrix $\mtx{U}$. This result provided a generalization of
the celebrated results about partial Fourier observations described
in~\cite{CRT06}, a special case where $\mu(\mtx{U})=1$. 
This paper generalizes the notion of incoherence to problems beyond
the setting of sparse signal recovery.

In~\cite{Recht07}, the authors studied the nuclear norm heuristic
applied to a related problem where partial information about a
matrix $\mtx{M}$ is available from $m$  equations of the form
\begin{equation}
  \label{eq:linearfunctional}
  \<\mtx{A}^{(k)}, M\> =  \sum_{ij} A^{(k)}_{ij} M_{ij} = b_k,  \qquad k = 1, \ldots, m,
\end{equation}
where for each $k$, $\{A^{(k)}_{ij}\}_{ij}$ is an i.i.d.~sequence of
Gaussian or Bernoulli random variables and the sequences
$\{\mtx{A}^{(k)}\}$ are also independent from each other (the
sequences $\{\mtx{A}^{(k)}\}$ and $\{b_k\}$ are available to the
analyst). Building on the concept of {\em restricted isometry}
introduced in \cite{Candes05} in the context of sparse signal
recovery, \cite{Recht07} establishes the first sufficient conditions
for which the nuclear norm heuristic returns the minimum rank element
in the constraint set.  They prove that the heuristic succeeds with
large probability whenever the number $m$ of available measurements is
greater than a constant times $2 n r \log n$ for $n \times n$
matrices. Although this is an interesting result, a serious impediment
to this approach is that one needs to essentially measure random
projections of the unknown data matrix---a situation which
unfortunately does not commonly arise in practice.  Further, the
measurements in \eqref{eq:linearfunctional} give some information
about {\em all} the entries of $\mtx{M}$ whereas in our problem,
information about most of the entries is simply not available. In
particular, the results and techniques introduced in~\cite{Recht07} do
not begin to address the matrix completion problem of interest to us
in this paper. As a consequence, our methods are completely different;
for example, they do not rely on any notions of restricted
isometry. Instead, as we discuss below, we prove the existence of a
Lagrange multiplier for the optimization (\ref{eq:sdp}) that certifies
the unique optimal solution is precisely the matrix that we wish to
recover.

Finally, we would like to briefly discuss the possibility of other
recovery algorithms when the sampling happens to be chosen in a very
special fashion.  For example, suppose that $\mtx{M}$ is generic and
that we precisely observe every entry in the first $r$ rows and
columns of the matrix.  Write $\mtx{M}$ in block form as
\begin{equation*}
   \mtx{M} = \left[\begin{array}{cc} \mtx{M}_{11} &  \mtx{M}_{12}\\ \mtx{M}_{21} & \mtx{M}_{22}
\end{array}\right]
\end{equation*}
with $\mtx{M}_{11}$ an $r\times r$ matrix. In the special case that
$\mtx{M}_{11}$ is invertible and $\mtx{M}$ has rank $r$, then it is
easy to verify that $\mtx{M}_{22} =
\mtx{M}_{21}\mtx{M}_{11}^{-1}\mtx{M}_{12}$. One can prove this
identity by forming the SVD of $\mtx{M}$, for example. That is, if
$\mtx{M}$ is generic, and the upper $r\times r$ block is invertible,
and we observe \emph{every} entry in the first $r$ rows and columns,
we can recover $\mtx{M}$.  This result immediately generalizes to
the case where one observes precisely $r$ rows and $r$ columns and
the $r \times r$ matrix at the intersection of the observed rows and
columns is invertible. However, this scheme has many practical
drawbacks that stand in the way of a generalization to a completion
algorithm from a general set of entries.  First, if we miss
\emph{any} entry in these rows or columns, we cannot recover
$\mtx{M}$, nor can we leverage any information provided by entries
of $\mtx{M}_{22}$. Second, if the matrix has rank less than $r$, and
we observe $r$ rows and columns, a combinatorial search to find the
collection that has an invertible square sub-block is required.
Moreover, because of the matrix inversion, the algorithm is rather
fragile to noise in the entries.

\subsection{Notations and organization of the paper}
\label{sec:notations}

The paper is organized as follows. We first argue in Section
\ref{sec:whichlow} that the random orthogonal model and, more
generally, matrices with incoherent column and row spaces obey the
assumptions of the general Theorem \ref{teo:main2}.  To prove Theorem
\ref{teo:main2}, we first establish sufficient conditions which
guarantee that the true low-rank matrix $\mtx{M}$ is the unique
solution to \eqref{eq:sdp} in Section \ref{sec:duality}. One of these
conditions is the existence of a dual vector obeying two crucial
properties. Section \ref{sec:architecture} constructs such a dual
vector and provides the overall architecture of the proof which shows
that, indeed, this vector obeys the desired properties provided that
the number of measurements is sufficiently large.  Surprisingly, as
explored in Section \ref{sec:coupon}, the existence of a dual vector
certifying that $\mtx{M}$ is unique is related to some problems in
random graph theory including ``the coupon collector's problem.''
Following this discussion, we prove our main result via several
intermediate results which are all proven in Section \ref{sec:proofs}.
Section \ref{sec:numerical} introduces numerical experiments showing
that matrix completion based on nuclear norm minimization works well
in practice.  Section \ref{sec:discussion} closes the paper with a
short summary of our findings, a discussion of important extensions
and improvements. In particular, we will discuss possible ways of
improving the 1.2 exponent in \eqref{eq:main4} so that it gets closer
to 1. Finally, the Appendix provides proofs of auxiliary lemmas
supporting our main argument.

Before continuing, we provide here a brief summary of the notations
used throughout the paper. Matrices are bold capital, vectors are bold
lowercase and scalars or entries are not bold. For instance, $\mtx{X}$
is a matrix and $X_{ij}$ its $(i,j)$th entry. Likewise $\vct{x}$ is a
vector and $x_i$ its $i$th component. When we have a collection of
vectors $\vct{u}_k\in \R^n$ for $1\leq k \leq d$, we will denote by
$u_{ik}$ the $i$th component of the vector $\vct{u}_k$ and
$[\vct{u}_1,\ldots,\vct{u}_d]$ will denote the $n\times d$ matrix
whose $k$th column is $\vct{u}_k$.

A variety of norms on matrices will be discussed. The spectral norm of
a matrix is denoted by $\|\mtx{X}\|$.  The Euclidean inner product
between two matrices is $\<\mtx{X}, \mtx{Y}\> = \trace(\mtx{X}^*
\mtx{Y})$, and the corresponding Euclidean norm, called the Frobenius
or Hilbert-Schmidt norm, is denoted $\|\mtx{X}\|_F$.  That is,
$\|\mtx{X}\|_F=\<\mtx{X},\mtx{X}\>^{1/2}$.  The nuclear norm of a
matrix $\mtx{X}$ is $\|\mtx{X}\|_*$. The maximum entry of $\mtx{X}$
(in absolute value) is denoted by $\|\mtx{X}\|_\infty \equiv \max_{ij}
|X_{ij}|$. For vectors, we will only consider the usual Euclidean
$\ell_2$ norm which we simply write as $\|\vct{x}\|$.

Further, we will also manipulate linear transformation which acts on
matrices and will use caligraphic letters for these operators as in
${\cal A}(\mtx{X})$.  In particular, the identity operator will be
denoted by $\OpId$. The only norm we will consider for these
operators is their spectral norm (the top singular value) denoted by
$\|{\cal A}\| = \sup_{\mtx{X} :
  \|\mtx{X}\|_F \le 1} \, \|{\cal A}(\mtx{X})\|_F$.

Finally, we adopt the convention that $C$ denotes a numerical
constant independent of the matrix dimensions, rank, and number of
measurements, whose value may change from line to line.  Certain
special constants with precise numerical values will be ornamented
with subscripts (e.g., $C_R$). Any exceptions to this notational
scheme will be noted in the text.

\section{Which matrices are incoherent?}
\label{sec:whichlow}

In this section we restrict our attention to square $n \times n$
matrices, but the extension to rectangular $n_1\times n_2$ matrices
immediately follows by setting $n=\max(n_1,n_2)$.

\subsection{Incoherent bases span incoherent subspaces}
Almost all $n\times n$ matrices $\mtx{M}$ with singular vectors
$\{\vct{u}_k\}_{1 \le k \le r}$ and $\{\vct{v}_k\}_{1 \le k \le r}$
obeying the size property \eqref{eq:muB} also satisfy the
assumptions {\bf A0} and {\bf A1} with $\mu_0 = \mu_B$, $\mu_1 = C
\mu_B \sqrt{\log n}$ for some positive constant $C$.  As mentioned
above, {\bf A0} holds automatically, but, observe that {\bf A1}
would not hold with a small value of $\mu_1$ if two rows of the
matrices $[\vct{u}_1, \ldots, \vct{u}_r]$ and $[\vct{v}_1, \ldots,
\vct{v}_r]$ are identical with all entries of magnitude
$\sqrt{\mu_B/n}$ since it is not hard to see that in this case
\[
\|\sum_k \vct{u}_k \vct{v}_k^*\|_\infty = \mu_B \, r/n.
\]
Certainly, this example is constructed in a very special way, and
should occur infrequently. We now show that it is generically
unlikely.

Consider the matrix
\begin{equation}\label{eq:example-inc-svd}
    \sum_{k=1}^r \epsilon_k \vct{u}_k \vct{v}_k^*,
\end{equation}
where $\{\epsilon_k\}_{1 \le k \le r}$ is an arbitrary sign
sequence. For almost all choices of sign sequences, {\bf A1} is
satisfied with $\mu_1 = O(\mu_B \sqrt{\log n})$. Indeed, if one
selects the signs uniformly at random, then for each $\beta > 0$,
\begin{equation}
  \label{eq:incoherent1}
  \P(\|  \sum_{k=1}^r \epsilon_k \vct{u}_k \vct{v}_k\|_\infty \ge \mu_B \, \sqrt{8\beta r \log n}/n) \le (2n^2)  \, n^{-\beta}. 
\end{equation}
This is of interest because suppose the low-rank matrix we wish to
recover is of the form
\begin{equation}\label{eq:example-inc-svd2}
   \mtx{M} = \sum_{k=1}^r \lambda_k \vct{u}_k \vct{v}_k^*
\end{equation}
with scalars $\lambda_k$.  Since the vectors $\{\vct{u}_k\}$ and
$\{\vct{v}_k\}$ are orthogonal, the singular values of $\mtx{M}$ are
given by $|\lambda_k|$ and the singular vectors are given by
$\sgn(\lambda_k) \vct{u}_k$ and $\vct{v}_k$ for $k = 1,\ldots, r$.
Hence, in this model {\bf A1} concerns the maximum entry of the matrix
given by \eqref{eq:example-inc-svd} with
$\epsilon_k=\sgn(\lambda_k)$. That is to say, for most sign patterns,
the matrix of interest obeys an appropriate size condition.  We
emphasize here that the only thing that we assumed about the
$\vct{u}_k$'s and $\vct{v}_k$'s was that they had small entries. In
particular, they could be equal to each other as would be the case for
a symmetric matrix.

The claim \eqref{eq:incoherent1} is a simple application of Hoeffding's
inequality. The $(i,j)$th entry of \eqref{eq:example-inc-svd} is given
by
\[
Z_{ij} = \sum_{1 \le k \le r} \epsilon_k u_{ik} v_{jk},
\]
and is a sum of $r$ zero-mean independent random variables, each
bounded by $\mu_B/n$. Therefore,
\[
\P(|Z_{ij}| \geq \lambda \mu_B\sqrt{r}/n) \leq 2e^{-\lambda^2/8}.
\]
Setting $\lambda$ proportional to $\sqrt{\log n}$ and applying the
union bound gives the claim.

To summarize, we say that $\mtx{M}$ is sampled from the {\em
  incoherent basis model} if it is of the form
\begin{equation}
\label{eq:incoherentmodel}
    \mtx{M} = \sum_{k=1}^r \epsilon_k \sigma_k \vct{u}_k \vct{v}_k^*;
  \end{equation}
  $\{\epsilon_k\}_{1 \le k \le r}$ is a random sign sequence, and
  $\{\vct{u}_k\}_{1 \le k \le r}$ and $\{\vct{v}_k\}_{1 \le k \le r}$
  have maximum entries of size at most $\sqrt{\mu_B/n}$.
  \begin{lemma}
    \label{teo:incoherentbases} There exist numerical constants $c$ and $C$
    such that for any $\beta > 0$, matrices from the incoherent basis
    model obey the assumption {\bf A1} with $\mu_1 \leq C\mu_B
    \sqrt{(\beta + 2) \log n}$ with probability at least $1-c n^{-\beta}$.
\end{lemma}

\subsection{Random subspaces span incoherent subspaces}

In this section, we prove that the random orthogonal model obeys the
two assumptions {\bf A0} and {\bf A1} (with appropriate values for the
$\mu$'s) with large probability.
\begin{lemma}
  \label{teo:gaussian} Set $\bar r = \max(r,\log n)$. Then there exist
  constants $C$ and $c$ such that the random orthogonal model
  obeys:\footnote{When $r \ge C' (\log n)^3$ for some positive constant
    $C'$, a better estimate is possible, namely, $\|\mtx{\sum_{1 \le k
        \le r} \vct{u}_k \vct{v}_k^*}\|_\infty \le C\, \sqrt{r\log
      n}/n$.}
  \begin{enumerate}
  \item $\max_i \|\mtx{P}_U \vct{e}_i\|^2 \le C \, \bar r/n$,
  \item $\|\sum_{1 \le k \le r} \vct{u}_k \vct{v}_k^*\|_\infty \le C
    \, \log n \, \sqrt{\bar r}/n$.
\end{enumerate}
  with probability $1- c n^{-3} \log n$.
\end{lemma}
We note that an argument similar to the following proof would give
that if $C$ of the form $K \beta$ where $K$ is a fixed numerical
constant, we can achieve a probability at least $1 - c n^{-\beta}$
provided that $n$ is sufficiently large. To establish these facts,
we make use of the standard result below \cite{LaurentMassart}.
\begin{lemma}
  Let $Y_d$ be distributed as a chi-squared random
  variable with $d$ degrees of freedom. Then for each $t > 0$
  \begin{equation}
    \label{eq:chi2}
    \P(Y_d - d \ge t \, \sqrt{2d}  + t^2)
\le e^{-t^2/2} \quad
\text{ and } \quad  \P(Y_d - d \le -t \,
\sqrt{2d}) \le e^{-t^2/2}.
  \end{equation}
\end{lemma}
We will use \eqref{eq:chi2} as follows: for each $\epsilon \in (0,1)$
we have
\begin{equation}
  \label{eq:chi2b}
  \P(Y_d \ge d \, (1-\epsilon)^{-1}) \le e^{-\epsilon^2 d/4}  \quad
  \text{ and } \quad  \P(Y_d \le d \, (1-\epsilon)) \le e^{-\epsilon^2 d/4}.
\end{equation}

We begin with the second assertion of Lemma \ref{teo:gaussian} since
it will imply the first as well. Observe that it follows from
\begin{equation}
  \label{eq:PU}
  \|\mtx{P}_U \vct{e}_i\|^2 = \sum_{1 \le k \le r} u_{ik}^2,
\end{equation}
that $Z_r \equiv \|\mtx{P}_U \vct{e}_i\|^2$ ($a$ is fixed) is the
squared Euclidean length of the first $r$ components of a unit vector
uniformly distributed on the unit sphere in $n$ dimensions.  Now
suppose that $x_1, x_2, \ldots, x_n$ are i.i.d.~$N(0,1)$. Then the
distribution of a unit vector uniformly distributed on the sphere is
that of $\vct{x}/\|\vct{x}\|$ and, therefore, the law of $Z_r$ is that
of $Y_r/Y_n$, where $Y_r = \sum_{k \le r} x_k^2$. Fix $\epsilon > 0$
and consider the event $A_{n,\epsilon} = \{Y_n/n \ge
1-\epsilon\}$. For each $\lambda > 0$, it follows from
\eqref{eq:chi2b} that
\begin{align*}
  \P(Z_r - r/n \ge \lambda \sqrt{2r}/n) & = \P(Y_r \ge
  [r+\lambda\sqrt{2r}] Y_n/n)\\
  & \le \P(Y_r \ge
  [r+\lambda\sqrt{2r}] Y_n/n \text{ and } A_{n,\epsilon}) + \P(A_{n,\epsilon}^c)\\
  & \le \P(Y_r \ge [r+\lambda\sqrt{2r}][1-\epsilon]) + e^{-\epsilon^2 n/4}\\
  & = \P(Y_r - r \ge \lambda\sqrt{2r}[1- \epsilon - \epsilon
  \sqrt{r/2\lambda^2}]) + e^{-\epsilon^2 n/4}.
\end{align*}
Now pick $\epsilon = 4 (n^{-1} \log n)^{1/2}$, $\lambda = 8\sqrt{2\log
  n}$ and assume that $n$ is sufficiently large so that
\[
\epsilon(1+ \sqrt{r/2\lambda^2}) \le 1/2.
\]
Then
\[
\P(Z_r - r/n \ge \lambda \sqrt{2r}/n) \le \P(Y_r - r \ge (\lambda/2)
\sqrt{2r}) + n^{-4}.
\]
Assume now that $r \ge 4\log n$ (which means that $\lambda \le 4
\sqrt{2r}$). Then it follows from \eqref{eq:chi2} that
\begin{equation*}
 \P(Y_r - r \ge (\lambda/2)
\sqrt{2r}) \le \P(Y_r - r \ge (\lambda/4)
\sqrt{2r} + (\lambda/4)^2) \le e^{-\lambda^2/32} = n^{-4}.
\end{equation*}
Hence
\[
\P(Z_r - r/n \ge 16 \sqrt{r \log n}/n) \le 2n^{-4}
\]
and, therefore,
\begin{equation}
\label{eq:conclusion1} \P(\max_i \|\mtx{P}_U \vct{e}_i\|^2 - r/n \ge
16 \sqrt{r \log n}/n) \le 2 n^{-3}
\end{equation}
by the union bound. Note that \eqref{eq:conclusion1} establishes the
first claim of the lemma (even for $r < 4 \log n$ since in this case
$Z_r \le Z_{\lceil 4\log n \rceil}$).

It remains to establish the second claim. Notice that by symmetry,
$\mtx{E} = \sum_{1 \le k \le r} \mtx{u}_k \mtx{v}_k^*$ has the same
distribution as
\[
\mtx{F} = \sum_{k = 1}^r \epsilon_k \mtx{u}_k \mtx{v}_k^*,
\]
where $\{\epsilon_k\}$ is an independent Rademacher sequence. It then
follows from Hoeffding's inequality that conditional on
$\{\mtx{u}_k\}$ and $\{\mtx{v}_k\}$ we have
\[
\P(|F_{ij}| > t) \le 2 e^{-t^2/2\sigma_{ij}^2}, \quad \sigma_{ij}^2
= \sum_{1 \le k \le r} u_{ik}^2 v_{ik}^2.
\]
Our previous results indicate that $\max_{ij} |v_{ij}|^2 \le (10\log
n)/n$ with large probability and thus
\[
\sigma^2_{ij} \le 10 \, \frac{\log n}{n} \, \|\mtx{P}_U
\vct{e}_i\|^2.
\]
Set $\bar r = \max(r,\log n)$. Since $\|\mtx{P}_U \vct{e}_i\|^2 \le
C \bar r/n$ with large probability, we have
\[
\sigma^2_{ij} \le C  (\log n) \, \bar r/n^2
\]
with large probability. Hence the marginal distribution of $F_{ij}$
obeys
\[
\P(|F_{ij}| > \lambda \sqrt{\bar r}/n) \le 2 e^{-\gamma
\lambda^2/\log n} + \P(\sigma_{ij}^2 \ge C (\log n) {\bar r}/n^2).
\]
for some numerical constant $\gamma$. Picking $\lambda = \gamma' \log
n$ where $\gamma'$ is a sufficiently large numerical constant gives
\[
\|\mtx{F}\|_\infty \le C\, (\log n) \, \sqrt{\bar r}/n
\]
with large probability. Since $\mtx{E}$ and $\mtx{F}$ have the same
distribution, the second claim follows.

The claim about the size of $\max_{ij} |v_{ij}|^2$ is straightforward
since our techniques show that for each $\lambda > 0$
\[
\P(Z_1 \ge \lambda (\log n)/n) \le\ \P(Y_1 \ge \lambda (1-\epsilon) \log
n) + e^{-\epsilon^2 n/4}.
\]
Moreover,
\[
\P(Y_1 \ge \lambda (1-\epsilon) \log n) = \P(|x_1| \ge
\sqrt{\lambda(1-\epsilon) \log n}) \le 2e^{-\frac{1}{2}\lambda
(1-\epsilon) \log
  n}.
\]
If $n$ is sufficiently large so that $\epsilon \le 1/5$, this gives
$\P(Z_1 \ge 10 (\log n)/n) \le 3n^{-4}$
and, therefore,
\[
\P(\max_{ij} |v_{ij}|^2 \ge 10 (\log n)/n) \le 12 n^{-3} \log n
\]
since the maximum is taken over at most $4 n \log n$ pairs.

\section{Duality}
\label{sec:duality}

Let ${\cal R}_\Omega:\R^{n_1\times n_2}\rightarrow \R^{|\Omega|}$ be
the sampling operator which extracts the observed entries, ${\cal
R}_\Omega(\mtx{X}) = (X_{ij})_{ij \in
  \Omega}$, so that the constraint in \eqref{eq:sdp} becomes ${\cal
  R}_\Omega(\mtx{X}) = {\cal R}_\Omega(\mtx{M})$. Standard convex
optimization theory asserts that $\mtx{X}$ is solution to
\eqref{eq:sdp} if there exists a dual vector (or Lagrange multiplier)
$\lambda\in\R^{|\Omega|}$ such that ${\cal R}_\Omega^* \, \lambda$ is
a {\em subgradient} of the nuclear norm at the point $\mtx{X}$, which
we denote by
\begin{equation}
  \label{eq:kkt}
  {\cal R}_\Omega^* \, \lambda \in \partial \|\mtx{X}\|_*
\end{equation}
(see, e.g.~\cite{BertsekasConvexBook}). Recall the definition of a
subgradient of a convex function $f : \R^{n_1
  \times n_2} \goto \R$. We say that $\mtx{Y}$ is a subgradient of $f$ at
$\mtx{X}_0$, denoted $\mtx{Y} \in \partial f(\mtx{X}_0)$, if
\begin{equation}
  \label{eq:subgradient}
  f(\mtx{X}) \ge f(\mtx{X}_0) + \<\mtx{Y}, \mtx{X} - \mtx{X}_0\>
\end{equation}
for all $\mtx{X}$.

Suppose $\mtx{X}_0 \in \R^{n_1 \times n_2}$ has rank $r$ with a
singular value decomposition given by
\begin{equation}
  \label{eq:near-svd}
  \mtx{X}_0 = \sum_{1 \le k \le r} \sigma_k \, \vct{u}_k \vct{v}_k^*,
\end{equation}
With these notations, $\mtx{Y}$ is a subgradient of the nuclear norm
at $\mtx{X}_0$ if and only if it is of the form
\begin{equation}
  \label{eq:sub-decomp}
  \mtx{Y} = \sum_{1 \le k \le r} \vct{u}_k \vct{v}_k^* + \mtx{W},
\end{equation}
where $\mtx{W}$ obeys the following two properties:
\begin{itemize}
\item[(i)] the column space of $\mtx{W}$ is orthogonal to $U \equiv
  \lspan{(\vct{u}_1, \ldots, \vct{u}_r)}$, and the row space of $\mtx{W}$ is orthogonal
  to $V \equiv \lspan{(\vct{v}_1, \ldots, \vct{v}_r)}$;
\item[(ii)] the spectral norm of $\mtx{W}$ is less than or equal to 1.
\end{itemize}
(see, e.g.,~\cite{Lewis03,Watson92}). To express these properties
concisely, it is convenient to introduce the orthogonal decomposition
$\R^{n_1 \times n_2} = T \oplus T^\perp$ where $T$ is the linear space
spanned by elements of the form $\vct{u}_k \vct{x}^*$ and $\vct{y}
\vct{v}_k^*$, $1 \le k \le r$, where $\vct{x}$ and $\vct{y}$ are
arbitrary, and $T^\perp$ is its orthogonal complement.  Note that
$\textrm{dim}(T) = r(n_1+n_2-r)$, precisely the number of degrees of
freedom in the set of $n_1 \times n_2$ matrices of rank $r$. $T^\perp$ is
the subspace of matrices spanned by the family $(\vct{x}\vct{y}^*)$,
where $\vct{x}$ (respectively $\vct{y}$) is any vector orthogonal to
$U$ (respectively $V$).

The orthogonal projection $\PT$ onto $T$ is given by
\begin{equation}
  \label{eq:PT}
  \PT(\mtx{X}) = \mtx{P}_U \mtx{X} + \mtx{X} \mtx{P}_V  - \mtx{P}_U \mtx{X} \mtx{P}_V,
\end{equation}
where $\mtx{P}_{U}$ and $\mtx{P}_V$ are the orthogonal projections
onto $U$ and $V$.  Note here that while $\mtx{P}_{U}$ and
$\mtx{P}_V$ are matrices, $\PT$ is a linear operator mapping
matrices to matrices.  We also have
\[
\PTc(\mtx{X}) = (\OpId - \PT)(\mtx{X}) = (\mtx{I}_{n_1} -
\mtx{P}_{U}) \mtx{X} (\mtx{I}_{n_2} - \mtx{P}_{V})
\]
where $\mtx{I}_d$ denotes the $d\times d$ identity matrix. With
these notations, $\mtx{Y} \in
\partial \|\mtx{X}_0\|_*$ if
\begin{itemize}
\item[(i')] ${\cal P}_T(\mtx{Y}) = \sum_{1 \le k \le r} \vct{u}_k \vct{v}_k^*$,
\item[(ii')] and $\|{\cal P}_{T^\perp} \mtx{Y}\| \le 1$.
\end{itemize}

Now that we have characterized the subgradient of the nuclear norm,
the lemma below gives sufficient conditions for the uniqueness of
the minimizer to \eqref{eq:sdp}.
\begin{lemma}
  \label{teo:unicity} Consider a matrix $\mtx{X}_0 = \sum_{k = 1}^r \sigma_k
  \, \vct{u}_k \vct{v}_k^*$ of rank $r$ which is feasible for the problem
  \eqref{eq:sdp}, and suppose that the following two conditions hold:
  \begin{enumerate}
  \item there exists a dual point $\lambda$ such that $\mtx{Y} = {\cal
      R}^*_\Omega \lambda$ obeys
\begin{equation}
\label{eq:key-decomp} {\cal P}_T(\mtx{Y}) = \sum_{k = 1}^r \vct{u}_k
\vct{v}_k^*,
    \qquad \|{\cal P}_{T^\perp}(\mtx{Y})\| < 1;
\end{equation}
\item the sampling operator ${\cal R}_\Omega$ restricted to elements
  in $T$ is injective.
  \end{enumerate}
Then $\mtx{X}_0$ is the unique minimizer.
\end{lemma}

Before proving this result, we would like to emphasize that this
lemma provides a clear strategy for proving our main result, namely,
Theorem \ref{teo:main2}. Letting $\mtx{M} = \sum_{k = 1}^r \sigma_k
\, \vct{u}_k \vct{v}_k^*$, $\mtx{M}$ is the unique solution to
\eqref{eq:sdp} if the injectivity condition holds and if one can
find a dual point $\lambda$ such that $\mtx{Y} = {\cal R}^*_\Omega
\lambda$ obeys \eqref{eq:key-decomp}.

The proof of Lemma \ref{teo:unicity} uses a standard fact which
states that the nuclear norm and the spectral norm are dual to one
another.
\begin{lemma}
\label{teo:dual}
  For each pair $\mtx{W}$ and $\mtx{H}$, we have
\[
\<\mtx{W}, \mtx{H}\> \le \|\mtx{W}\| \, \|\mtx{H}\|_{*}.
\]
In addition, for each $\mtx{H}$, there is a $\mtx{W}$ obeying
$\|\mtx{W}\| = 1$ which achieves the equality.
\end{lemma}
A variety of proofs are available for this Lemma, and an elementary
argument is sketched in~\cite{Recht07}.  We now turn to the proof of
Lemma \ref{teo:unicity}.

\begin{proof}[of Lemma~\ref{teo:unicity}]
Consider any perturbation $\mtx{X}_0 + \mtx{H}$ where ${\cal
R}_\Omega(\mtx{H}) = 0$. Then for any $\mtx{W}^0$ obeying (i)--(ii),
$\sum_{k = 1}^r \vct{u}_k \vct{v}_k^* + \mtx{W}^0$ is a subgradient
of the nuclear norm at $X_0$ and, therefore,
\[
\norm{\mtx{X}_0+\mtx{H}}_* \ge \norm{\mtx{X}_0}_* + \<\sum_{k=1}^r
\vct{u}_k \vct{v}_k^* + \mtx{W}^0,\mtx{H}\>.
\]
Letting $\mtx{W} = \PTc(\mtx{Y})$, we may write $\sum_{k=1}^r
\vct{u}_k \vct{v}_k^* = {\cal R}^*_\Omega \lambda - \mtx{W}$.  Since
$\|\mtx{W}\| < 1$ and ${\cal R}_\Omega(\mtx{H}) = 0$,  it then
follows that
\[
\norm{\mtx{X}_0+\mtx{H}}_* \ge \norm{\mtx{X}_0}_* +  \<\mtx{W}^0 -
\mtx{W},\mtx{H}\>.
\]
Now by construction
\[
\<\mtx{W}^0 - \mtx{W},\mtx{H}\> = \<{\cal P}_{T^\perp} (\mtx{W}^0 -
\mtx{W}), \mtx{H}\> = \<\mtx{W}^0 - \mtx{W}, {\cal
  P}_{T^\perp} (\mtx{H})\>.
\]

We use Lemma \ref{teo:dual} and set $\mtx{W}^0 = \PTc(\mtx{Z})$ where
$\mtx{Z}$ is any matrix obeying $\|\mtx{Z}\| \le 1$ and
$\<\mtx{Z},\PTc(\mtx{H})\> = \|\PTc(\mtx{H})\|_*$. Then $\mtx{W}^0 \in
T^\perp$, $\| \mtx{W}^0\| \le 1$, and
\[
\<\mtx{W}^0 - \mtx{W},\mtx{H}\> \ge (1 - \|\mtx{W}\|) \,
\|\PTc(\mtx{H})\|_*,
\]
which by assumption is strictly positive unless $\PTc(\mtx{H}) = 0$.
In other words, $\|\mtx{X}_0+\mtx{H}\|_*
> \|\mtx{X}_0\|_*$ unless $\PTc(\mtx{H}) = 0$. Assume then that
$\PTc(\mtx{H}) = 0$ or equivalently that $\mtx{H} \in T$. Then
${\cal R}_\Omega(\mtx{H}) = 0$ implies that $\mtx{H} = 0$ by the
injectivity assumption.  In conclusion, $\|\mtx{X}_0 + \mtx{H}\|_*
> \|\mtx{X}\|_*$ unless $\mtx{H} = 0$.
\end{proof}

\section{Architecture of the proof}
\label{sec:architecture}

Our strategy to prove that $\mtx{M} = \sum_{1 \le k \le r} \sigma_k
\vct{u}_k \vct{v}_k^*$ is the unique minimizer to \eqref{eq:sdp} is
to construct a matrix $\mtx{Y}$ which vanishes on $\Omega^c$ and obeys the
conditions of Lemma \ref{teo:unicity} (and show the injectivity of
the sampling operator restricted to matrices in $T$ along the way).
Set $\PO$ to be the orthogonal projector onto the indices in
$\Omega$ so that the $(i,j)$th component of $\PO(\mtx{X})$ is equal
to $X_{ij}$ if $(i,j) \in \Omega$ and zero otherwise.  Our
candidate $\mtx{Y}$ will be the solution to
\begin{equation}
\label{eq:frob2}
  \begin{array}{ll}
    \textrm{minimize}   & \quad \|\mtx{X}\|_F \\
    \textrm{subject to} & \quad  ({\cal P}_T {\cal P}_\Omega)(\mtx{X})  =  \sum_{k  = 1}^r \vct{u}_k \vct{v}_k^*.
  \end{array}
\end{equation}
The matrix $\mtx{Y}$ vanishes on $\Omega^c$ as otherwise it would not be an
optimal solution since $\PO(\mtx{Y})$ would obey the constraint and
have a smaller Frobenius norm. Hence $\mtx{Y} = \PO(\mtx{Y})$ and
$\PT(\mtx{Y}) = \sum_{k = 1}^r \vct{u}_k \vct{v}_k^*$. Since the
Pythagoras formula gives
\begin{align*}
  \|\mtx{Y}\|^2_{F} = \|\PT(\mtx{Y})\|_F^2 +
  \|\PTc(\mtx{Y})\|_F^2 & = \|\sum_{k = 1}^r \vct{u}_k \vct{v}_k^*\|_F^2 + \|\PTc(\mtx{Y})\|_F^2\\
  & = r + \|\PTc(\mtx{Y})\|_F^2,
\end{align*}
minimizing the Frobenius norm of $\mtx{X}$ amounts to minimizing the
Frobenius norm of $\PTc(\mtx{X})$ under the constraint $\PT(\mtx{X})
= \sum_{k = 1}^r \vct{u}_k \vct{v}_k^*$. Our motivation is twofold.
First, the solution to the least-squares problem \eqref{eq:frob2}
has a closed form that is amenable to analysis. Second, by forcing
$\PTc(\mtx{Y})$ to be small in the Frobenius norm, we hope that it
will be small in the spectral norm as well, and establishing that
$\|\PTc(\mtx{Y})\| < 1$ would prove that $\mtx{M}$ is the unique
solution to \eqref{eq:sdp}.

To compute the solution to \eqref{eq:frob2}, we introduce the
operator $\cAOT$ defined by
\[
\cAOT(\mtx{M}) = \PO \PT (\mtx{M}).
\]
Then, if $\cAOT^* \cAOT = \PT \PO \PT$ has full rank when restricted
to $T$, the minimizer to \eqref{eq:frob2} is given by
\begin{equation}
  \label{eq:Lambda1}
  \mtx{Y} = \cAOT (\cAOT^* \cAOT)^{-1}(\mtx{E}),
\qquad \mtx{E} \equiv  \sum_{k  = 1}^r \vct{u}_k \vct{v}_k^*.
\end{equation}
We clarify the meaning of \eqref{eq:Lambda1} to avoid any confusion.
$(\cAOT^* \cAOT)^{-1}(\mtx{E})$ is meant to be that element
$\mtx{F}$ in $T$ obeying $(\cAOT^* \cAOT)(\mtx{F}) = \mtx{E}$.

To summarize the aims of our proof strategy,
\begin{itemize}
\item We must first show that $\cAOT^* \cAOT = \PT \PO \PT$ is a
  one-to-one linear mapping from $T$ onto itself.  In this case,
  $\cAOT = \PO \PT$---as a mapping from $T$ to $\R^{n_1 \times
    n_2}$---is injective. This is the second sufficient condition of
  Lemma \ref{teo:unicity}.  Moreover, our ansatz for $\mtx{Y}$ given
  by \eqref{eq:Lambda1} is well-defined.

\item Having established that $\mtx{Y}$ is well-defined, we will show that
\[
\|{\cal P}_{T^\perp}(\mtx{Y})\| < 1,
\]
thus proving the first sufficient condition.
\end{itemize}

\subsection{The Bernoulli model}
\label{sec:bernoulli}

Instead of showing that the theorem holds when $\Omega$ is a set of
size $m$ sampled uniformly at random, we prove the theorem for a
subset $\Omega'$ sampled according to the {\em Bernoulli model}. Here
and below, $\{\delta_{ij}\}_{1\le i \le n_1, 1\le j \le n_2}$ is a
sequence of independent identically distributed $0/1$ Bernoulli random
variables with
\begin{equation}
\label{eq:delta}
\P(\delta_{ij} = 1) = p \equiv \frac{m}{n_1n_2},
\end{equation}
and define
\begin{equation}
\label{eq:Omega}
  \Omega' = \{ (i,j) : \delta_{ij} = 1\}.
\end{equation}
Note that $\E |\Omega'| = m$, so that the average cardinality of
$\Omega'$ is that of $\Omega$.  Then following the same reasoning as
the argument developed in Section II.C of~\cite{CRT06} shows that
the probability of `failure' under the uniform model is bounded by 2
times the probability of failure under the Bernoulli model; the
failure event is the event on which the solution to \eqref{eq:sdp}
is not exact.  Hence, we can restrict our attention to the Bernoulli
model and from now on, we will assume that $\Omega$ is given by
\eqref{eq:Omega}.  This is advantageous because the Bernoulli model
admits a simpler analysis than uniform sampling thanks to the
independence between the $\delta_{ij}$'s.

\subsection{The injectivity property}
\label{sec:injectivity}

We study the injectivity of $\cAOT$, which also shows that $\mtx{Y}$
is well-defined.  To prove this, we will show that the linear operator
$p^{-1} \PT (\PO - p \OpId) \PT$ has small operator norm, which we
recall is $\sup_{\|\mtx{X}\|_F \le 1} \, p^{-1}\|\PT (\PO - p \OpId)
\PT(\mtx{X})\|_F$.
\begin{theorem}
  \label{teo:rudelson}
  Suppose $\Omega$ is sampled according to the Bernoulli model
  \eqref{eq:delta}--\eqref{eq:Omega} and put $n =
  \max(n_1,n_2)$. Suppose that the coherences obey
  $\max(\mu(U),\mu(V)) \le \mu_0$.  Then, there is a numerical
  constants $C_R$ such that for all $\beta>1$,
  \begin{equation}
    \label{eq:near-isometry}
    p^{-1} \, \|\PT \PO \PT - p \PT\| \le C_R \,
\sqrt{\frac{\mu_0 \, nr (\beta \log n)}{m}}
  \end{equation}
  with probability at least $1-3n^{-\beta}$ provided that $C_R \,
  \sqrt{\frac{\mu_0 \, nr (\beta \log n)}{m}} < 1$.
\end{theorem}
\begin{proof}
Decompose any matrix $\mtx{X}$ as $\mtx{X} = \sum_{ab} \<\mtx{X},
\eab\> \eab$ so that
\[
\PT(\mtx{X}) = \sum_{ab} \<\PT(\mtx{X}), \eab\> \eab = \sum_{ab}
\<\mtx{X}, \PT(\eab)\> \eab.
\]
Hence, $\PO \PT(\mtx{X}) = \sum_{ab} \delta_{ab} \, \<\mtx{X},
\PT(\eab)\> \, \eab$ which gives
\[
(\PT \PO \PT)(\mtx{X}) = \sum_{ab} \delta_{ab} \, \<\mtx{X},
\PT(\eab)\> \,\PT(\eab).
\]
In other words,
\[
\PT \PO \PT = \sum_{ab} \delta_{ab} \, \PT(\eab) \otimes  \PT(\eab).
\]
It follows from the definition \eqref{eq:PT} of $\PT$ that
\begin{equation}
  \label{eq:PTeab}
  \PT(\eab) = (\mtx{P}_U \vct{e}_a) \vct{e}_b^* + \vct{e}_a (\mtx{P}_V \vct{e}_b)^* - (\mtx{P}_U \vct{e}_a)(\mtx{P}_V \vct{e}_b)^*.
\end{equation}
This gives
\begin{equation}
  \label{eq:PTeabF2} \|\PT(\eab)\|_F^2 = \<\PT(\eab), \eab\>   = \|\mtx{P}_U
  \vct{e}_a\|^2 + \|\mtx{P}_V
  \vct{e}_b\|^2 - \|\mtx{P}_U \vct{e}_a\|^2 \, \|\mtx{P}_V \vct{e}_b\|^2
\end{equation}
and since $\|\mtx{P}_U \vct{e}_a\|^2 \le \mu(U)r/n_1$ and $\|\mtx{P}_V
\vct{e}_b\|^2 \le \mu(U)r/n_2$,
\begin{equation}
  \|\PT(\eab)\|_F^2 \le  2\mu_0 r/\min(n_1,n_2).
 \label{eq:PTeabF}
\end{equation}
Now the fact that the operator $\PT \PO \PT$ does not deviate from its
expected value
\[
\E (\PT \PO \PT) = \PT  (\E \PO) \PT  = \PT (p \OpId) \PT = p\PT
\]
in the spectral norm is related to Rudelson's selection theorem
\cite{rudelson}. The first part of the theorem below may be found in
\cite{CR07} for example, see also \cite{RV_JACM} for a very similar
statement.
\begin{theorem}\cite{CR07}
\label{teo:rudelson2}
Let $\{\delta_{ab}\}$ be independent 0/1 Bernoulli variables with
$\P(\delta_{ab} = 1) = p = \frac{m}{n_1 n_2}$ and put $n = \max(n_1,
n_2)$. Suppose that $\|\PT(\eab)\|^2_F \le 2\mu_0 r/n$.
Set
\[
Z \equiv p^{-1} \|\sum_{ab}
(\delta_{ab} -p)\, \PT(\eab) \otimes  \PT(\eab)\| =
p^{-1} \|\PT \PO \PT - p \PT\|.
\]
\begin{enumerate}
\item There exists a constant $C'_R$ such that
  \begin{equation}
    \label{eq:ERudel}
    \E Z \le C'_R \, \sqrt{\frac{\mu_0 \, nr \, \log n}{m}}
  \end{equation}
provided that the right-hand side is smaller than 1.
\item Suppose $\E Z \le 1$. Then for each $\lambda > 0$, we have
  \begin{equation}
    \label{eq:largeRudel}
    \P\left(|Z - \E Z| > \lambda \, \sqrt{\frac{\mu_0 \, nr \,\log n}{m}}\right)
    \le 3 \exp\left(-\gamma'_0 \, \min \left\{\lambda^2 \log n, \lambda \sqrt{\frac{m \log n}{\mu_0\, nr}}\right\}\right)
  \end{equation}
for some positive constant $\gamma'_0$.
\end{enumerate}
\end{theorem}
As mentioned above, the first part, namely, \eqref{eq:ERudel} is an
application of an established result which states that if $\{y_i\}$ is a
family of vectors in $\R^d$ and $\{\delta_i\}$ is a 0/1 Bernoulli
sequence with $\P(\delta_i = 1) = p$, then
\[
p^{-1} \|\sum_i (\delta_i -p) y_i \otimes y_i\| \le C\,
\sqrt{\frac{\log d}{p}} \max_i \|y_i\|
\]
for some $C > 0$ provided that the right-hand side is less than 1.
The proof may be found in the cited literature, e.g.~in
\cite{CR07}. Hence, the first part follows from applying this result
to vectors of the form $\PT(\eab)$ and using the available bound on
$\|\PT(\eab)\|_F$. The second part follows from Talagrand's
concentration inequality and may be found in the Appendix.

Set $\lambda = \sqrt{\beta/\gamma'_0}$ and assume that $m >
(\beta/\gamma'_0)\mu_0 \, nr \log n$. Then the left-hand side of
\eqref{eq:largeRudel} is bounded by $3n^{-\beta}$ and thus, we
established that
\[
Z \le C'_R \sqrt{\frac{\mu_0 \, nr \, \log n}{m}} +
\frac{1}{\sqrt{\gamma'_0}} \sqrt{\frac{\mu_0 \, nr \, \beta \log
    n}{m}}
\]
with probability at least $1 - 3n^{-\beta}$. Setting $C_R = C'_R +
1/\sqrt{\gamma'_0}$ finishes the proof.
\end{proof}

Take $m$ large enough so that $C_R \, \sqrt{\mu_0 \, (nr/m) \log n}
\le 1/2$. Then it follows from \eqref{eq:near-isometry} that
\begin{equation}
\label{eq:near-isometry2} \frac{p}{2} \|\PT(\mtx{X})\|_F \le \|(\PT
\PO \PT)(\mtx{X})\|_F \le \frac{3p}{2} \|\PT(\mtx{X})\|_F
\end{equation}
for all $\mtx{X}$ with large probability. In particular, the operator
$\cAOT^* \cAOT = \PT \PO \PT$ mapping $T$ onto itself is
well-conditioned and hence invertible. An immediate consequence is the
following:
\begin{corollary}
\label{teo:POPT}
  Assume that $C_R \, \sqrt{\mu_0 nr (\log n)/m} \le 1/2$. With the same
  probability as in Theorem \ref{teo:rudelson}, we have
\begin{equation}
\label{eq:POPT} \|\PO \PT(\mtx{X})\|_F \le \sqrt{3p/2}
\|\PT(\mtx{X})\|_F.
\end{equation}
\end{corollary}
\begin{proof}
We have $\|\PO \PT(\mtx{X})\|_F^2 = \<\mtx{X}, (\PO \PT)^* (\PO \PT)
\mtx{X}\> = \<\mtx{X}, (\PT \PO \PT) \mtx{X}\>$ and thus
\[
\|\PO \PT(\mtx{X})\|_F^2 = \<\PT \mtx{X}, (\PT \PO \PT) \mtx{X}\>
\le \|\PT(\mtx{X})\|_F \, \|(\PT \PO \PT)(\mtx{X})\|_F,
\]
where the inequality is due to Cauchy-Schwarz. The conclusion
\eqref{eq:POPT} follows from \eqref{eq:near-isometry2}.
\end{proof}

\subsection{The size property}
\label{sec:size}

In this section, we explain how we will show that $\|\PTc(\mtx{Y})\|
< 1$. This result will follow from five lemmas that we will prove
in Section~\ref{sec:proofs}.  Introduce
\[
\mathcal{H} \equiv \PT - p^{-1} \PT \PO \PT,
\]
which obeys $\|\mathcal{H}(\mtx{X})\|_F \le C_R \, \sqrt{\mu_0
(nr/m) \, \beta \log n} \|\PT(\mtx{X})\|_F$ with large probability
because of Theorem \ref{teo:rudelson}.  For any matrix $\mtx{X}\in
T$, $(\PT \PO \PT)^{-1}(\mtx{X})$ can be expressed in terms of the
power series
\[
(\PT \PO \PT)^{-1}(\mtx{X}) = p^{-1}(\mtx{X} + \mathcal{H}(\mtx{X})
+ \mathcal{H}^2(\mtx{X}) + \ldots)
\]
for $\mathcal{H}$ is a contraction when $m$ is sufficiently large.
Since $\mtx{Y} = \PO \PT (\PT \PO \PT)^{-1}(\sum_{1\le k \le r}
\vct{u}_k \vct{v}_k^*)$, $\PTc(\mtx{Y})$ may be decomposed as
\begin{equation}
\label{eq:size} \PTc(\mtx{Y}) = p^{-1} (\mathcal{P}_{T^\perp} \PO \PT)
(\mtx{E} + \mathcal{H}(\mtx{E}) + {\cal
  H}^2(\mtx{E}) + \ldots), \qquad \mtx{E} = \sum_{1 \le k \le r} \vct{u}_k
\vct{v}_k^*.
\end{equation}
To bound the norm of the left-hand side, it is of course sufficient to
bound the norm of the summands in the right-hand side.  Taking the
following five lemmas together establishes Theorem \ref{teo:main2}.
\begin{lemma}
  \label{teo:order0}
  Fix $\beta \ge 2$ and $\lambda \ge 1$.  There is a numerical
  constant $C_0$ such that if $m \ge \lambda \, \mu_1^2 \, nr \beta
  \log n$, then
  \begin{equation}
  \label{eq:order0} p^{-1} \, \|(\PTc \PO \PT) \mtx{E}\| \le
C_0 \, \lambda^{-1/2}.
\end{equation}
with probability at least $1-n^{-\beta}$.
\end{lemma}
\begin{lemma}
  \label{teo:order1}
  Fix $\beta \ge 2$ and $\lambda \ge 1$.  There are numerical
  constants $C_1$ and $c_1$ such that if $m \ge \lambda \, \mu_1
  \max(\sqrt{\mu_0}, \mu_1) \, nr \beta \log n$, then
\begin{equation}
  \label{eq:order1} p^{-1} \, \|(\PTc \PO \PT) \mathcal{H}(\mtx{E})\|
  \le C_1 \, \lambda^{-1}
\end{equation}
with probability at least $1-c_1 n^{-\beta}$.
\end{lemma}
\begin{lemma}
  \label{teo:order2}
  Fix $\beta \ge 2$ and $\lambda \ge 1$. There are numerical constants
  $C_2$ and $c_2$ such that if $m \ge \lambda \, \mu_0^{4/3} \,
  nr^{4/3} \beta \log n$, then
\begin{equation}
\label{eq:order2} p^{-1} \, \|(\PTc \PO \PT) \mathcal{H}^2(\mtx{E})\|
\le C_2 \, \lambda^{-3/2}
\end{equation}
with probability at least $1-c_2 n^{-\beta}$.
\end{lemma}
\begin{lemma}
  \label{teo:order3}
  Fix $\beta \ge 2$ and $\lambda \ge 1$. There are numerical constants
  $C_3$ and $c_3$ such that if $m \ge \lambda \mu_0^2 \, nr^2 \beta
  \log n$, then
  \begin{equation}
\label{eq:order3} p^{-1} \, \|(\PTc \PO \PT)
\mathcal{H}^3(\mtx{E})\| \le C_3 \, \lambda^{-1/2}
\end{equation}
with probability at least $1-c_3 n^{-\beta}$.
\end{lemma}
\begin{lemma}
  \label{teo:highorder}
  Under the assumptions of Theorem \ref{teo:rudelson}, there is a
  numerical constant $C_{k_0}$ such that if $m \ge
(2C_R)^2 \mu_0 nr \beta \log n$, then
  \begin{equation}
    \label{eq:highorder} p^{-1} \, \|(\PTc \PO \PT) \sum_{k \ge k_0}
    \mathcal{H}^{k}(\mtx{E})\| \le C_{k_0} \, \left(\frac{n^2 r}{m}\right)^{1/2} \, \left(\frac{\mu_0 nr \beta \log n}{m}\right)^{k_0/2}
\end{equation}
with probability at least $1-n^{-\beta}$.
\end{lemma}

Let us now show how we may combine these lemmas to prove our main
results. Under all of the assumptions of Theorem~\ref{teo:main2},
consider the four Lemmas \ref{teo:order0}, \ref{teo:order1},
\ref{teo:order2} and \ref{teo:highorder}, the latter applied with $k_0
= 3$. Together they imply that there are numerical constants
$c$ and $C$ such that $\|\PTc(\mtx{Y})\| < 1$ with
probability at least $1-c n^{-\beta}$ provided that the number
of samples obeys
\begin{equation}
  \label{eq:bound1}
  m \ge C \, \, \max(\mu_1^2, \mu_0^{1/2} \mu_1, \mu_0^{4/3} r^{1/3}, \mu_0 n^{1/4}) \, nr \beta \log n
\end{equation}
for some constant $C$. The four expressions in the maximum come
from Lemmas \ref{teo:order0}, \ref{teo:order1}, \ref{teo:order2} and
\ref{teo:highorder} in this order. Now the bound \eqref{eq:bound1} is
only interesting in the range when $\mu_0 n^{1/4} r$ is smaller than a
constant times $n$ as otherwise the right-hand side is greater than
$n^2$ (this would say that one would see all the entries in which case
our claim is trivial). When $\mu_0 r \le n^{3/4}$, $(\mu_0 r)^{4/3}
\le \mu_0 n^{5/4} r$ and thus the recovery is exact provided that $m$
obeys \eqref{eq:main3}.

For the case concerning small values of the rank, we consider all five
lemmas and apply Lemma \ref{teo:highorder}, the latter applied with
$k_0 = 4$. Together they imply that $\|\PTc(\mtx{Y})\| < 1$ with
probability at least $1-c n^{-\beta}$ provided that the number
of samples obeys
\begin{equation}
  \label{eq:bound2}
  m \ge C \max(\mu_0^2 r, \mu_0 n^{1/5}) \,  nr \beta \log n
\end{equation}
for some constant $C$. The two expressions in the maximum come
from  Lemmas \ref{teo:order3} and \ref{teo:highorder} in this
order. The reason for this simplified formulation is that the terms
$\mu_1^2$, $\mu_0^{1/2} \mu_1$ and $\mu_0^{4/3} r^{1/3}$ which come
from Lemmas \ref{teo:order0}, \ref{teo:order1} and \ref{teo:order2}
are bounded above by $\mu_0^2 r$ since $\mu_1 \le \mu_0 \sqrt{r}$.
When $\mu_0 r \le n^{1/5}$, the recovery is exact provided that $m$
obeys \eqref{eq:main4}.

\section{Connections with Random Graph Theory}
\label{sec:coupon}

\subsection{The injectivity property and the coupon collector's problem}

We argued in the Introduction that to have any hope of recovering an
unknown matrix of rank 1 by any method whatsoever, one needs at
least one observation per row and one observation per column. Sample
$m$ entries uniformly at random.  Viewing the row indices as bins,
assign the $k$th sampled entry to the bin corresponding to its row
index. Then to have any hope of recovering our matrix, all the bins
need to be occupied. Quantifying how many samples are required to
fill all of the bins is the famous \emph{coupon collector's
problem}.

Coupon collection is also connected to the injectivity of the
sampling operator $\PO$ restricted to elements in $T$. Suppose we
sample the entries of a rank 1 matrix equal to $\vct{x} \vct{y}^*$
with left and right singular vectors $\vct{u} = \vct{x}/\|\vct{x}\|$
and $\vct{v} = \vct{y}/\|\vct{y}\|$ respectively and have not seen
anything in the $i$th row. Then we claim that $\PO$ (restricted to
$T$) has a nontrivial null space and thus $\PT \PO \PT$ is not
invertible. Indeed, consider the matrix $\vct{e}_i \vct{v}^*$. This
matrix is in $T$ and
\[
\PO(\vct{e}_i \vct{v}^*) = 0
\] since $\vct{e}_i \vct{v}^*$ vanishes outside of the $i$th row. The same applies
to the columns as well. If we have not seen anything in column $j$,
then the rank-1 matrix $\vct{u} \vct{e}_j^* \in T$ and $\PO(\vct{u}
\vct{e}_j^*) = 0$. In conclusion, the invertibility of $\PT \PO \PT$
implies a complete collection.

When the entries are sampled uniformly at random, it is well known
that one needs on the order of $n \log n$ samples to sample all the
rows. What is interesting is that Theorem \ref{teo:rudelson} implies
that $\PT \PO \PT$ is invertible---a stronger property---when the
number of samples is also on the order of $n \log n$. A particular
implication of this discussion is that the logarithmic factors in
Theorem \ref{teo:rudelson} are unavoidable.

\subsection{The injectivity property and the connectivity problem}

To recover a matrix of rank 1, one needs much more than at
least one observation per row and column. Let $R$ be the set of row
indices, $1 \le i \le n$, and $C$ be the set of column indices, $1 \le
j \le n$, and consider the bipartite graph connecting vertices $i \in
R$ to vertices $j \in C$ if and only if $(i,j) \in \Omega$, i.e.~the
$(i,j)$th entry is observed. We claim that if this graph is not fully
connected, then one cannot hope to recover a matrix of rank 1.

To see this, we let $I$ be the set of row indices and $J$ be the set
of column indices in any connected component. We will assume that $I$
and $J$ are nonempty as otherwise, one is in the previously discussed
situation where some rows or columns are not sampled.  Consider a
rank 1 matrix equal to $\vct{x} \vct{y}^*$ as before with singular
vectors $\vct{u} = \vct{x}/\|\vct{x}\|$ and $\vct{v} =
\vct{y}/\|\vct{y}\|$. Then all the information about the values of the
$x_i$'s with $i \in I$ and of the $y_j$'s with $j \in J$ are given by
the sampled entries connecting $I$ to $J$ since all the other observed
entries connect vertices in $I^c$ to those in $J^c$.  Now even if one
observes all the entries $x_i y_j$ with $i \in I$ and $j \in J$, then
at least the signs of $x_i$, $i \in I$, and of $y_j$, $j \in J$, would
remain undetermined. Indeed, if the values $(x_i)_{i \in I}$,
$(y_j)_{j \in J}$ are consistent with the observed entries, so are the
values $(-x_i)_{i \in I}$, $(-y_j)_{j \in J}$. However, since the same
analysis holds for the sets $I^c$ and $J^c$, there are at least two
matrices consistent with the observed entries and exact matrix
completion is impossible.

The connectivity of the graph is also related to the injectivity of
the sampling operator $\PO$ restricted to elements in $T$. If the
graph is not fully connected, then we claim that $\PO$ (restricted to
$T$) has a nontrivial null space and thus $\PT \PO \PT$ is not
invertible. Indeed, consider the matrix
\[
\vct{M} = \vct{a} \vct{v}^* + \vct{u} \vct{b}^*,
\]
where $a_i = -u_i$ if $i \in I$ and $a_i = u_i$ otherwise, and $b_j =
v_j$ if $j \in J$ and $b_j = - v_j$ otherwise. Then this matrix is in
$T$ and obeys
\[
M_{ij} = 0
\]
if $(i,j) \in I \times J$ or $(i,j) \in I^c \times J^c$. Note that
on the complement, i.e. $(i,j) \in I \times J^c$ or $(i,j) \in I^c
\times J$, one has $M_{ij} = 2 u_i v_j$ and one can show that
$\mtx{M} \neq 0$ unless $\vct{u}\vct{v}^* = 0$.  Since $\Omega$ is
included in the union of $I \times J$ and $I^c \times J^c$, we have
that $\PO(\vct{M}) = 0$. In conclusion, the invertibility of $\PT
\PO \PT$ implies a fully connected graph.

When the entries are sampled uniformly at random, it is well known
that one needs on the order of $n \log n$ samples to obtain a fully
connected graph with large probability (see,
e.g.,~\cite{BollobasGraphsBook}). Remarkably, Theorem
\ref{teo:rudelson} implies that $\PT \PO \PT$ is invertible---a
stronger property---when the number of samples is also on the order
of $n \log n$.

\section{Proofs of the Critical Lemmas}
\label{sec:proofs}

In this section, we prove the five lemmas of Section
\ref{sec:size}. Before we begin, however, we develop a simple
estimate which we will use throughout. For each pair $(a,b)$ and
$(a',b')$, it follows from the expression of $\PT(\eab)$
\eqref{eq:PTeab} that
\begin{equation}
  \label{eq:PTeabp}
  \<\PT(\eabp), \eab\> = \<\vct{e}_a, \mtx{P}_U \vct{e}_{a'}\> \, 1_{\{b = b'\}} +
  \<\vct{e}_b, \mtx{P}_V \vct{e}_{b'}\> \, 1_{\{a = a'\}} - \<\vct{e}_a, \mtx{P}_U \vct{e}_{a'}\> \<\vct{e}_b, \mtx{P}_V \vct{e}_{b'}\>.
\end{equation}
Fix $\mu_0$ obeying $\mu(U) \le \mu_0$ and $\mu(V) \le \mu_0$ and
note that
\[
|\<\vct{e}_a, \mtx{P}_U \vct{e}_{a'}\>| = |\<\mtx{P}_U \vct{e}_a,
\mtx{P}_U \vct{e}_{a'}\>| \le \|\mtx{P}_U \vct{e}_a\| \, \|\mtx{P}_U
\vct{e}_{a'}\| \le \mu_0 r/n_1
\]
and similarly for $\<\vct{e}_b, \mtx{P}_V \vct{e}_{b'}\>$.  Suppose
that $b = b'$ and $a \neq a'$, then
\[
|\<\PT(\eabp), \eab\>| = |\<\vct{e}_a, \mtx{P}_U \vct{e}_{a'}\>| (1
- \|\mtx{P}_V \vct{e}_b\|^2) \le \mu_0 r/n_1.
\]
We have a similar bound when $a = a'$ and $b \neq b'$ whereas when $a
\neq a'$ and $b \neq b'$,
\[
|\<\PT(\eabp), \eab\>| \le (\mu_0 r)^2/(n_1 n_2).
\]
In short, it follows from this analysis (and from \eqref{eq:PTeabF} for
the case where $(a,b) = (a',b')$) that
\begin{equation}
  \label{eq:useful}
  \max_{ab, a'b'} |\<\PT(\eabp), \eab\>| \le 2\mu_0 r/\min(n_1, n_2).
\end{equation}
A consequence of \eqref{eq:PTeabF} is the estimate:
\begin{align}
  \nonumber \sum_{a'b'} |\<\PT(\eabp), \eab\>|^2 & = \sum_{a'b'}
  |\<\PT(\eab), \eabp\>|^2 \\
  \label{eq:useful2}
  & = \|\PT(\eab)\|_F^2 \le 2\mu_0 r/\min(n_1,n_2),
\end{align}
which we will apply several times. A related estimate is this:
\begin{equation}
  \label{eq:Cauchy}
  \max_a \sum_{b} |E_{ab}|^2 \le \mu_0r/\min(n_1,n_2),
\end{equation}
and the same is true by exchanging the role of $a$ and $b$. To see this, write
\[
\sum_{b} |E_{ab}|^2 = \|\vct{e}_a^* \mtx{E}\|^2 = \|\sum_{j \le r}
\vct{v}_j \<\vct{u}_j, e_a\>\|^2 = \sum_{j\le r} |\<\vct{u}_j,
e_a\>|^2 = \|\mtx{P}_U \vct{e}_a\|^2,
\]
and the conclusion follows from the coherence property.

We will prove the lemmas in the case where $n_1 = n_2 = n$ for
simplicity, i.e. in the case of square matrices of dimension $n$.  The
general case is treated in exactly the same way. In fact, the argument
only makes use of the bounds \eqref{eq:useful}, \eqref{eq:useful2}
(and sometimes \eqref{eq:Cauchy}), and the general case is obtained
by replacing $n$ with $\min(n_1,n_2)$.

Each of the following subsections computes the operator norm of some
random variable.  In each section, we denote $\mtx{S}$ as the
quantity whose norm we wish to analyze. We will also frequently use
the notation $\mtx{H}$ for some auxiliary matrix variable whose norm
we will need to bound.  Hence, we will reuse the same notation many
times rather than introducing a dozens new names---just like in
computer programming where one uses the same variable name in
distinct routines.

\subsection{Proof of Lemma \ref{teo:order0}}
\label{sec:order0}

In this section, we develop a bound on
\begin{align*}
p^{-1} \|\PTc \PO \PT(\mtx{E})\| & = p^{-1} \|\PTc (\PO - p\OpId)
\PT(\mtx{E})\|\\
& \le  p^{-1} \|(\PO - p\OpId)(\mtx{E})\|,
\end{align*}
where the equality follows from $\PTc \PT = 0$, and the inequality
from $\PT(\mtx{E}) = \mtx{E}$ together with $\|\PTc(\mtx{X})\| \le
\|\mtx{X}\|$ which is valid for any matrix $\mtx{X}$. Set
\begin{equation}
\label{eq:S} \mtx{S} \equiv p^{-1} (\PO - p\OpId) (\mtx{E}) = p^{-1}
\sum_{ab} (\delta_{ab} - p) E_{ab} \eab.
\end{equation}
We think of $\mtx{S}$ as a random variable since it depends on the
random $\delta_{ab}$'s, and note that $\E \mtx{S} = 0$.

The proof of Lemma \ref{teo:order0} operates by developing an
estimate on the size of $(\E \|\mtx{S}\|^q)^{1/q}$ for some $q \ge
1$ and by applying Markov inequality to bound the tail of the random
variable $\|\mtx{S}\|$.  To do this, we shall use a symmetrization
argument and the noncommutative Khintchine inequality. Since the
function $f(\mtx{S}) = \|\mtx{S}\|^q$ is convex, Jensen's inequality
gives that
\[
\E \|\mtx{S}\|^q \le \E\|\mtx{S} - \mtx{S}'\|^q,
\]
where $\mtx{S}' = p^{-1} \sum_{ab} (\delta'_{ab} - p) E_{ab} \eab$
is an independent copy of $\mtx{S}$. Since $(\delta_{ab} -
\delta'_{ab})$ is symmetric, $\mtx{S} - \mtx{S}'$ has the same
distribution as
\[
p^{-1} \, \sum_{ab} \epsilon_{ab} (\delta_{ab} - \delta'_{ab})
E_{ab} \eab \equiv  \mtx{S}_\epsilon - \mtx{S}_\epsilon',
\]
where $\{\epsilon_{ab}\}$ is an independent Rademacher sequence and
$\mtx{S}_\epsilon = p^{-1} \sum_{ab} \epsilon_{ab} \delta_{ab} E_{ab}
\vct{e}_a \vct{e}_b^*$. Further, the triangle inequality gives
\[
(\E \|\mtx{S}_\epsilon-\mtx{S}_\epsilon'\|^q)^{1/q} \le (\E
\|\mtx{S}_\epsilon\|^q)^{1/q} + (\E \|\mtx{S}_\epsilon'\|^q)^{1/q} =
2 (\E \|\mtx{S}_\epsilon\|^q)^{1/q}
\]
since $\mtx{S}_\epsilon$ and $\mtx{S}_\epsilon'$ have the same distribution and, therefore,
\[
(\E \|\mtx{S}\|^q)^{1/q} \le 2p^{-1} \left(\E_{\delta} \E_\epsilon
\|\sum_{ab}
    \epsilon_{ab} \delta_{ab} \, E_{ab} \vct{e}_a \vct{e}_b^*\|^q\right)^{1/q}.
\]

We are now in position to apply the noncommutative Khintchine
inequality which bounds the Schatten norm of a Rademacher series. For
$q\geq 1$, the \emph{Schatten q-norm} of a matrix is denoted by
\[
     \|\mtx{X}\|_{S_q} = \left(\sum_{i=1}^n \sigma_i(\mtx{X})^q
     \right)^{1/q}\,.
\]
Note that the nuclear norm is equal to the Schatten 1-norm and the
Frobenius norm is equal to the Schatten 2-norm.  The following
theorem was originally proven by
Lust-Picquard~\cite{LustPicquard86}, and was later sharpened by
Buchholz~\cite{Buchholz01}.

\begin{lemma}[Noncommutative Khintchine inequality] Let $(\mtx{X}_i)_{1 \le
    i \le r}$ be a finite sequence of matrices of the same dimension
  and let $\{\epsilon_i\}$ be a Rademacher sequence. For each $q \ge 2$
\[
\left[\E_\epsilon \norm{\sum_i \epsilon_i \mtx{X}_i}_{S_q}^q
\right]^{1/q} \le C_K \, \sqrt{q} \, \max\left[\norm{\left(\sum_{i}
\mtx{X}_i^*
      \mtx{X}_i\right)^{1/2}}_{S_q}, \norm{\left(\sum_{i} \mtx{X}_i
      \mtx{X}_i^*\right)^{1/2}}_{S_q}\right],
\]
where $C_K = 2^{-1/4} \sqrt{\pi/e}$.
\end{lemma}
For reference, if $\mtx{X}$ is an $n \times n$ matrix and $q \ge
\log n$, we have
\[
\|\mtx{X}\| \le \|\mtx{X}\|_{S_q} \le e \|\mtx{X}\|,
\]
so that the Schatten $q$-norm is within a multiplicative constant from
the operator norm. Observe now that with $q' \ge q$
\[
\left(\E_{\delta} \E_\epsilon \|\mtx{S}_\epsilon\|^q\right)^{1/q}
\le \left(\E_{\delta} \E_\epsilon
\|\mtx{S}_\epsilon\|_{S_{q'}}^q\right)^{1/q} \le \left(\E_{\delta}
\E_\epsilon \|\mtx{S}_\epsilon\|_{S_{q'}}^{q'}\right)^{1/q'}.
\]
We apply the noncommutative Khintchine inequality with $q' \ge \log
n$, and after a little algebra, obtain
\[
\left(\E_{\delta} \E_\epsilon
\|\mtx{S}_\epsilon\|_{S_{q'}}^{q'}\right)^{1/q'} \le C_K \, \frac{e \,
\sqrt{q'}}{p} \, \left( \E_\delta \max \left[\|\sum_{ab}
    \delta_{ab} E^2_{ab} \vct{e}_a \vct{e}_a^*\|^{q'/2}, \|\sum_{ab} \delta_{ab}
    E^2_{ab} \vct{e}_b \vct{e}_b^*\|^{q'/2}\right]\right)^{1/q'}.
\]
The two terms in the right-hand side are essentially the same and if
we can bound any one of them, the same technique will apply to the
other. We consider the first and since $\sum_{ab}
    \delta_{ab} E^2_{ab} \vct{e}_a \vct{e}_a^*$ is a diagonal matrix,
\[
\|\sum_{ab} \delta_{ab} E^2_{ab} \vct{e}_a \vct{e}_a^*\| = \max_{a}
\sum_{b} \delta_{ab} E^2_{ab}.
\]
The following lemma bounds the $q$th moment of this quantity.
\begin{lemma}
  \label{teo:E-max-indep2} Suppose that $q$ is an integer obeying $1
  \le q \le np$ and assume $np \ge 2\log n$. Then
  \begin{equation}
    \label{eq:E-max-indep2}
    \E_\delta \left(\max_{a} \sum_{b}
      \delta_{ab} E^2_{ab}\right)^q \le 2 \, \left(2np \, \|\mtx{E}\|^2_\infty \right)^q.
  \end{equation}
\end{lemma}
The proof of this lemma is in the Appendix. The same estimate applies
to $\E \left(\max_{b} \sum_{a} \delta_{ab} E^2_{ab}\right)^q$ and thus
for each $q \ge 1$
\[
\E_\delta \max \left[\|\sum_{ab} \delta_{ab} E^2_{ab} \vct{e}_a
  \vct{e}_a^*\|^{q}, \|\sum_{ab} \delta_{ab} E^2_{ab} \vct{e}_b
  \vct{e}_b^*\|^{q}\right] \le 4\, \left(2np \, \|\mtx{E}\|^2_\infty
\right)^q.
\]
(In the rectangular case, the same estimate holds with $n = \max(n_1,
n_2)$.)

Take $q = \beta \log n$ for some $\beta \ge 1$, and set $q' = q$. Then
since $\|\mtx{E}\|_\infty \le \mu_1 \sqrt{r}/n$, we established that
\[
\left(\E \|\mtx{S}\|^q\right)^{1/q} \le C \, \frac{1}{p} \,
\sqrt{\beta \log n} \,\sqrt{np} \, \|\mtx{E}\|_\infty  = C \, \mu_1
\, \sqrt{\frac{n r \, \beta \log n}{m}} \equiv K_0.
\]
Then by Markov's inequality, for each $t > 0$,
\[
\P(\|\mtx{S}\| > t K_0) \le t^{-q},
\]
and for $t = e$, we conclude that
\[
\P\left(\|\mtx{S}\| > C e \, \mu_1 \, \sqrt{\frac{nr\, \beta \log
n}{m}}\right) \le n^{-\beta}
\]
with the proviso that $m \ge \max(\beta, 2) \, n \log n$ so that Lemma
\ref{teo:E-max-indep2} holds.

We have not made any assumption in this section about the matrix
$\mtx{E}$ (except that we have a bound on the maximum entry) and,
therefore, have proved the theorem below, which shall be used many
times in the sequel.
\begin{theorem}
  \label{teo:PO} Let $\mtx{X}$ be a fixed $n\times n$ matrix. There is
  a constant $C_0$ such that for each $\beta > 2$
\begin{equation}
  \label{eq:PObound} p^{-1} \|(\PO - p\OpId)(\mtx{X})\|
\le C_0\, \left(\frac{\beta n \log n}{p}\right)^{1/2} \,
\|\mtx{X}\|_\infty
\end{equation}
with probability at least $1 - n^{-\beta}$ provided that
$np \ge \beta \log n$.
\end{theorem}
Note that this is the same $C_0$ described in
Lemma~\ref{teo:order0}.

\subsection{Proof of Lemma \ref{teo:order1}}
\label{sec:order1}

We now need to bound the spectral norm of $\PTc \PO \PT \,
\mathcal{H}(\mtx{E})$ and will use some of the ideas developed in
the previous section. Just as before,
\[
p^{-1} \|\PTc \PO \PT \, \mathcal{H}(\mtx{E})\| \le p^{-1} \|(\PO -
p\OpId) \, \mathcal{H}(\mtx{E})\|,
\]
and put
\[
\mtx{S} \equiv p^{-1} (\PO - p\OpId) \, \mathcal{H}(E) = p^{-2}
\sum_{ab,
  a'b'} \xi_{ab} \xi_{a'b'} \, E_{a'b'} \<\PT\eabp,
\eab\> \eab,
\]
where here and below, $\xi_{ab} \equiv \delta_{ab} - p$.  Decompose
$\mtx{S}$ as \begin{equation}
\label{eq:diagoff} \mtx{S} \,\,\, = \,\,\, p^{-2} \sum_{(a,b) = (a',b')} \,\,\, +
\,\,\, p^{-2} \sum_{(a,b) \, \neq \, (a',b')} \,\,\, \equiv \,\,\,
\mtx{S}_0 + \mtx{S}_1.
\end{equation}
We bound the spectral norm of the diagonal and off-diagonal
contributions separately.

We begin with $\mtx{S}_0$ and decompose $(\xi_{ab})^2$ as
\[
\xi_{ab}^2 = (\dab - p)^2 = (1-2p)(\dab - p) + p(1-p) =
(1-2p)\xi_{ab} + p(1-p),
\]
which allows us to express $\mtx{S}_0$ as
\begin{equation}
\label{eq:S0} \mtx{S}_0 = \frac{1-2p}{p} \sum_{ab} \xiab \, H_{ab}
\eab + (1-p) \sum_{ab} H_{ab} \eab, \quad H_{ab} \equiv p^{-1} \,
E_{ab} \<\PT\eab, \eab\>.
\end{equation}
Theorem \ref{teo:PO} bounds the spectral norm of the first term of the
right-hand side and we have
\[
p^{-1} \|\sum_{ab} \xiab\, H_{ab} \eab\| \le C_0 \, \sqrt{\frac{n^3
\beta \log n}{m}} \, \|\mtx{H}\|_\infty
\]
with probability at least $1 - n^{-\beta}$.  Now since
$\|\mtx{E}\|_\infty \le \mu_1\sqrt{r}/n$ and $|\<\PT\eab, \eab\>| \le
2\mu_0 r/n$ by \eqref{eq:useful}, $\|\mtx{H}\|_\infty \le \mu_0\mu_1(2r/np)
\, \sqrt{r}/n$, and
\[
p^{-1} \| \sum_{ab} \xiab \, H_{ab} \eab\| \le C \mu_0 \mu_1 \,
\frac{nr}{m} \, \sqrt{\frac{ n r \beta \log n}{m}}
\]
with the same probability.  The second term of the right-hand side in
\eqref{eq:S0} is deterministic and we develop an argument that we will
reuse several times. We record a useful lemma.
\begin{lemma}
  \label{teo:lemma0p}
  Let $\mtx{X}$ be a fixed matrix and set $\mtx{Z} \equiv \sum_{ab}
  X_{ab} \<\PT(\eab), \eab\> \eab$. Then
\[
\|\mtx{Z}\| \le \frac{2\mu_0 r}{n} \|\mtx{X}\|.
\]
\end{lemma}
\begin{proof}
  Let $\mtx{\Lambda}_U$ and $\mtx{\Lambda}_V$ be the diagonal matrices
  with entries $\|\mtx{P}_U e_a\|^2$ and $\|\mtx{P}_V e_b\|^2$
  respectively,
\begin{equation}
  \label{eq:Lambda}
  \mtx{\Lambda}_U = \textrm{diag}(\|\mtx{P}_U e_a\|^2), \quad
  \mtx{\Lambda}_V = \textrm{diag}(\|\mtx{P}_V e_b\|^2).
\end{equation}
To bound  the spectral norm of $\mtx{Z}$, observe that it follows
from \eqref{eq:PTeabF2} that
\begin{equation}
  \label{eq:LambdaH}
  \mtx{Z} =  \mtx{\Lambda}_U \mtx{X} + \mtx{X} \mtx{\Lambda}_V -
  \mtx{\Lambda}_U \mtx{X} \mtx{\Lambda}_V =  \mtx{\Lambda}_U \mtx{X} (\mtx{I} - \mtx{\Lambda}_V) + \mtx{X} \mtx{\Lambda}_V.
\end{equation}
Hence, since $\|\mtx{\Lambda}_U\|$ and $\|\mtx{\Lambda}_V\|$ are
bounded by $\min(\mu_0 r/n,1)$ and $\|\mtx{I} - \mtx{\Lambda_V}\| \le
1$, we have
\[
\|\mtx{Z}\| \le \|\mtx{\Lambda}_U\| \|\mtx{X}\| \|\mtx{I} -
\mtx{\Lambda}_V\| + \|\mtx{X}\| \|\mtx{\Lambda}_V\| \le (2\mu_0 r/n)
\|\mtx{X}\|.
\]
\end{proof}
Clearly, this lemma and $\|\mtx{E}\| = 1$ give that $\mtx{H}$ defined
in \eqref{eq:S0} obeys $\|\mtx{H}\| \le 2\mu_0 r/np$.  In summary,
\[
\|\mtx{S}_0\| \le C \frac{nr}{m}\left(\mu_0 \mu_1 \sqrt{\frac{\beta nr\log
      n}{m}} + \mu_0\right)
\]
for some $C > 0$ with the same probability as in Lemma
\ref{teo:order0}.

It remains to bound the off-diagonal term. To this end, we use a
useful decoupling lemma:
\begin{lemma}\cite{delaPena2}
\label{teo:decoupling}
  Let $\{\eta_i\}_{1\le i \le n}$ be a sequence of independent random
  variables, and $\{x_{ij}\}_{i \neq j}$ be elements taken from a Banach
  space. Then
\begin{equation}
  \label{eq:decoupling}
  \P(\|\sum_{i \neq j} \eta_i \eta_j x_{ij}\| \ge t) \le C_D \P(\|\sum_{i \neq j} \eta_i \eta'_j x_{ij}\| > t/C_D),
\end{equation}
where $\{\eta'_i\}$ is an independent copy of $\{\eta_i\}$.
\end{lemma}
This lemma asserts that it is sufficient to estimate
$\P(\|\mtx{S}_1'\| \ge t)$ where $\mtx{S}'_1$ is given by
\begin{equation}
  \label{eq:Sp1}
  \mtx{S}'_1 \equiv p^{-2} \sum_{ab \neq
    a'b'} \xiab \xiabpp \, E_{a'b'} \<\PT\eabp,
  \eab\> \eab
\end{equation}
in which $\{\xi'_{ab}\}$ is an independent copy of $\{\xiab\}$. We write
$\mtx{S}'_1$ as
\begin{equation}
\label{eq:Sp1b} \mtx{S}'_1 = p^{-1} \sum_{ab} \xiab\, H_{ab} \eab,
\qquad H_{ab} \equiv p^{-1} \sum_{a'b': (a',b') \neq (a,b)} \xiabpp\,
E_{a'b'} \<\PT\eabp, \eab\>.
\end{equation}
To bound the tail of $\|\mtx{S}'_1\|$, observe that
\[
\P(\|\mtx{S}'_1\| \ge t) \le \P(\|\mtx{S}'_1\| \ge t \,\, | \,\,
\|\mtx{H}\|_{\infty} \le K) + \P(\|\mtx{H}\|_{\infty} > K).
\]
By independence, the first term of the right-hand side is bounded by
Theorem \ref{teo:PO}. On the event $\{\|\mtx{H}\|_\infty \le K\}$,
we have
\[
p^{-1} \|\sum_{ab} \xiab \, H_{ab} \eab\| \le C \, \sqrt{\frac{n^3
    \beta \log n}{m}} \, K.
\]
with probability at least $1 - n^{-\beta}$.  To bound
$\|\mtx{H}\|_\infty$, we use Bernstein's inequality.
\begin{lemma}
  \label{teo:bernstein}
  Let $\mtx{X}$ be a fixed matrix and define $\mathcal{Q}(\mtx{X})$ as
  the matrix whose $(a,b)$th entry is
\[
[\mathcal{Q}(\mtx{X})]_{ab} = p^{-1} \sum_{a'b' : (a',b') \neq
  (a,b)} (\delta_{a'b'} - p) \, X_{a'b'} \<\PT\eabp, \eab\>,
  \]
  where $\{\delta_{ab}\}$ is an independent Bernoulli sequence obeying
  $\P(\delta_{ab} = 1) = p$. Then
      \begin{equation}
  \label{eq:Hinf}
  \P\left(\|\mathcal{Q}(\mtx{X})\|_\infty > \lambda \, \sqrt{\frac{\mu_0 r}{np}} \, \|\mtx{X}\|_\infty\right) \le
  2 n^2\, \exp\left(-\frac{\lambda^2}{2 + \frac{2}{3}  \sqrt{\frac{\mu_0 r}{np}} \lambda}\right).
\end{equation}
With $\lambda = \sqrt{3\beta \log n}$, the right-hand side is bounded
by $2n^{2-\beta}$ provided that $np \ge \frac{4\beta}{3} \mu_0 r\log
n$. In particular, for $\lambda = \sqrt{6\beta \log n}$ with $\beta >
2$, the bound is less than $2n^{-\beta}$ provided that $np \ge
\frac{8\beta}{3} \mu_0 r\log n$.
\end{lemma}
\begin{proof}
  The inequality \eqref{eq:Hinf} is an application of Bernstein's
  inequality, which states that for a sum of uniformly bounded
  independent zero-mean random variables obeying $|Y_k| \le c$,
\begin{equation}
  \label{eq:bern2}
  \P\left(|\sum_{k = 1}^n Y_k| > t\right) \le 2e^{-t^2/(2\sigma^2 + 2ct/3)},
\end{equation}
where $\sigma^2$ is the sum of the variances, $\sigma^2 \equiv \sum_{k
  = 1}^n \textrm{Var}(Y_k)$. We have
\begin{align*}
  \textrm{Var}([\mathcal{Q}(\mtx{X})]_{ab}) & = \frac{1-p}{p} \sum_{a'b': (a',b')
    \neq (a,b)} |X_{a'b'}|^2 |\<\PT\eabp, \eab\>|^2\\
  & \le \frac{1-p}{p} \, \|\mtx{X}\|_\infty^2 \sum_{a'b' : (a',b') \neq (a,b)}
  |\<\PT\eab, \eabp\>|^2 \le \frac{1-p}{p} \, \|\mtx{X}\|_\infty^2 \, 2\mu_0 r/n
\end{align*}
by \eqref{eq:useful2}.  Also,
\[
p^{-1} \, \left|(\delta_{a'b'} - p) X_{a'b'} \<\PT\eabp, \eab\>
\right| \le p^{-1} \, \|\mtx{X}\|_\infty \, 2\mu_0 r/n
\]
and hence, for each $t > 0$,
\eqref{eq:bern2} gives
\begin{equation}
  \label{eq:bern3}
  \P(|[\mathcal{Q}(\mtx{X})]_{ab}| > t) \le 2 \exp\left(-\frac{t^2}{2\frac{\mu_0 r}{np}
      \|\mtx{X}\|_\infty^2 + \frac{2}{3} \frac{\mu_0 r}{np} \|\mtx{X}\|_\infty t}\right).
\end{equation}
Putting $t = \lambda \sqrt{\mu_0 r/np} \|\mtx{X}\|_\infty$ for some
$\lambda > 0$ and applying the union bound gives \eqref{eq:Hinf}.
\end{proof}

Since $\|\mtx{E}\|_\infty \le \mu_1 \sqrt{r}/n$ it follows that
$\mtx{H}=\mathcal{Q}(\mtx{E})$ introduced in \eqref{eq:Sp1b} obeys
\[
\|\mtx{H}\|_\infty \le C \, \frac{\mu_1 \sqrt{r}}{n} \,
\sqrt{\frac{\mu_0 nr \beta \log n}{m}}
\]
with probability at least $1 - 2n^{-\beta}$ for each $\beta > 2$ and,
therefore,
\[
\|\mtx{S}'_1\| \le  C \, \sqrt{\mu_0} \mu_1 \frac{nr \beta\log n}{m}
\]
with probability at least $1 - 3n^{-\beta}$. In conclusion, we have
\begin{equation}
  \label{eq:important2}
  p^{-1} \|(\PO - p\OpId) \, \mathcal{H}(\mtx{E})\| \le C
  \frac{nr}{m}\left(\sqrt{\mu_0} \mu_1 \left(\sqrt{\frac{\mu_0 nr\beta \log n}{m}} + \beta \log n\right) + \mu_0 \right)
\end{equation}
with probability at least $1 - (1+3C_D)n^{-\beta}$. A simple algebraic
manipulation concludes the proof of Lemma \ref{teo:order1}. Note that
we have not made any assumption about the matrix $\mtx{E}$ and,
therefore, established the following:
\begin{lemma}
  \label{teo:PO2} Let $\mtx{X}$ be a fixed $n\times n$ matrix. There is
  a constant $C'_0$ such that 
\begin{equation}
  \label{eq:PObound2} p^{-2} \|\sum_{(a,b) \neq (a',b')} \xiab \xiabp X_{ab} \<\PT(\eabp), \eab\> \eab\| \le C'_0 \frac{\sqrt{\mu_0 r} \, \beta \log n}{p} \|\mtx{X}\|_\infty
\end{equation}
with probability at least $1 - O(n^{-\beta})$ for all $\beta > 2$
provided that $np \ge 3 \mu_0 r \beta \log n$.
\end{lemma}

\subsection{Proof of Lemma \ref{teo:order2}}
\label{sec:order2}

To prove Lemma \ref{teo:order2}, we need to bound the spectral norm
of $p^{-1} \, (\PO - p\OpId) \, \mathcal{H}^2(\mtx{E})$, a matrix
given by
\[
p^{-3} \sum_{a_1 b_1, a_2 b_2, a_3 b_3} \xi_{a_1 b_1} \xi_{a_2 b_2}
\xi _{a_3 b_3} E_{a_3 b_3} \<\PT \eabthree, \eabtwo\> \<\PT \eabtwo,
\eabone\> \eabone,
\]
where $\xi_{ab} = \delta_{ab} - p$ as before. It is convenient to
introduce notations to compress this expression. Set $\omega =
(a,b)$ (and $\omega_i = (a_i, b_i)$ for $i = 1, 2, 3$),
$\mtx{F}_\omega = \eab$, and $P_{\omega' \omega} = \<\PT \eabp,
\eab\>$ so that
\[
p^{-1} \, (\PO - p\OpId) \, \mathcal{H}^2(\mtx{E}) = p^{-3}
\sum_{\omega_1, \omega_2, \omega_3} \xi_{\omega_1} \xi_{\omega_2}
\xi _{\omega_3} \, E_{\omega_3} P_{\omega_3 \omega_2} P_{\omega_2
  \omega_1} \mtx{F}_{\omega_1}.
\]
Partition the sum depending on whether some of the $\omega_i$'s
are the same or not
\begin{equation}
  \label{eq:equiv}
  \frac{1}{p}(\PO - p\OpId)\mathcal{H}^2(\mtx{E}) = \frac{1}{p^3}
\left[\sum_{\omega_1 = \omega_2 = \omega_3} +
    \sum_{\omega_1 \neq \omega_2 = \omega_3}  +  \sum_{\omega_1 = \omega_3
      \neq \omega_2}  +  \sum_{\omega_1 = \omega_2 \neq \omega_3}
     +  \sum_{\omega_1 \neq \omega_2 \neq \omega_3}\right].
\end{equation}
The meaning should be clear; for instance, the sum $\sum_{\omega_1
  \neq \omega_2 = \omega_3}$ is the sum over the $\omega$'s such that
$\omega_2 = \omega_3$ and $\omega_1 \neq \omega_2$. Similarly,
$\sum_{\omega_1 \neq \omega_2 \neq \omega_3}$ is the sum over the
$\omega$'s such that they are all distinct. The idea is now to use a
decoupling argument to bound each sum in the right-hand side of
\eqref{eq:equiv} (except for the first which does not need to be
decoupled) and show that all terms are appropriately small in the
spectral norm.

We begin with the first term which is equal to
\begin{equation}
  \label{eq:onetwothree}
  \frac{1}{p^{3}} \, \sum_\omega (\xi_{\omega})^3 \, E_{\omega} P_{\omega
    \omega}^2 \mtx{F}_{\omega} = \frac{1-3p+3p^2}{p^3} \, \sum_\omega
  \xi_{\omega} \, E_{\omega} P_{\omega \omega}^2 \mtx{F}_{\omega} +
  \frac{1-3p+2p^2}{p^2} \sum_\omega E_{\omega} P_{\omega \omega}^2
  \mtx{F}_{\omega},
\end{equation}
where we have used the identity
\[
(\xi_\omega)^3 = (1-3p+3p^2)\xi_\omega + p(1-3p+2p^2).
\]
Set $H_{\omega} = E_\omega (p^{-1} P_{\omega \omega})^2$.  For the
first term in the right-hand side of \eqref{eq:onetwothree}, we need
to control $\| \sum_\omega \xi_{\omega} \, H_\omega
\mtx{F}_{\omega}\|$.  This is easily bounded by Theorem
\ref{teo:PO}. Indeed, it follows from
\[
|H_{\omega}| \le \left(\frac{2\mu_0 r}{np}\right)^2 \,
\|\mtx{E}\|_\infty
\]
that for each $\beta > 0$,
\[
p^{-1} \| \sum_\omega \xi_{\omega} \, H_{\omega} \, \mtx{F}_{\omega}\|
\le C \, \left(\frac{\mu_0 nr}{m}\right)^2 \, \mu_1 \sqrt{\frac{ nr
    \beta \log n}{m}} = C \, \mu_0^2 \mu_1 \sqrt{\beta \log n} \,
\left(\frac{nr}{m}\right)^{5/2}
\]
with probably at least $1 - n^{-\beta}$.  For the second term in the
right-hand side of \eqref{eq:onetwothree}, we apply Lemma
\ref{teo:lemma0p} which gives
\[
\|\sum_{\omega} E_\omega P_{\omega \omega}^2 \mtx{F}_\omega\| \le (2\mu_0 r/n)^2
\]
so that $\|\mtx{H}\| \le (2\mu_0r/np)^2$. In conclusion, the first
term in \eqref{eq:equiv} has a spectral norm which is bounded by
\[
C \, \left(\frac{nr}{m}\right)^2 \, \left(\mu_0^2 \mu_1 \left(\frac{nr
      \beta \log n}{m}\right)^{1/2} + \mu_0^2\right)
\]
with probability at least $1 - n^{-\beta}$.

We now turn our attention to the second term which can be written as
\begin{multline*}
 p^{-3}   \sum_{\omega_1 \neq \omega_2}
\xi_{\omega_1} (\xi_{\omega_2})^2 \, E_{\omega_2} P_{\omega_2
\omega_2} P_{\omega_2 \omega_1} \mtx{F}_{\omega_1} =
  \frac{1-2p}{p^3} \sum_{\omega_1 \neq \omega_2} \xi_{\omega_1}
  \xi_{\omega_2} \, E_{\omega_2} P_{\omega_2 \omega_2} P_{\omega_2
    \omega_1} \mtx{F}_{\omega_1} \\
+ \frac{1-p}{p^2} \sum_{\omega_1 \neq
    \omega_2} \xi_{\omega_1}  \, E_{\omega_2}
  P_{\omega_2 \omega_2} P_{\omega_2 \omega_1} \mtx{F}_{\omega_1}.
\end{multline*}
Put $\mtx{S}_1$ for the first term; bounding $\|\mtx{S}_1\|$ is a
simple application of Lemma \ref{teo:PO2} with $X_\omega = p^{-1}
E_\omega P_{\omega \omega}$, which gives
\[
\|\mtx{S}_1\| \le C\, \mu_0^{3/2} \mu_1 \, (\beta \log n) \,
\left(\frac{nr}{m}\right)^{2}
\]
since $\|\mtx{E}\|_\infty \le \mu_1 \sqrt{r}/n$.  For the second term,
we need to bound the spectral norm of $\mtx{S}_2$ where
\[
\mtx{S}_2 \equiv  p^{-1} \sum_{\omega_1} \xi_{\omega_1}
H_{\omega_1} \mtx{F}_{\omega_1}, \qquad H_{\omega_1} = p^{-1}
\sum_{\omega_2 : \omega_2 \neq \omega_1} E_{\omega_2}
  P_{\omega_2 \omega_2} P_{\omega_2 \omega_1}.
\]
Note that $\mtx{H}$ is deterministic. The lemma below provides an
estimate about $\|\mtx{H}\|_\infty$.
\begin{lemma}
\label{teo:Hp} The matrix $\mtx{H}$ obeys
\begin{equation}
  \label{eq:Hp}
  \|\mtx{H}\|_\infty \le
  \frac{\mu_0 r}{np}\left(3\|\mtx{E}\|_\infty + 2\frac{\mu_0 r}{n}\right).
\end{equation}
\end{lemma}
\begin{proof}
We begin by rewriting $\mtx{H}$ as
\[
pH_\omega = \sum_{\omega'} E_{\omega'} P_{\omega' \omega'}
P_{\omega'
  \omega} - E_\omega P^2_{\omega \omega}.
\]
Clearly, $|E_\omega P^2_{\omega \omega}| \le (\mu_0 r/n)^2
\|\mtx{E}\|_\infty$ so that it suffices to bound the first term,
which is the $\omega$th entry of the matrix
\[
\sum_{\omega, \omega'} E_{\omega'} P_{\omega' \omega'} P_{\omega'
  \omega} \mtx{F}_\omega = \PT(\mtx{\Lambda}_U \mtx{E} + \mtx{E}
\mtx{\Lambda}_V - \mtx{\Lambda}_U \mtx{E} \mtx{\Lambda}_V).
\]
Now it is immediate to see that $\mtx{\Lambda}_U \mtx{E} \in T$ and
likewise for $\mtx{E} \mtx{\Lambda}_V$. Hence,
\begin{align*}
  \| \PT(\mtx{\Lambda}_U \mtx{E} + \mtx{E} \mtx{\Lambda}_V -
  \mtx{\Lambda}_U \mtx{E} \mtx{\Lambda}_V)\|_\infty & \le \|\mtx{\Lambda}_U
  \mtx{E}\|_\infty + \|\mtx{E} \mtx{\Lambda}_V\|_\infty + \|\PT(\mtx{\Lambda}_U
  \mtx{E} \mtx{\Lambda}_V)\|_\infty \\ & \le 2 \|\mtx{E}\|_\infty
  \mu_0r/n + \|\PT(\mtx{\Lambda}_U \mtx{E} \mtx{\Lambda}_V)\|_\infty.
\end{align*}
We finally use the crude estimate
\[
\|\PT(\mtx{\Lambda}_U \mtx{E} \mtx{\Lambda}_V)\|_\infty \le
\|\PT(\mtx{\Lambda}_U \mtx{E} \mtx{\Lambda}_V)\| \le 2
\|\mtx{\Lambda}_U \mtx{E} \mtx{\Lambda}_V\| \le 2(\mu_0 r/n)^2
\]
to complete the proof of the lemma.
\end{proof}

As a consequence of this lemma, Theorem \ref{teo:PO} gives
\[
\|\mtx{S}_2\| \le C \, \sqrt{\beta \log n} \,
\left(\frac{nr}{m}\right)^{3/2}(\mu_0 \mu_1 + \mu_0^2 \sqrt{r})
\]
with probability at least $1 - n^{-\beta}$.  In conclusion, the second
term in \eqref{eq:equiv} has spectral norm bounded by
\[
C \, \sqrt{\beta \log n} \left(\frac{nr}{m}\right)^{3/2}\left(\mu_0
  \mu_1 \sqrt{\frac{\mu_0 nr \beta \log n}{m}} + \mu_0 \mu_1 + \mu_0^2
  \sqrt{r}\right)
\]
with probability at least $1 - O(n^{-\beta})$.

We now examine the third term which can be written as
\begin{multline*}
p^{-3}   \sum_{\omega_1 \neq \omega_2} (\xi_{\omega_1})^2
\xi_{\omega_2} \, E_{\omega_1} P_{\omega_1 \omega_2} P_{\omega_2
\omega_1} \mtx{F}_{\omega_1} = \frac{1-2p}{p^3} \sum_{\omega_1 \neq
\omega_2} \xi_{\omega_1}
\xi_{\omega_2} \, E_{\omega_1} P^2_{\omega_2 \omega_1} \mtx{F}_{\omega_1} \\
+ \frac{1-p}{p^2} \sum_{\omega_1 \neq \omega_2} \xi_{\omega_2} \,
E_{\omega_1} P^2_{\omega_2 \omega_1} \mtx{F}_{\omega_1}.
\end{multline*}
We use the decoupling argument once more so that for the first term of
the right-hand side, it suffices to estimate the tail of the norm of
\[
\mtx{S}_1 \equiv p^{-1} \sum_{\omega_1} \xi^{(1)}_{\omega_1}
E_{\omega_1} H_{\omega_1} \mtx{F}_{\omega_1}, \qquad H_{\omega_1}
\equiv p^{-2} \sum_{\omega_2
  : \omega_2 \neq \omega_1} \xi^{(2)}_{\omega_2} \, P^2_{\omega_2
  \omega_1},
\]
where $\{\xi^{(1)}_\omega\}$ and $\{\xi^{(2)}_\omega\}$ are
independent copies of $\{\xi_\omega\}$. It follows from Bernstein's
inequality and the estimates
\[
|P_{\omega_2 \omega_1}| \le 2\mu_0 r/n
\]
and
\[
\sum_{\omega_2 : \omega_2 \neq \omega_1} |P_{\omega_2 \omega_1}|^4 \le
\max_{\omega_2 : \omega_2 \neq \omega_1} |P_{\omega_2 \omega_1}|^2 \,
\sum_{\omega_2 : \omega_2 \neq \omega_1} |P_{\omega_2 \omega_1}|^2 \le
\left(\frac{2\mu_0 r}{n}\right)^2 \,\frac{2\mu_0 r}{n}
\]
that for each $\lambda > 0$,\footnote{We would like to remark that
one can often get better estimates; when $\omega_1 \neq \omega_2$,
the bound $|P_{\omega_2 \omega_1}| \le 2\mu_0 r/n$ may be rather
crude. Indeed, one can derive better estimates for the random
orthogonal model, for example.}
\[
\P\left(|H_{\omega_1}| > \lambda \left(\frac{2\mu_0
      r}{np}\right)^{3/2}\right) \le 2 \exp\left(-\frac{\lambda^2}{2
      + \frac{2}{3} \lambda \left(\frac{2\mu_0
      r}{np}\right)^{1/2}}\right).
\]
It is now not hard to see that this inequality implies that
\[
P\left(\|\mtx{H}\|_\infty > \sqrt{8\beta \log n} \,
  \left(\frac{2 \mu_0 nr}{m}\right)^{3/2}  \right) \le 2 \,
n^{-2\beta + 2}
\]
provided that $m \ge \frac{16}{9} \mu_0 nr \, \beta\log n$.  As a
consequence, for each $\beta > 2$, Theorem \ref{teo:PO} gives
\[
\|\mtx{S}_1\| \le C \, \mu_0^{3/2} \mu_1 \, \beta \log n \left(\frac{nr}{m}\right)^{2}
\]
with probability at least $1-3n^{-\beta}$.  The other term is equal to
$(1-p)$ times $\sum_{\omega_1} E_{\omega_1} H_{\omega_1}
\mtx{F}_{\omega_1}$, and
\begin{align*}
  \|\sum_{\omega_1} E_{\omega_1} H_{\omega_1} \mtx{F}_{\omega_1}\| &
  \le \|\sum_{\omega_1} E_{\omega_1} H_{\omega_1}
  \mtx{F}_{\omega_1}\|_F \\ & \le \|\mtx{H}\|_\infty \|\mtx{E}\|_F \le C \,
  {\sqrt{\beta \log n}} \left(\frac{\mu_0 nr}{m}\right)^{3/2} \, \sqrt{r}.
\end{align*}
In conclusion, the third term in \eqref{eq:equiv} has spectral norm
bounded by
\[
C \, \mu_0\, \sqrt{\beta \log n} \left(\frac{nr}{m}\right)^{3/2}
\left(\mu_1 \sqrt{\frac{\mu_0 n r \beta \log n}{m}} +
  \sqrt{\mu_0 r}\right)
\]
with probability at least $1-O(n^{-\beta})$.

We proceed to the fourth term which can be written as
\begin{multline*}
  p^{-3}   \sum_{\omega_1 \neq \omega_3}
(\xi_{\omega_1})^2 \xi_{\omega_3} \, E_{\omega_3} P_{\omega_3
\omega_1} P_{\omega_1 \omega_1} \mtx{F}_{\omega_1} =
\frac{1-2p}{p^3} \sum_{\omega_1 \neq \omega_3} \xi_{\omega_1}
\xi_{\omega_3} \, E_{\omega_3} P_{\omega_3 \omega_1} P_{\omega_1 \omega_1} \mtx{F}_{\omega_1} \\
+ \frac{1-p}{p^2} \sum_{\omega_1 \neq \omega_3} \xi_{\omega_3} \,
E_{\omega_3} P_{\omega_3 \omega_1} P_{\omega_1 \omega_1}
\mtx{F}_{\omega_1}.
\end{multline*}
Let $\mtx{S}_1$ be the first term and set $H_{\omega_1} = p^{-2}
\sum_{\omega_1 \neq \omega_3} \xi_{\omega_1} \xi_{\omega_3} \,
E_{\omega_3} P_{\omega_3 \omega_1} \mtx{F}_{\omega_1}$. Then Lemma
\ref{teo:lemma0p} gives
\[
\|\mtx{S}_1\| \le \frac{2\mu_0 r}{np} \|\mtx{H}\| \le C \, \mu_0^{3/2}
\mu_1\, (\beta \log n) \, \left(\frac{nr}{m}\right)^{2}
\]
where the last inequality is given by Lemma \ref{teo:PO2}.  For the
other term---call it $\mtx{S}_2$---set $H_{\omega_1} = p^{-1}
\sum_{\omega_3 : \omega_3 \neq \omega_1} \xi_{\omega_3} \,
E_{\omega_3} P_{\omega_3 \omega_1}$. Then Lemma \ref{teo:lemma0p}
gives
\[
\|\mtx{S}_2\| \le \frac{2\mu_0 r}{np} \|\mtx{H}\|.
\]
Notice that $H_{\omega_1} = p^{-1} \sum_{\omega_3} \xi_{\omega_3} \,
E_{\omega_3} P_{\omega_3 \omega_1} - p^{-1} \xi_{\omega_1}
E_{\omega_1} P_{\omega_1 \omega_1}$ so that with $G_{\omega_1} =
E_{\omega_1} P_{\omega_1 \omega_1}$
\[
\mtx{H} = p^{-1} [\PT(\PO - p\OpId)(\mtx{E}) - (\PO -
p\OpId)(\mtx{G})].
\]
Now for any matrix $\mtx{X}$, $\|\PT(\mtx{X})\| = \|\mtx{X} -
\PTc(\mtx{X})\| \le 2 \|\mtx{X}\|$ and, therefore,
\[
\|\mtx{H}\| \le 2p^{-1} \|(\PO - p\OpId)(\mtx{E})\| + p^{-1} \|(\PO -
p\OpId)(\mtx{G})\|.
\]
As a consequence and since $\|\mtx{G}\|_\infty \le
\|\mtx{E}\|_\infty$, Theorem \ref{teo:PO} gives for each $\beta > 2$,
\[
\|\mtx{H}\| \le C \mu_1 \sqrt{\frac{nr \beta \log n}{m}}
\]
with probability at least $1-n^{-\beta}$. In conclusion, the fourth
term in \eqref{eq:equiv} has spectral norm bounded by
\[
C \, \mu_0 \mu_1 \sqrt{\beta \log n} \,
\left(\frac{nr}{m}\right)^{3/2} \left(\sqrt{\frac{\mu_0 nr \beta \log
      n}{m}} + 1\right)
\]
with probability at least $1-O(n^{-\beta})$.

We finally examine the last term
\[
p^{-3} \sum_{\omega_1 \neq \omega_2 \neq \omega_3} \xi_{\omega_1}
\xi_{\omega_2} \xi _{\omega_3} \, E_{\omega_3} P_{\omega_3
\omega_2} P_{\omega_2 \omega_1} \mtx{F}_{\omega_1}.
\]
Now just as one has a decoupling inequality for pairs of variables, we
have a decoupling inequality for triples as well and we thus simply
need to bound the tail of
\[
\mtx{S}_1 \equiv p^{-3} \sum_{\omega_1 \neq \omega_2 \neq \omega_3}
\xi^{(1)}_{\omega_1} \xi^{(2)}_{\omega_2} \xi^{(3)}_{\omega_3} \,
E_{\omega_3} P_{\omega_3 \omega_2} P_{\omega_2 \omega_1}
\mtx{F}_{\omega_1}
\]
in which the sequences $\{\xi_\omega^{(1)}\}$, $\{\xi_\omega^{(2)}\}$
and $\{\xi_\omega^{(3)}\}$ are independent copies of $\{\xi_\omega\}$.
We refer to \cite{delaPena2} for details. We now argue as in Section
\ref{sec:order1} and write $\mtx{S}_1$ as
\[
\mtx{S}_1 = p^{-1} \sum_{\omega_1} \xi^{(1)}_{\omega_1} H_{\omega_1}
\mtx{F}_{\omega_1},
\]
where
\begin{equation}
\label{eq:forlater} H_{\omega_1} \equiv p^{-1} \sum_{\omega_2 :
\omega_2 \neq \omega_1} \xi^{(2)}_{\omega_2} \, G_{\omega_2} \,
P_{\omega_2 \omega_1}, \qquad G_{\omega_2} \equiv p^{-1}
\sum_{\omega_3 : \omega_3 \neq \omega_1,
  \omega_3 \neq \omega_2} \xi^{(3)}_{\omega_3} \,
E_{\omega_3} \, P_{\omega_3 \omega_2}.
\end{equation}
By Lemma \ref{teo:bernstein}, we have for each $\beta > 2$
\[
\|\mtx{G}\|_\infty \le C \, \sqrt{\frac{\mu_0 nr \beta \log n}{m}} \,
\|\mtx{E}\|_\infty
\]
with large probability and the same argument then gives
\[
\|\mtx{H}\|_\infty \le C\, \sqrt{\frac{\mu_0 nr \beta \log n}{m}} \,
\|\mtx{G}\|_\infty \le C \, \frac{\mu_0 nr \beta \log n}{m} \, \|\mtx{E}\|_\infty
\]
with probability at least $1 - 4n^{-\beta}$. As a consequence, Theorem
\ref{teo:PO} gives
\[
\| \mtx{S} \| \le C \, \mu_0 \mu_1 \, \left(\frac{nr \beta \log
    n}{m}\right)^{3/2}
\]
with probability at least $1 - O(n^{-\beta})$.

To summarize the calculations of this section and using the fact that
$\mu_0 \ge 1$ and $\mu_1 \le \mu_0 \sqrt{r}$, we have established that
if $m \ge \mu_0 \, nr (\beta \log n)$,
\begin{multline*}
  p^{-1} \|(\PO -
  p\OpId) \, \mathcal{H}^2(\mtx{E})\| \le C \left(\frac{nr}{m}\right)^2\left(\mu_0^2 \mu_1 \sqrt{\frac{nr\beta \log n}{m}} + \mu_0^2\right)\\
  + C \sqrt{\beta \log n}
  \left(\frac{nr}{m}\right)^{3/2} \mu_0^2 \sqrt{r}
  + C \left(\frac{nr\beta \log n}{m}\right)^{3/2} \mu_0 \mu_1
\end{multline*}
with probability at least $1 - O(n^{-\beta})$.  One can check that if
$m = \lambda \, \mu_0^{4/3} n r^{4/3} \beta \log n$ for a fixed $\beta
\ge 2$ and $\lambda \ge 1$, then there is a constant $C$ such that
\[
\|p^{-1} \, (\PO - p\OpId) \, \mathcal{H}^2(\mtx{E})\| \le C
\lambda^{-3/2}
\]
with probability at least $1 - O(n^{-\beta})$. This is the content of
Lemma \ref{teo:order2}.

\subsection{Proof of Lemma \ref{teo:order3}}
\label{sec:improved}

\newcommand{\aij}{{\<\vct{e}_i, \mtx{F}_\omega \vct{v}_j\>}}
\newcommand{\aijone}{{\<\vct{e}_i, \mtx{F}_{\omega_1} \vct{v}_j\>}}

Clearly, one could continue on the same path and estimate the
spectral norm of $p^{-1} (\PO - p\OpId) \, \mathcal{H}^3(\mtx{E})$
by the same technique as in the previous sections. That is to say,
we would write
\[
p^{-1} (\PO - p\OpId) \, \mathcal{H}^3(\mtx{E}) = p^{-4}
\sum_{\omega_1,
  \omega_2, \omega_3, \omega_4} \, \left[ \prod_{i = 1}^4 \xi_{\omega_i} \right] \,
E_{\omega_4} \, \left[\prod_{i = 1}^3 P_{\omega_{i+1}
\omega_i}\right] \, \mtx{F}_{\omega_1}
\]
with the same notations as before, and partition the sum depending
on whether some of the $\omega_i$'s are the same or not. Then we
would use the decoupling argument to bound each term in the sum.
Although this is a clear possibility, one would need to consider 18
cases and the calculations would become a little laborious. In this
section, we propose to bound the term $p^{-1} (\PO - p\OpId) \,
\mathcal{H}^3(\mtx{E})$ with a different argument which has two main
advantages: first, it is much shorter and second, it uses much of
what we have already established. The downside is that it is not as
sharp.

The starting point is to note that
\[
p^{-1} (\PO - p\OpId) \, \mathcal{H}^3(\mtx{E}) = p^{-1} (\Xi \circ
{\cal
  H}^3(\mtx{E})),
\]
where $\Xi$ is the matrix with i.i.d.~entries equal to $\xi_{ab} =
\delta_{ab} - p$ and $\circ$ denotes the Hadamard product
(componentwise multiplication).  To bound the spectral norm of this
Hadamard product,  we apply an inequality due to Ando, Horn, and
Johnson~\cite{Ando87}.  An elementary proof can be found in \S 5.6
of ~\cite{HJ2}.
\begin{lemma}\cite{HJ2}
  \label{teo:HJ} Let $\mtx{A}$ and $\mtx{B}$ be two $n_1 \times n_2$
  matrices. Then
\begin{equation}
\label{eq:HJ} \|\mtx{A} \circ \mtx{B}\| \le \|\mtx{A}\| \,
\nu(\mtx{B}),
\end{equation}
where $\nu$ is the function
\[
\nu(\mtx{B}) = \inf \{c(\mtx{X}) c(\mtx{Y}) : \mtx{X}\mtx{Y}^* =
\mtx{B}\},
\]
and $c(\mtx{X})$ is the maximum Euclidean norm of the rows
\[
c(\mtx{X})^2 = \max_{1 \le i \le n} \sum_{j} X_{ij}^2.
\]
\end{lemma}

To apply \eqref{eq:HJ}, we first notice that one can estimate the
norm of $\Xi$ via Theorem \ref{teo:PO}. Indeed, let $\mtx{Z} = {\bf
1} {\bf 1}^*$ be the matrix with all entries equal to one. Then
$p^{-1} \Xi = p^{-1} (\PO - p\OpId)(\mtx{Z})$ and thus
\begin{equation}
  \label{eq:Xi}
  p^{-1} \|\Xi\| \le C\, \left(\frac{n^3\beta \log n}{m}\right)^{1/2}
\end{equation}
with probability at least $1-n^{-\beta}$. One could obtain a similar
result by appealing to the recent literature on random matrix theory
and on concentration of measure. Potentially this could allow to
derive an upper bound without the logarithmic term but we will not
consider these refinements here. (It is interesting to note in
passing, however, that the two page proof of Theorem \ref{teo:PO}
gives a large deviation result about the largest singular value of a
matrix with i.i.d.~entries which is sharp up to a multiplicative
factor proportional to at most $\sqrt{\log n}$.)

Second, we bound the second factor in \eqref{eq:HJ} via the following
estimate:
\begin{lemma}
  \label{teo:HE}
  There are numerical constants $C$ and $c$ so that for each
  $\beta > 2$, $\mathcal{H}^3(\mtx{E})$ obeys
\begin{equation}
  \label{eq:HE} \nu(\mathcal{H}^3(\mtx{E})) \le C \mu_0 r/n
\end{equation}
with probability at least $1-O(n^{-\beta})$ provided that $m \ge c
\, \mu_0^{4/3} \, nr^{5/3} (\beta \log n)$.
\end{lemma}
The two inequalities \eqref{eq:Xi} and \eqref{eq:HE} give
\[
p^{-1} \|\Xi \circ \mathcal{H}^3(\mtx{E})\| \le C \,
\sqrt{\frac{\mu_0^2 \, nr^2 \, \beta \log
    n}{m}},
\]
with large probability. Hence, when $m$ is substantially larger than a
constant times $\mu_0^2 n r^{2} (\beta \log n)$, we have that the spectral
norm of $p^{-1}(\PO - p\OpId) \, \mathcal{H}^3(\mtx{E})$ is much less
than $1$. This is the content of Lemma \ref{teo:order3}.

The remainder of this section proves Lemma \ref{teo:HE}. Set $\mtx{S}
\equiv \mathcal{H}^3(\mtx{E})$ for short. Because $\mtx{S}$ is in $T$,
$\mtx{S} = \PT(\mtx{S}) = \mtx{P}_U \mtx{S} + \mtx{S} \mtx{P}_V -
\mtx{P}_U \mtx{S} \mtx{P}_V$. Writing $\mtx{P}_U = \sum_{j = 1}^r
\vct{u}_j \vct{u}_j^*$ and similarly for $\mtx{P}_V$ gives
\[
\mtx{S} = \sum_{j = 1}^r \vct{u}_j (\vct{u}_j^* \mtx{S}) + \sum_{j =
1}^r ((I - \mtx{P}_{U}) \mtx{S} \vct{v}_j) \vct{v}_j^*.
\]
For each $1 \le j \le r$, let $\vct{\alpha}_j \equiv \mtx{S} \vct{v}_j$
and $\vct{\beta}_j^* \equiv \vct{u}_j^* \mtx{S}$. Then the decomposition
\[
\mtx{S} = \sum_{j = 1}^r \vct{u}_j \vct{\beta}_j^* + \sum_{j = 1}^r
(\mtx{P}_{U^\perp} \vct{\alpha}_j) \vct{v}_j^*,
\]
where $\mtx{P}_{U^\perp} = I - \mtx{P}_U$, provides a factorization of the form
\[
\mtx{S} = \mtx{X}\mtx{Y}^*, \quad \begin{cases}
  \mtx{X}  = [\vct{u}_1, \ldots, \vct{u}_r,  \mtx{P}_{U^\perp} \vct{\alpha}_1, \ldots, \mtx{P}_{U^\perp} \vct{\alpha}_r],\\
  \mtx{Y}  = [\vct{v}_1, \ldots, \vct{v}_r, \vct{\beta}_1, \ldots, \vct{\beta}_r].
\end{cases}
\]
It follows from our assumption that
\[
c^2([\vct{u}_1, \ldots, \vct{u}_r]) = \max_{1 \le i \le n} \sum_{1
\le
  j \le r} u_{ij}^2 = \max_{1 \le i \le n} \|\mtx{P}_U \vct{e}_i\|^2 \le \mu_0 r/n,
\]
and similarly for $[\vct{v}_1, \ldots, \vct{v}_r]$. Hence, to prove
Lemma \ref{teo:HE}, it suffices to prove that the maximum row norm
obeys $c([\vct{\beta}_1, \ldots, \vct{\beta}_r]) \le C \sqrt{\mu_0 r/n}$
for some constant $C > 0$, and similarly for the matrix
$[\mtx{P}_{U^\perp} \vct{\alpha}_1, \ldots, \mtx{P}_{U^\perp}
\vct{\alpha}_r]$.
\begin{lemma}
\label{teo:toshow}
There is a numerical constant $C$ such that for each $\beta > 2$,
\begin{equation}
  \label{eq:toshow}
  c([\vct{\alpha}_1, \ldots, \vct{\alpha}_r]) \le  C \sqrt{\mu_0 r/n}
\end{equation}
with probability at least $1-O(n^{-\beta})$ provided that $m$ obeys
the condition of Lemma \ref{teo:HE}.
\end{lemma}
A similar estimate for $[\vct{\beta}_1, \ldots, \vct{\beta}_r]$ is
obtained in the same way by exchanging the roles of $\vct{u}$ and
$\vct{v}$. Moreover, a minor modification of the argument gives
\begin{equation}
  \label{eq:toshow1}
  c([\mtx{P}_{U^\perp}\vct{\alpha}_1, \ldots, \mtx{P}_{U^\perp} \vct{\alpha}_r]) \le  C \sqrt{\mu_0 r/n}
\end{equation}
as well, and we will omit the details. In short, the estimate
\eqref{eq:toshow} implies Lemma \ref{teo:HE}.

\begin{proof}[of Lemma~\ref{teo:toshow}]
To prove \eqref{eq:toshow}, we use the notations of the previous
section and write
\begin{align*}
  \vct{\alpha}_j & = p^{-3} \sum_{a_1 b_1, a_2 b_2, a_3 b_3} \xi_{a_1
    b_1} \xi_{a_2 b_2} \xi _{a_3 b_3} E_{a_3 b_3} \<\PT \eabthree,
  \eabtwo\> \<\PT \eabtwo, \eabone\> \PT(\eabone) \vct{v}_j\\
  & = p^{-3} \sum_{\omega_1, \omega_2, \omega_3} \xi_{\omega_1}
  \xi_{\omega_2} \xi _{\omega_3} \, E_{\omega_3} P_{\omega_3    \omega_2} P_{\omega_2 \omega_1} \PT(\mtx{F_{\omega_1}}) \vct{v}_j\\
& =  p^{-3} \sum_{\omega_1, \omega_2, \omega_3} \xi_{\omega_1}
  \xi_{\omega_2} \xi _{\omega_3} \, E_{\omega_3} P_{\omega_3
    \omega_2} P_{\omega_2 \omega_1} (\mtx{F_{\omega_1}} \vct{v}_j)
\end{align*}
since for any matrix $\mtx{X}$, $\PT(\mtx{X}) \vct{v}_j = \mtx{X}
\vct{v}_j$ for each $1 \le j \le r$.  We then follow the same steps as
in Section \ref{sec:order2} and partition the sum depending on whether
some of the $\omega_i$'s are the same or not
\begin{equation}
  \label{eq:equiv1}
  \vct{\alpha}_{j} = p^{-3} \left[\sum_{\omega_1 = \omega_2 = \omega_3} ~ + ~
    \sum_{\omega_1 \neq \omega_2 = \omega_3} ~ + ~ \sum_{\omega_1 = \omega_3
      \neq \omega_2} ~ + ~ \sum_{\omega_1 = \omega_2 \neq \omega_3}
    ~ + ~ \sum_{\omega_1 \neq \omega_2 \neq \omega_3}\right].
\end{equation}
The idea is this: to establish \eqref{eq:toshow}, it is sufficient to
show that if $\vct{\gamma}_{j}$ is any of the five terms above, it
obeys
\begin{equation}
\label{eq:fifth}
\sqrt{\sum_{1 \le j \le r} |\gamma_{ij}|^2} \le C \sqrt{\mu_0 r/n}
\end{equation}
($\gamma_{ij}$ is the $i$th component of $\vct{\gamma}_{j}$ as usual)
with large probability.
The strategy for getting such estimates is to use decoupling whenever
applicable.

Just as Theorem \ref{teo:PO} proved useful to bound the norm of
$p^{-1} (\PO-p\OpId) \mathcal{H}^2(\mtx{E})$ in Section
\ref{sec:order2}, the lemma below will help bounding the magnitudes
of the components of $\vct{\alpha}_{j}$.

\begin{lemma}
  \label{teo:sij}
  Define $\mtx{S} \equiv p^{-1} \sum_{ij} \sum_\omega \xi_\omega
  H_\omega \<\vct{e}_i, \mtx{F}_\omega \vct{v}_j\> \vct{e}_i \vct{e}_j^*$. Then for each $\lambda > 0$
\begin{equation}
\label{eq:sij}
\P(\|\mtx{S}\|_\infty \ge \sqrt{\mu_0/n}) \le 2n^2
\exp\left(-\frac{1}{\frac{2n}{\mu_0 p} \|H\|_\infty^2 + \frac{2}{3p}
    \sqrt{r}  \|H\|_\infty}\right).
\end{equation}
\end{lemma}
\begin{proof} The proof is an application of Bernstein's inequality
\eqref{eq:bern2}.  Note that $\<\vct{e}_i, \mtx{F}_\omega
\vct{v}_j\> = 1_{\{a = i\}} v_{bj}$ and hence
\[
\textrm{Var}(S_{ij}) \le p^{-1} \|\mtx{H}\|_\infty^2 \sum_{\omega}
|\<\vct{e}_i, \mtx{F}_\omega \vct{v}_j\>|^2 =  p^{-1}
\|\mtx{H}\|_\infty^2
\]
since $\sum_\omega |\<\vct{e}_i, \mtx{F}_\omega \vct{v}_j\>|^2 = 1$,
and $|p^{-1} H_\omega \<\vct{e}_i, \mtx{F}_\omega \vct{v}_j\>| \le
p^{-1} \, \|\mtx{H}\|_\infty \, \sqrt{\mu_0 r/n}$ since
$|\<\vct{e}_i, \mtx{F}_\omega \vct{v}_j\>| \le |v_{bj}|$ and
\[
|v_{bj}| \le \|\mtx{P}_V \vct{e}_b\| \le \sqrt{\mu_0 r/n}.
\]
\end{proof}

Each term in \eqref{eq:equiv1} is given by the corresponding term in
\eqref{eq:equiv} after formally substituting $\mtx{F}_\omega$ with
$\mtx{F}_\omega \vct{v}_j$. We begin with the first term whose $i$th
component is equal to
\begin{equation}
\label{eq:term1}
\gamma_{ij} \equiv p^{-3}(1-3p+3p^2) \sum_\omega \xi_{\omega} \, E_{\omega}
P_{\omega \omega}^2 \aij + p^{-2}(1-3p+2p^2) \sum_\omega E_{\omega}
P_{\omega \omega}^2 \aij.
\end{equation}
Ignoring the constant factor $(1-3p+3p^2)$ which is bounded by 1, we
write the first of these two terms as
\[
(\mtx{S}_0)_{ij} \equiv p^{-1} \sum_\omega \xi_\omega \, H_\omega \aij, \qquad
H_\omega = E_\omega \, (p^{-1} P_{\omega \omega})^2.
\]
Since $\|\mtx{H}\|_\infty \le (\mu_0 nr/m)^2 \, \mu_1 \sqrt{r}/n$, it
follows from Lemma \eqref{teo:sij} that
\begin{equation*}
  \P\left(\|\mtx{S}_0\|_\infty \ge
    \sqrt{\mu_0/n}\right) \le 2n^2 \,
  \e^{-1/D}, \quad D \le C\left(\mu_0^3 \mu_1^2 \left(\frac{nr}{m}\right)^5 + \mu_0^{2} \mu_1 \left(\frac{nr}{m}\right)^3\right)
\end{equation*}
for some numerical $C > 0$.  Since $\mu_1 \le \mu_0 \sqrt{r}$, we have
that when $m \ge \lambda \mu_0 \, nr^{6/5} (\beta \log n)$ for some
numerical constant $\lambda > 0$, $\|\mtx{S}_0\|_\infty \ge
\sqrt{\mu_0/n}$ with probability at most $2n^2 e^{-(\beta \log n)^3}$;
this probability is inversely proportional to a superpolynomial in
$n$.  For the second term, the matrix with entries $E_{\omega}
P_{\omega \omega}^2$ is given by
\[
\mtx{\Lambda}^2_U \mtx{E} + \mtx{E} \mtx{\Lambda}^2_V +
2\mtx{\Lambda}_U \mtx{E} \mtx{\Lambda}_V  + \mtx{\Lambda}^2_U \mtx{E} \mtx{\Lambda}^2_V
- 2\mtx{\Lambda}^2_U \mtx{E} \mtx{\Lambda}_V - 2\mtx{\Lambda}_U
\mtx{E} \mtx{\Lambda}^2_V
\]
and thus
\[
\sum_\omega E_{\omega} P_{\omega \omega}^2 \aij = \<\vct{e}_i,
(\mtx{\Lambda}^2_U \mtx{E} + \mtx{E} \mtx{\Lambda}^2_V +
2\mtx{\Lambda}_U \mtx{E} \mtx{\Lambda}_V  + \mtx{\Lambda}^2_U \mtx{E} \mtx{\Lambda}^2_V
- 2\mtx{\Lambda}^2_U \mtx{E} \mtx{\Lambda}_V - 2\mtx{\Lambda}_U
\mtx{E} \mtx{\Lambda}^2_V) \vct{v}_j\>.
\]
This is a sum of six terms and we will show how to bound the first
three; the last three are dealt in exactly the same way and obey
better estimates. For the first, we have
\[
\<\vct{e}_i, \mtx{\Lambda}^2_U \mtx{E} \vct{v}_j \> = \<
\mtx{\Lambda}^2_U \vct{e}_i, \mtx{E} \vct{v}_j \> = \|\mtx{P}_U
\vct{e}_i\|^4 \<\vct{e}_i, \vct{u}_j\>
\]
Hence
\[
p^{-2} \sqrt{\sum_{1 \le j \le r} |\<\vct{e}_i, \mtx{\Lambda}^2_U
  \mtx{E} \vct{v}_j \>|^2} = p^{-2} \|\mtx{P}_U \vct{e}_i\|^4
\sqrt{\sum_{1 \le j \le r} | \<\vct{e}_i, \vct{u}_j\>|^2} = p^{-2}
\|\mtx{P}_U \vct{e}_i\|^5 \le \left(\frac{\mu_0 r}{np}\right)^2
\sqrt{\frac{\mu_0 r}{n}}.
\]
In other words, when $m \ge \mu_0 nr$, the right hand-side is bounded
by $\sqrt{{\mu_0 r}/{n}}$ as desired. For the second term, we have
\[
\<\vct{e}_i, \mtx{E} \mtx{\Lambda}^2_V \vct{v}_j \> = \sum_{b}
\|\mtx{P}_V \vct{e}_b\|^4 v_{bj} \<\vct{e}_i, \mtx{E} \vct{e}_b\> =
\sum_{b} \|\mtx{P}_V \vct{e}_b\|^4 v_{bj} E_{ib}.
\]
Hence it follows from the Cauchy-Schwarz inequality and
\eqref{eq:Cauchy} that
\[
p^{-2} |\<\vct{e}_i, \mtx{E} \mtx{\Lambda}^2_V \vct{v}_j \>| \le
\left(\frac{\mu_0 r}{np}\right)^2 \sqrt{\frac{\mu_0 r}{n}}.
\]
In other words, when $m \ge \mu_0 nr^{5/4}$,
\begin{equation}
 \label{eq:lambda6}
 p^{-2} \sqrt{\sum_{1 \le j \le r} |\<\vct{e}_i, \mtx{E} \mtx{\Lambda}^2_V \vct{v}_j \>|^2} \le \sqrt{\frac{\mu_0 r}{n}}
\end{equation}
as desired. For the third term, we have
\[
\<\vct{e}_i, \mtx{\Lambda}_U \mtx{E} \mtx{\Lambda}_V \vct{v}_j \> =
\|\mtx{P}_U \vct{e}_i\|^2 \sum_{b} \|\mtx{P}_V \vct{e}_b\|^2 v_{bj}
E_{ib}.
\]
The Cauchy-Schwarz inequality gives
\[
2p^{-2}|\<\vct{e}_i, \mtx{\Lambda}_U \mtx{E} \mtx{\Lambda}_V \vct{v}_j
\>| \le 2\left(\frac{\mu_0 r}{np}\right)^2 \sqrt{\frac{\mu_0 r}{n}}
\]
just as before. In other words, when $m \ge \mu_0 nr^{5/4}$,
$2p^{-2}\sqrt{\sum_{ 1\le j \le r} |\<\vct{e}_i, \mtx{\Lambda}_U
  \mtx{E} \mtx{\Lambda}_V \vct{v}_j \>|^2}$ is bounded by
$2\sqrt{\mu_0 r/n}$. The other terms obey \eqref{eq:lambda6} as well
when $m \ge \mu_0 nr^{5/4}$.  In conclusion, the first term
\eqref{eq:term1} in \eqref{eq:equiv1} obeys \eqref{eq:fifth} with
probability at least $1 - O(n^{-\beta})$ provided that $m \ge \mu_0
n r^{5/4} (\beta \log n)$.

We now turn our attention to the second term which can be written as
\begin{multline*}
\gamma_{ij} \equiv p^{-3} ({1-2p}) \sum_{\omega_1 \neq \omega_2}
\xi_{\omega_1} \xi_{\omega_2} \, E_{\omega_2} P_{\omega_2 \omega_2}
P_{\omega_2
  \omega_1} \aijone \\
+ p^{-2}({1-p}) \sum_{\omega_1 \neq \omega_2} \xi_{\omega_1} \,
E_{\omega_2} P_{\omega_2 \omega_2} P_{\omega_2 \omega_1} \aijone.
\end{multline*}
We decouple the first term so that it suffices to bound
\[
(\mtx{S}_0)_{ij} \equiv p^{-1} \sum_{\omega_1} \xi^{(1)}_{\omega_1} H_{\omega_1}
\aijone, \qquad H_{\omega_1} \equiv p^{-2} \sum_{\omega_2 : \omega_2
  \neq \omega_1} \xi^{(2)}_{\omega_2} \, E_{\omega_2} P_{\omega_2
  \omega_2} P_{\omega_2 \omega_1},
\]
where the sequences $\{\xi_\omega^{(1)}\}$ and $\{\xi_\omega^{(2)}\}$
are independent. The method from Section \ref{sec:order1} shows that
\[
\|\mtx{H}\|_\infty \le C \, \sqrt\frac{\mu_0 nr \beta \log n}{m} \,
\sup_\omega |E_\omega (p^{-1} P_{\omega \omega})| \le C\, \sqrt{\beta
  \log n} \, \left(\frac{\mu_0 nr}{m}\right)^{3/2} \,
\|\mtx{E}\|_\infty
\]
with probability at least $1-2n^{-\beta}$ for each $\beta >
2$. Therefore, Lemma \ref{teo:sij} gives
\begin{equation}
\label{eq:threetimesa}
\P\left(\|\mtx{S}_0\|_\infty \ge \sqrt{\mu_0/n}\right) \le 2n^2
e^{-1/D},
\end{equation}
where $D$ obeys
\begin{equation}
\label{eq:threetimesb}
  D \le C \left(\mu_0^2 \mu_1^2 (\beta \log n)
    \left(\frac{nr}{m}\right)^4 + \mu_0^{3/2} \mu_1 \sqrt{\beta \log
      n}\left(\frac{nr}{m}\right)^{5/2}\right).
\end{equation}
for some positive constant $C$.  Hence, when $m \ge \lambda \mu_0 \,
nr^{5/4} (\beta \log n)$ for some sufficiently large numerical
constant $\lambda > 0$, we have that $\|\mtx{S}_0\|_\infty \ge
\sqrt{\mu_0/n}$ with probability at most $2n^2 e^{-(\beta
  \log n)^2}$. This is inversely proportional to a superpolynomial in
$n$.  We write the second term as
\[
(\mtx{S}_1)_{ij} \equiv p^{-1} \sum_{\omega_1 \neq \omega_2} \xi_{\omega_1}
H_{\omega_1} \aijone, \qquad H_{\omega_1} = p^{-1} \sum_{\omega_2 :
  \omega_2 \neq \omega_1} E_{\omega_2} P_{\omega_2 \omega_2}
P_{\omega_2 \omega_1}.
\]
We know from Section \ref{sec:order2} that $\mtx{H}$ obeys
$\|\mtx{H}\|_\infty \le C \, \mu_0^2\, r^2/m$ since $\mu_1 \le \mu_0
\sqrt{r}$ so that Lemma \ref{teo:sij} gives
\[
\P\left(\|\mtx{S}_1\|_\infty \ge \sqrt{\mu_0/n}\right) \le
2n^2 e^{-1/D}, \quad D \le C \left(\mu_0^3 \frac{n^3 r^4}{m^3} +
  \mu_0^2 \frac{n^2 r^{5/2}}{m^2}\right)
\]
for some $C > 0$.  Hence, when $m \ge \lambda \mu_0 \, nr^{4/3} (\beta
\log n)$ for some numerical constant $\lambda > 0$, we have that
$\|\mtx{S}_1\|_\infty \ge  \sqrt{\mu_0/n}$ with
probability at most $2n^2 e^{-(\beta \log n)^2}$. This is inversely
proportional to a superpolynomial in $n$.  In conclusion and taking
into account the decoupling constants in \eqref{eq:decoupling}, the
second term in \eqref{eq:equiv1} obeys \eqref{eq:fifth} with
probability at least $1-O(n^{-\beta})$ provided that $m$ is
sufficiently large as above.

We now examine the third term which can be written as
\[
p^{-3}({1-2p}) \sum_{\omega_1 \neq \omega_2} \xi_{\omega_1}
\xi_{\omega_2} \, E_{\omega_1} P^2_{\omega_2 \omega_1} \aijone +
p^{-2}({1-p}) \sum_{\omega_1 \neq \omega_2} \xi_{\omega_2} \,
E_{\omega_1} P^2_{\omega_2 \omega_1} \aijone.
\]
For the first term of the right-hand side, it suffices to estimate the
tail of
\[
(\mtx{S}_0)_{ij} \equiv p^{-1} \sum_{\omega_1} \xi^{(1)}_{\omega_1} E_{\omega_1}
H_{\omega_1} \aijone, \qquad H_{\omega_1} \equiv p^{-2} \sum_{\omega_2
  : \omega_2 \neq \omega_1} \xi^{(2)}_{\omega_2} \, P^2_{\omega_2
  \omega_1},
\]
where $\{\xi^{(1)}_\omega\}$ and $\{\xi^{(2)}_\omega\}$ are
independent.  We know from Section \ref{sec:order2} that
$\|\mtx{H}\|_\infty$ obeys $\|\mtx{H}\|_\infty \le C \, \sqrt{\beta
  \log n} \, (\mu_0 nr/m)^{3/2}$ with probability at least
$1-2n^{-\beta}$ for each $\beta > 2$. Thus, Lemma \eqref{teo:sij}
shows that $\mtx{S}_0$ obeys
\eqref{eq:threetimesa}--\eqref{eq:threetimesb} just as before.  The
other term is equal to $(1-p)$ times $\sum_{\omega_1} E_{\omega_1}
H_{\omega_1} \aijone$, and by the Cauchy-Schwarz inequality and
\eqref{eq:Cauchy}
\[
\left|\sum_{\omega_1} E_{\omega_1} H_{\omega_1} \aijone\right| \le
\|\mtx{H}\|_\infty \, \|\vct{e}_i^* \mtx{E}\| \, \left(\sum_{b}
v_{bj}^2\right)^{1/2} \le C \sqrt{\frac{\mu_0}{n}} \, \sqrt{\beta
\log n} \, \left(\frac{\mu_0
    nr^{4/3}}{m}\right)^{3/2}
\]
on the event where $\|\mtx{H}\|_\infty \le C \, \sqrt{\beta \log n} \,
(\mu_0 nr/m)^{3/2}$.  Hence, when $m \ge \lambda \mu_0 \, nr^{4/3} \,
(\beta \log n)$ for some numerical constant $\lambda > 0$, we have
that $|\sum_{\omega_1} E_{\omega_1} H_{\omega_1} \aijone| \le
 \sqrt{\mu_0/n}$ on this event. In conclusion, the third
term in \eqref{eq:equiv1} obeys \eqref{eq:fifth} with probability at
least $1-O(n^{-\beta})$ provided that $m$ is sufficiently large as
above.

We proceed to the fourth term which can be written as
\[
p^{-3}({1-2p}) \sum_{\omega_1 \neq \omega_3} \xi_{\omega_1}
\xi_{\omega_3} \, E_{\omega_3} P_{\omega_3 \omega_1} P_{\omega_1
  \omega_1} \aijone + p^{-2}({1-p}) \sum_{\omega_1 \neq \omega_3}
\xi_{\omega_3} \, E_{\omega_3} P_{\omega_3 \omega_1} P_{\omega_1
  \omega_1} \aijone.
\]
We use the decoupling trick for the first term and bound the tail of
\[
(\mtx{S}_0)_{ij} \equiv p^{-1} \sum_{\omega_1} \xi^{(1)}_{\omega_1} H_{\omega_1}
(p^{-1} P_{\omega_1 \omega_1}) \, \aijone, \qquad H_{\omega_1} \equiv
p^{-1} \sum_{\omega_3 : \omega_3 \neq \omega_1} \xi^{(3)}_{\omega_3}
\, E_{\omega_3} P_{\omega_3 \omega_1},
\]
where $\{\xi^{(1)}_\omega\}$ and $\{\xi^{(3)}_\omega\}$ are
independent.  We know from Section \ref{sec:order1} that
\[
\|\mtx{H}\|_\infty \le C \, \sqrt{\frac{\mu_0 nr\beta \log n}{m}} \,
\|\mtx{E}\|_\infty
\]
with probability at least $1-2n^{-\beta}$ for each $\beta >
2$. Therefore, Lemma \ref{teo:sij} shows that $\mtx{S}_0$ obeys
\eqref{eq:threetimesa}--\eqref{eq:threetimesb} just as before.  The
other term is equal to $(1-p)$ times $\sum_{\omega_1} H_{\omega_1}
(p^{-1} P_{\omega_1 \omega_1}) \, \aijone$, and the Cauchy-Schwarz
inequality gives
\begin{equation*}
  \left|\sum_{\omega_1} H_{\omega_1} (p^{-1} P_{\omega_1 \omega_1}) \,
    \aijone\right| \le \sqrt{n} \, \|\mtx{H}\|_\infty \, \frac{\mu_0 nr}{m}
  \le
  C\, \frac{\mu_1 \sqrt{r \beta \log n}}{\sqrt{n}} \,
  \left(\frac{\mu_0 nr}{m}\right)^{3/2}
\end{equation*}
on the event $\|\mtx{H}\|_\infty \le C \, \sqrt{\mu_0 nr (\beta \log
  n)/m} \, \|\mtx{E}\|_\infty$. Because $\mu_1 \le \mu_0 \sqrt{r}$, we
have that whenever $m \ge \lambda \, \mu_0^{4/3} nr^{5/3} \, (\beta
\log n)$ for some numerical constant $\lambda > 0$, $p^{-1}
|\sum_{\omega_1} H_{\omega_1} P_{\omega_1 \omega_1} \aijone| \le
\sqrt{\mu_0/n}$ just as before. In conclusion, the fourth term in
\eqref{eq:equiv1} obeys \eqref{eq:fifth} with probability at least
$1-O(n^{-\beta})$ provided that $m$ is sufficiently large as above.

We finally examine the last term
\[
p^{-3} \sum_{\omega_1 \neq \omega_2 \neq \omega_3} \xi_{\omega_1}
\xi_{\omega_2} \xi _{\omega_3} \, E_{\omega_3} P_{\omega_3
\omega_2} P_{\omega_2 \omega_1} \aijone.
\]
Just as before, we need to bound the tail of
\[
(\mtx{S}_0)_{ij} \equiv p^{-1} \sum_{\omega_1, \omega_2, \omega_3}
\xi^{(1)}_{\omega_1} H_{\omega_1} \aijone,
\]
where $\mtx{H}$ is given by \eqref{eq:forlater}.  We know from Section
\ref{sec:order2} that $\mtx{H}$ obeys $$\|\mtx{H}\|_\infty \le C \,
(\beta \log n)\, \frac{\mu_0 nr}{m} \, \mu_1 \frac{\sqrt{r}}{n}$$ with
probability at least $1-4n^{-\beta}$ for each $\beta > 2$. Therefore,
Lemma \ref{teo:sij} gives
\[
\P\left(\|\mtx{S}_0\|_\infty \ge \frac{1}{5} \sqrt{\mu_0/n}\right) \le 2n^2
e^{-1/D}, \quad D \le C \left(\mu_0 \mu_1^2 (\beta \log n)^2
  \left(\frac{n r}{m}\right)^3 + \mu_0 \mu_1 (\beta \log n)
  \left(\frac{n r}{m}\right)^2\right)
\]
for some $C > 0$.  Hence, when $m \ge \lambda \mu_0 \, nr^{4/3} (\beta
\log n)$ for some numerical constant $\lambda > 0$, we have that
$\|\mtx{S}_0\|_\infty \ge \frac{1}{5} \sqrt{\mu_0/n}$ with probability
at most $2n^2 e^{-(\beta \log n)}$.  In conclusion, the fifth term in
\eqref{eq:equiv1} obeys \eqref{eq:fifth} with probability at least
$1-O(n^{-\beta})$ provided that $m$ is sufficiently large as above.

To summarize the calculations of this section, if $m = \lambda \,
\mu_0^{4/3} nr^{5/3} \, (\beta \log n)$ where $\beta \ge 2$ is fixed
and $\lambda$ is some sufficiently large numerical constant, then
\[
\sum_{1\le j \le r} |\alpha_{ij}|^2 \le \mu_0 r/n
\]
with probability at least $1 - O(n^{-\beta})$. This concludes the proof.
\end{proof}

\subsection{Proof of Lemma \ref{teo:highorder}}

It remains to study the spectral norm of $p^{-1} (\PTc \PO \PT)\sum_{k
  \ge k_0} \mathcal{H}^k(\mtx{E})$ for some positive integer $k_0$, which we
bound by the Frobenius norm
\begin{align*}
p^{-1}  \|(\PTc \PO \PT)\sum_{k \ge k_0} \mathcal{H}^k(\mtx{E})\| &
\le p^{-1} \|(\PO
  \PT) \sum_{k \ge
    k_0} \mathcal{H}^k(\mtx{E})\|_F \\
 & \le \sqrt{3/2p} \, \|\sum_{k \ge k_0} \mathcal{H}^k(\mtx{E})\|_F,
\end{align*}
where the inequality follows from Corollary \ref{teo:POPT}. To bound
the Frobenius of the series, write
\begin{align*}
  \|\sum_{k \ge k_0} \mathcal{H}^k(\mtx{E})\|_F & \le \|\mathcal{H}\|^{k_0} \|\mtx{E}\|_F +
  \|\mathcal{H}\|^{k_0+1} \|\mtx{E}\|_F + \ldots  \\
  & \le \frac{ \|\mathcal{H}\|^{k_0}}{1- \|\mathcal{H}\|} \, \|\mtx{E}\|_F.
\end{align*}
Theorem \ref{teo:rudelson} gives an upper bound on $\|\mathcal{H}\|$
since $\|\mathcal{H}\| \le C_R \, \sqrt{\mu_0 nr \beta \log n/m}<1/2$
on an event with probability at least $1 - 3n^{-\beta}$. Since
$\|\mtx{E}\|_F = \sqrt{r}$, we conclude that
\[
p^{-1} \|(\PO \PT) \sum_{k \ge k_0} \mathcal{H}^k(\mtx{E})\|_F \le C
\, \frac{1}{\sqrt{p}} \, \left(\frac{\mu_0 nr \beta \log
    n}{m}\right)^{k_0/2} \, \sqrt{r} = C \, \left(\frac{n^2
    r}{m}\right)^{1/2} \, \left(\frac{\mu_0 nr \beta \log
    n}{m}\right)^{k_0/2}
\]
with large probability. This is the content of Lemma
\ref{teo:highorder}.

\section{Numerical Experiments}
\label{sec:numerical}

To demonstrate the practical applicability of the nuclear norm
heuristic for recovering low-rank matrices from their entries, we
conducted a series of numerical experiments for a variety of the
matrix sizes $n$, ranks $r$, and numbers of entries $m$.  For each
$(n,m,r)$ triple, we repeated the following procedure $50$ times. We
generated $\mtx{M}$, an $n\times n$ matrix of rank $r$, by sampling
two $n\times r$ factors $\mtx{M}_L$ and $\mtx{M}_R$ with i.i.d.
Gaussian entries and setting $\mtx{M}=\mtx{M}_L\mtx{M}_R^*$.  We
sampled a subset $\Omega$ of $m$ entries uniformly at random. Then
the nuclear norm minimization
\begin{equation*}
  \begin{array}{ll}
\textrm{minimize}   & \quad \|\mtx{X}\|_*\\
\textrm{subject to} & \quad   X_{ij} = M_{ij}, \quad (i,j) \in
\Omega
 \end{array}
\end{equation*}
was solved using the SDP solver SDPT3~\cite{SDPT3}. We declared
$\mtx{M}$ to be recovered if the solution returned by the SDP,
$\mtx{X_{\mathrm{opt}}}$, satisfied
$\|\mtx{X_{\mathrm{opt}}}-\mtx{M}\|_F/\|\mtx{M}\|_F<10^{-3}$.
Figure~\ref{fig:phase-trans-full} shows the results of these
experiments for $n=40$ and $50$.  The $x$-axis corresponds to the
fraction of the entries of the matrix that are revealed to the SDP
solver.  The $y$-axis corresponds to the ratio between the dimension
of the set of rank $r$ matrices, $d_r=r(2n-r)$, and the number of
measurements $m$. Note that both of these axes range from zero to
one as a value greater than one on the $x$-axis corresponds to an
overdetermined linear system where the semidefinite program always
succeeds,  and a value of greater than one on the $y$-axis
corresponds to a situation where there is always an infinite number
of matrices with rank $r$ with the given entries.  The color of each
cell in the figures reflects the empirical recovery rate of the $50$
runs (scaled between $0$ and $1$). White denotes perfect recovery in
all experiments, and black denotes failure for all experiments.
Interestingly, the experiments reveal very similar plots for
different $n$, suggesting that our asymptotic conditions for
recovery may be rather conservative.

\begin{figure}
 \centering
 \begin{tabular}{cc}
   \includegraphics[width=8cm]{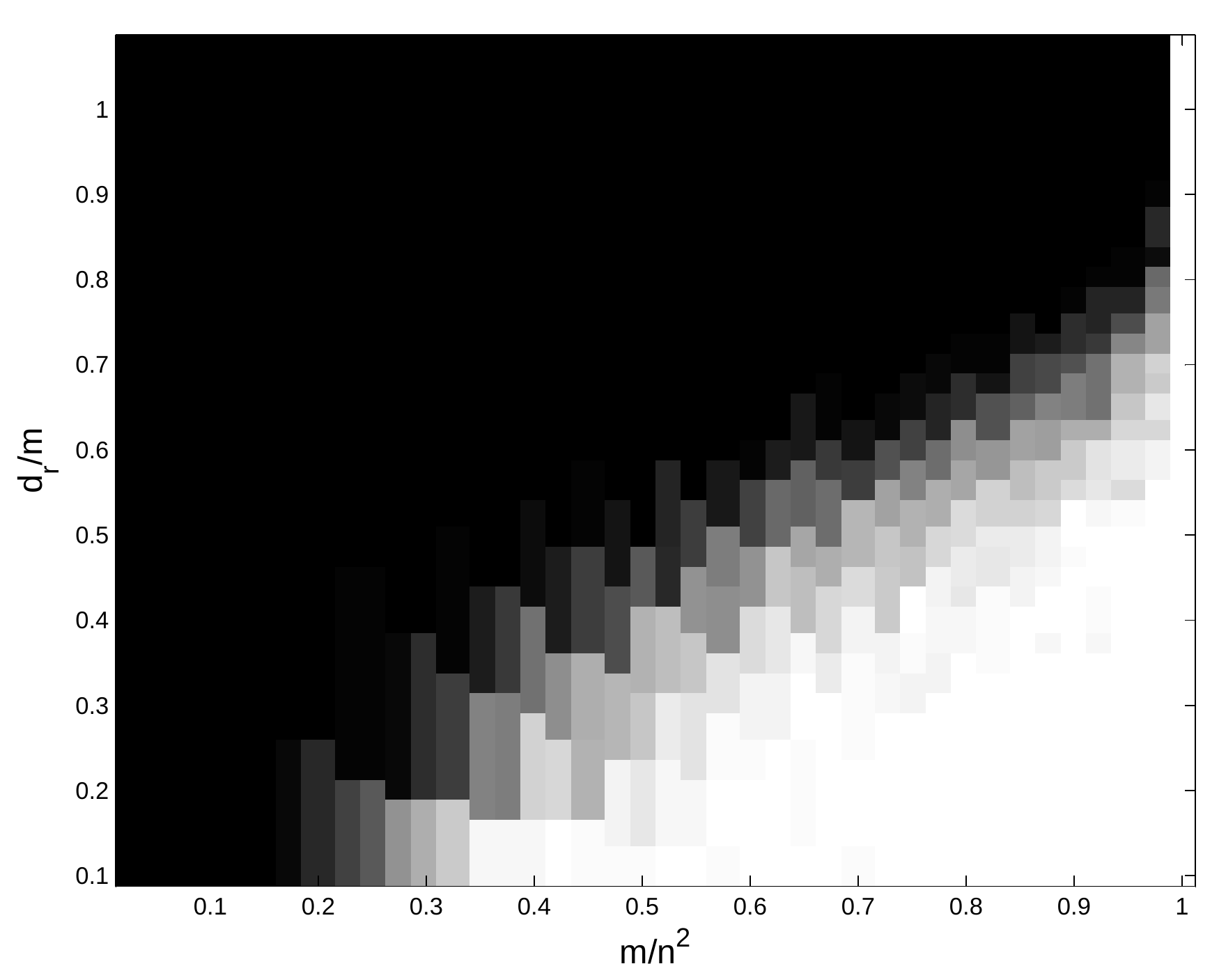}  &
   \includegraphics[width=8cm]{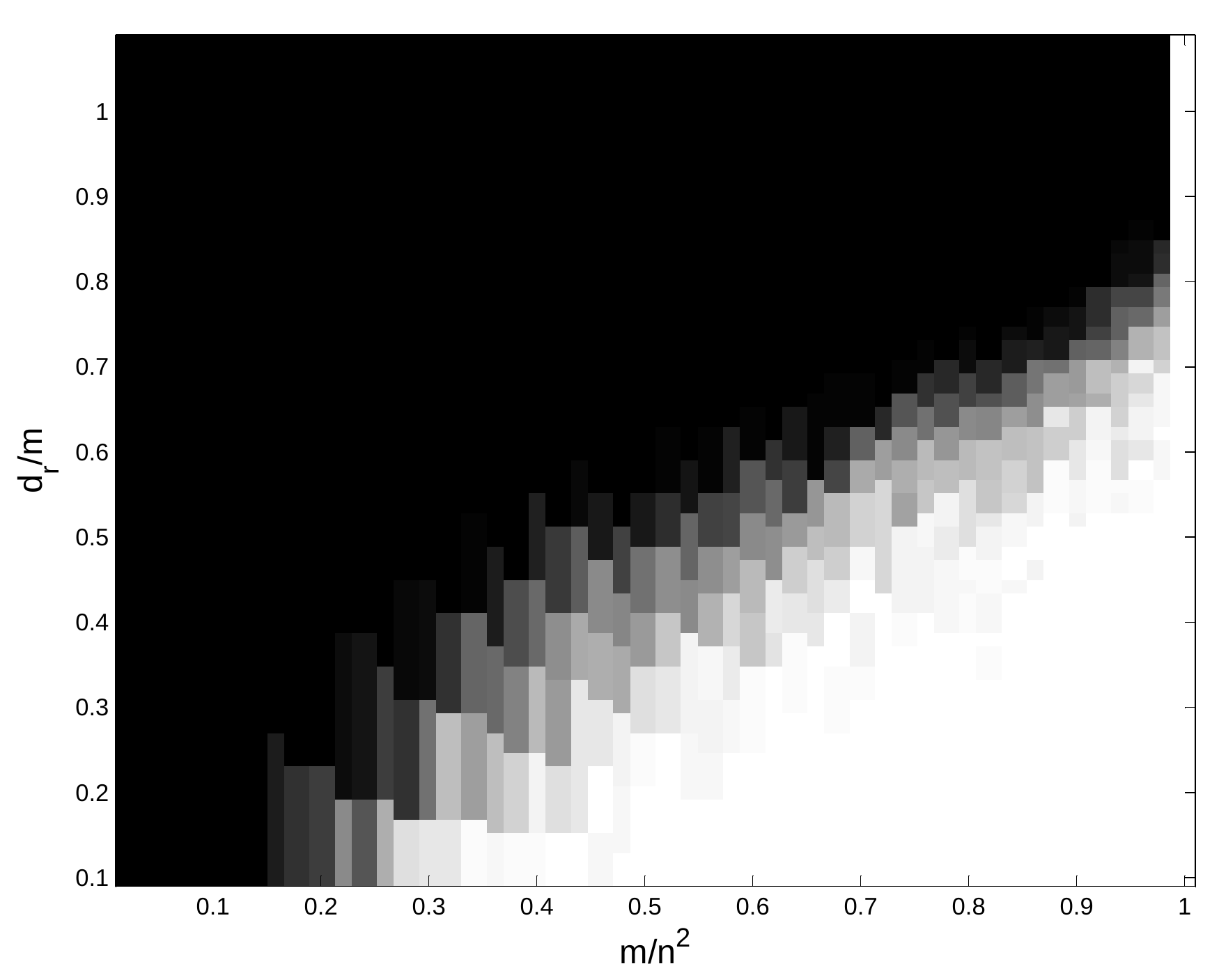}\\ (a) & (b)
     \end{tabular}
 \caption{\small
{\bf Recovery of full matrices from their entries.} For each
$(n,m,r)$ triple, we repeated the following procedure $50$ times. A
matrix $\mtx{M}$ of rank $r$ and a subset of $m$ entries were
selected at random. Then we solved the nuclear norm minimization for
$\mtx{X}$ subject to $X_{ij} = M_{ij}$ on the selected entries. We
declared $\mtx{M}$ to be recovered if
$\|\mtx{X_{\mathrm{opt}}}-\mtx{M}\|_F/\|\mtx{M}\|_F<10^{-3}$. The
results are shown for (a) $n=40$ and (b) $n=50$. The color of each
cell reflects the empirical recovery rate (scaled between $0$ and
$1$).  White denotes perfect recovery in all experiments, and black
denotes failure for all experiments.} \label{fig:phase-trans-full}
\end{figure}

For a second experiment, we generated random \emph{positive
semidefinite} matrices and tried to recover them from their entries
using the nuclear norm heuristic. As above, we repeated the same
procedure $50$ times for each $(n,m,r)$ triple. We generated
$\mtx{M}$, an $n\times n$ positive semidefinite matrix of rank $r$,
by sampling an $n\times r$ factor $\mtx{M}_F$ with i.i.d. Gaussian
entries and setting $\mtx{M}=\mtx{M}_F\mtx{M}_F^*$.  We sampled a
subset $\Omega$ of $m$ entries uniformly at random. Then we solved
the nuclear norm minimization problem
\begin{equation*}
  \begin{array}{ll}
\textrm{minimize}   & \quad \trace(\mtx{X})\\
\textrm{subject to} & \quad   X_{ij} = M_{ij}, \quad (i,j) \in \Omega\\
& \quad \mtx{X} \succeq 0
 \end{array}.
\end{equation*}
As above, we declared $\mtx{M}$ to be recovered if
$\|\mtx{X_{\mathrm{opt}}}-\mtx{M}\|_F/\|\mtx{M}\|_F<10^{-3}$.
Figure~\ref{fig:phase-trans-psd} shows the results of these
experiments for $n=40$ and $50$.  The $x$-axis again corresponds to
the fraction of the entries of the matrix that are revealed to the SDP
solver, but, in this case, the number of measurements is divided by
$D_n = n(n+1)/2$, the number of unique entries in a
positive-semidefinite matrix and the dimension of the rank $r$
matrices is $d_r = nr - r(r-1)/2$.  The color of each cell is chosen
in the same fashion as in the experiment with full matrices.
Interestingly, the recovery region is much larger for positive
semidefinite matrices, and future work is needed to investigate if the
theoretical scaling is also more favorable in this scenario of
low-rank matrix completion.

\begin{figure}
 \centering
 \begin{tabular}{cc}
   \includegraphics[width=8cm]{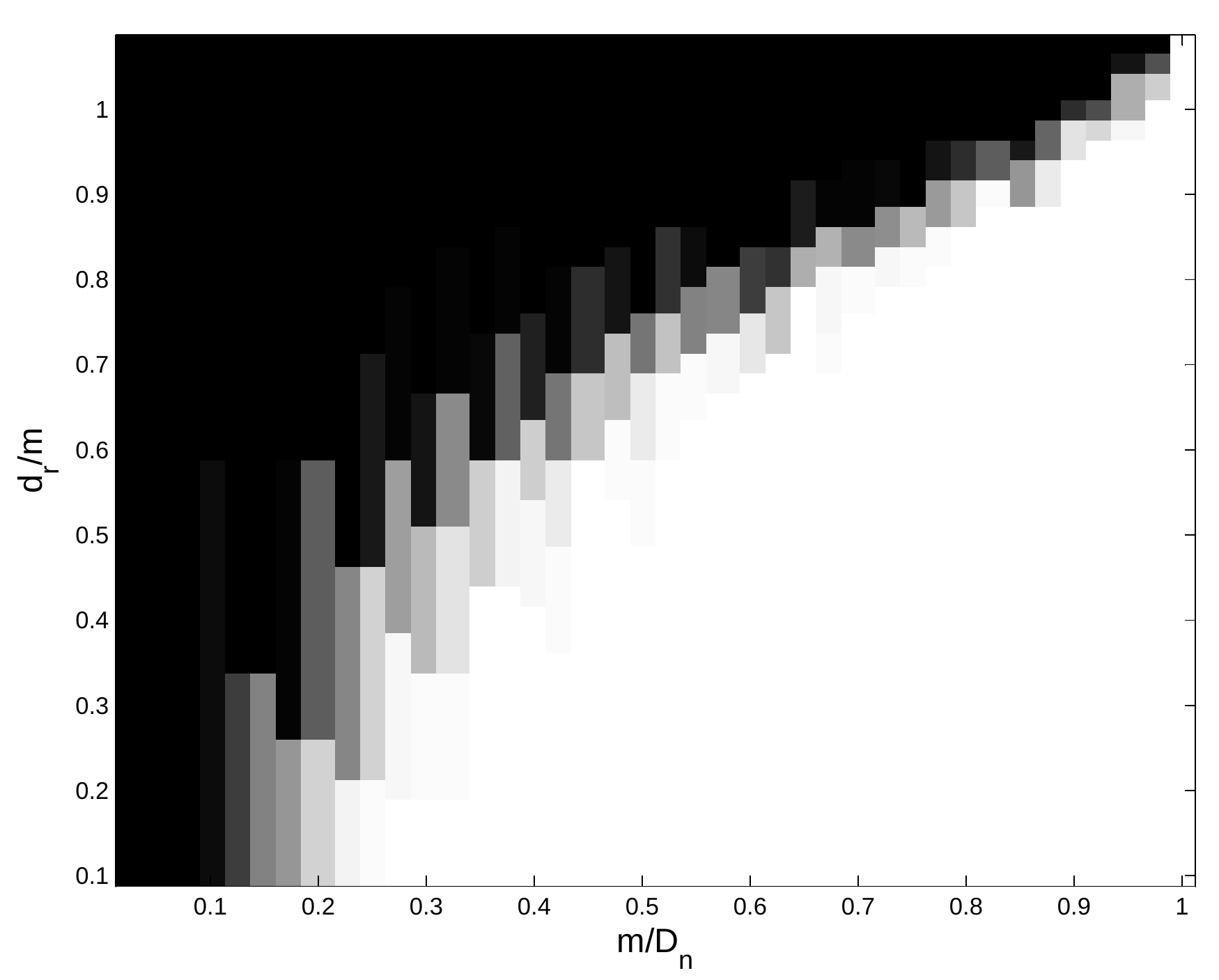}  &
   \includegraphics[width=8cm]{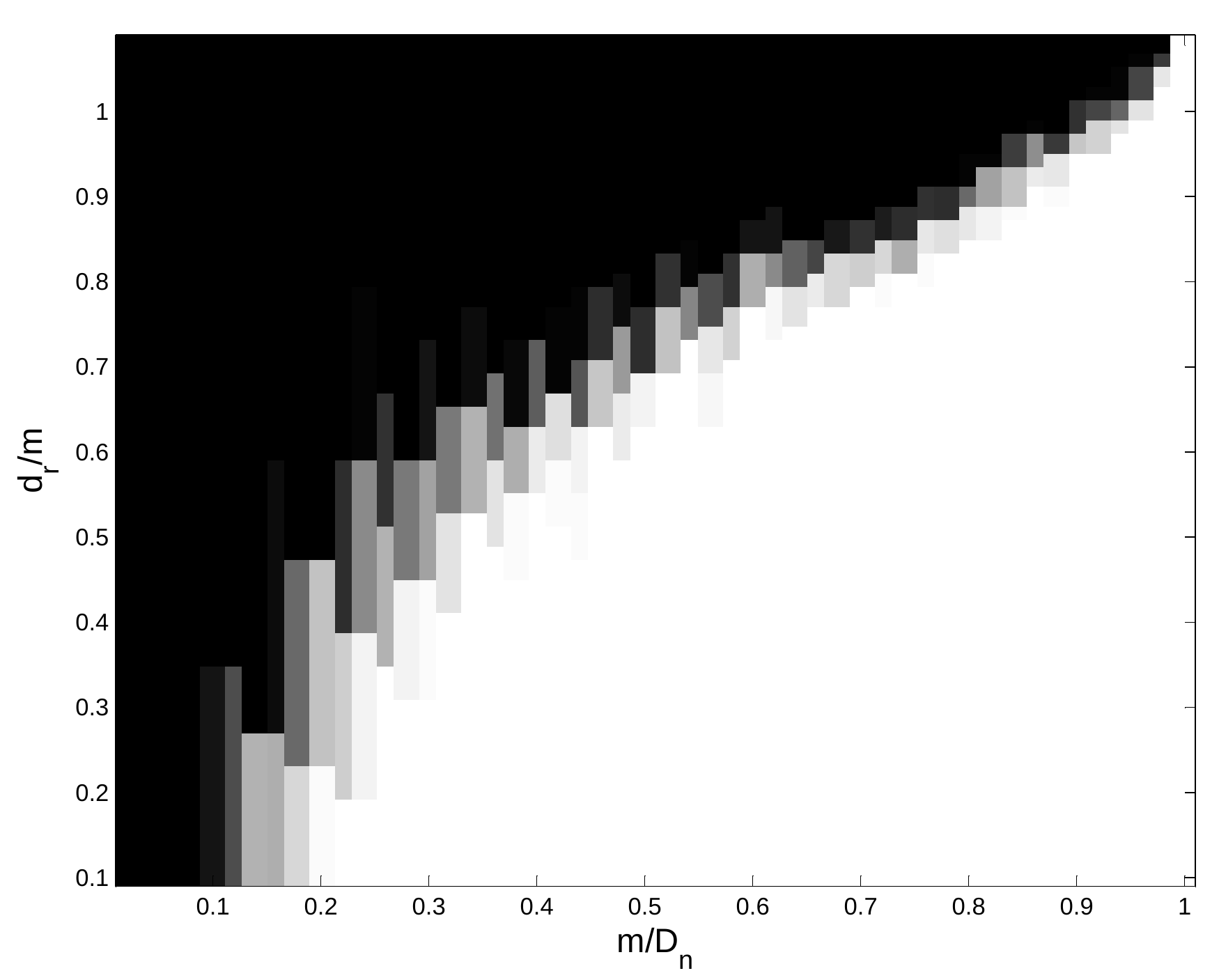}\\ (a) & (b)
     \end{tabular}
\caption{\small {\bf Recovery of positive semidefinite matrices from
their entries.} For each $(n,m,r)$ triple, we repeated the following
procedure $50$ times. A positive semidefinite matrix $\mtx{M}$ of
rank $r$ and a set of $m$ entries were selected at random. Then we
solved the nuclear norm minimization subject to $X_{ij} = M_{ij}$ on
the selected entries with the constraint that $\mtx{X}\succeq 0$.
The color scheme for each cell denotes empirical recovery
probability and is the same as in Figure~\ref{fig:phase-trans-full}.
The results are shown for (a) $n=40$ and (b) $n=50$.}
\label{fig:phase-trans-psd}
\end{figure}

Finally, in Figure~\ref{fig:phase-trans-gaussian}, we plot the
performance of the nuclear norm heuristic when recovering low-rank
matrices from Gaussian projections of these matrices.  In these
cases, $\mtx{M}$ was generated in the same fashion as above, but, in
place of sampling entries, we generated $m$ random Gaussian
projections of the data (see the discussion in
Section~\ref{sec:background}). Then we solved the optimization
\begin{equation*}
  \begin{array}{ll}
\textrm{minimize}   & \quad  \|\mtx{X}\|_*\\
\textrm{subject to} & \quad  {\cal A}(\mtx{X}) = {\cal A}(\mtx{M})
 \end{array}.
\end{equation*}
with the additional constraint that $\mtx{X}\succeq 0$ in the
positive semidefinite case.  Here $\mathcal{A}(\mtx{X})$ denotes a
linear map of the form (\ref{eq:linearfunctional}) where the entries
are sampled i.i.d. from a zero-mean unit variance Gaussian
distribution. In these experiments, the recovery regime is far
larger than in the case of that of sampling entries, but this is not
particularly surprising as each Gaussian observation measures a
contribution from every entry in the matrix $\mtx{M}$.  These
Gaussian models were studied extensively in~\cite{Recht07}.

\begin{figure}
 \centering
 \begin{tabular}{cc}
   \includegraphics[width=8cm]{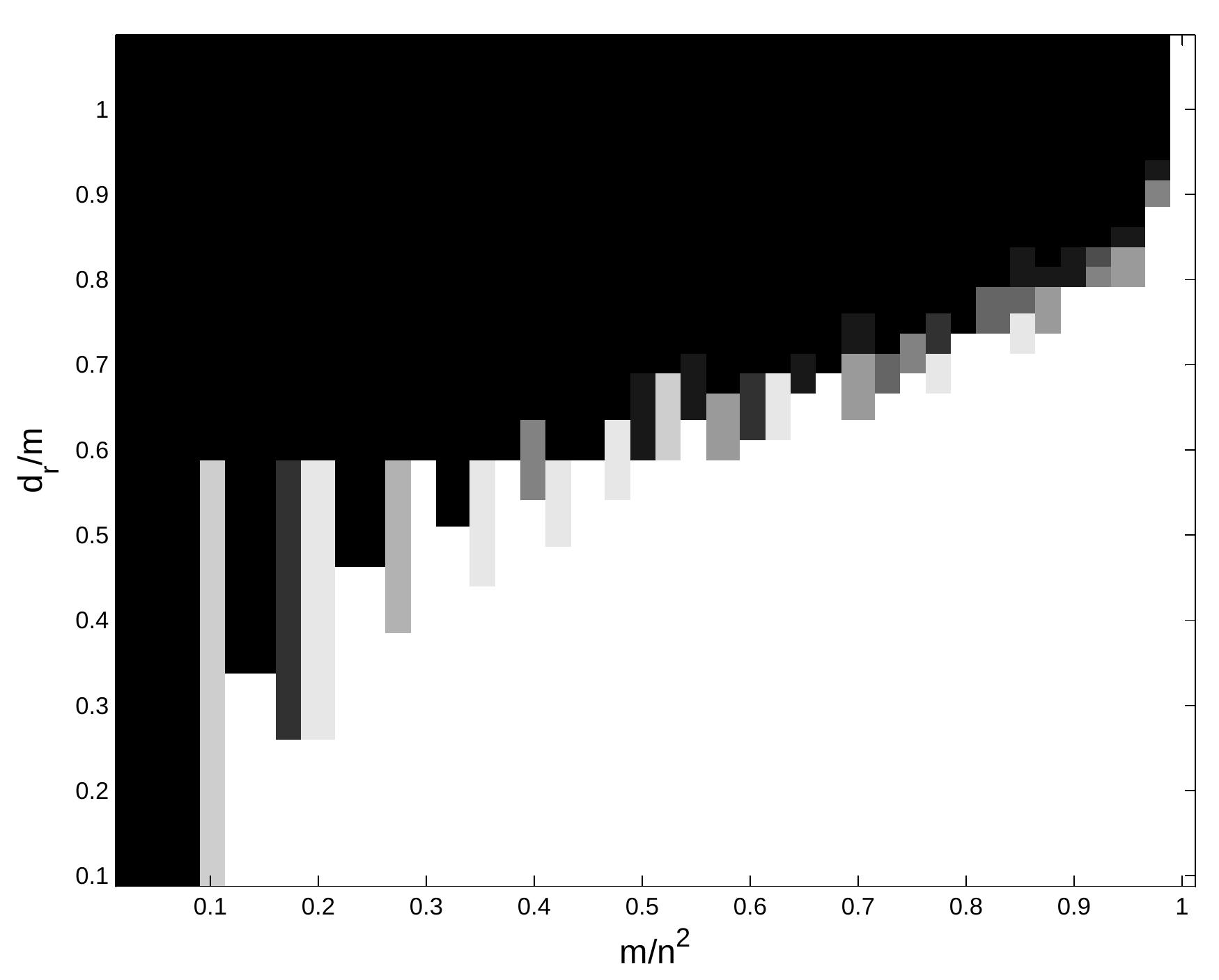}  &
   \includegraphics[width=8cm]{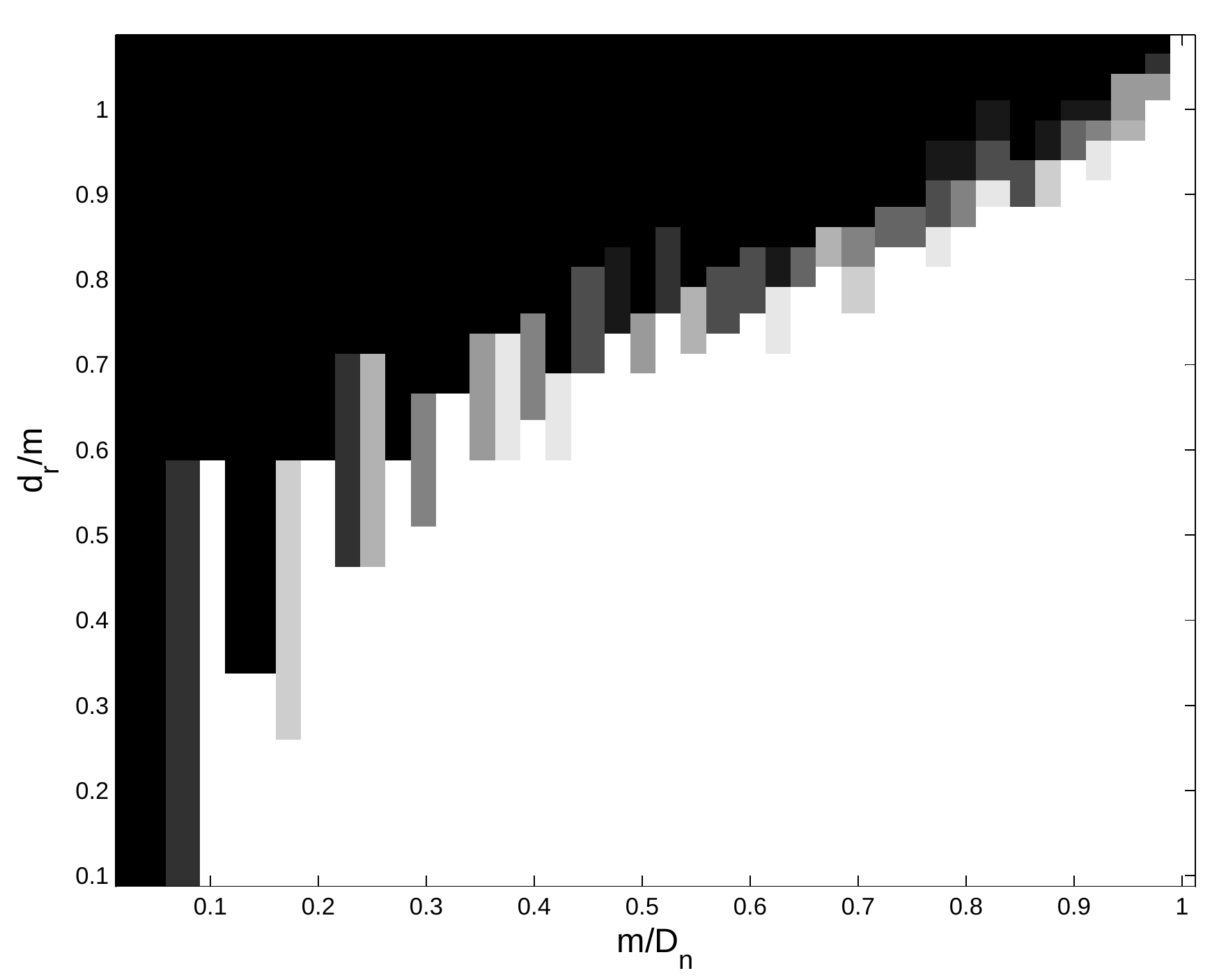}\\ (a) & (b)
     \end{tabular}
     \caption{\small {\bf Recovery of matrices from Gaussian
         observations.}  For each $(n,m,r)$ triple, we repeated the
       following procedure $10$ times. In (a), a matrix of rank $r$
       was generated as in Figures~\ref{fig:phase-trans-full}.  In (b)
       a positive semidefinite matrix of rank $r$ was generated as in
       Figures~\ref{fig:phase-trans-psd}. In both plots, we select a
       matrix $\mathcal{A}$ from the Gaussian ensemble with $m$ rows and
       $n^2$ (in (a)) or $D_n=n(n+1)/2$ (in (b)) columns. Then we
       solve the nuclear norm minimization subject to
       $\mathcal{A}(\mtx{X}) = \mathcal{A}(\mtx{M})$.  The color
       scheme for each cell denotes empirical recovery probability and
       is the same as in Figures~\ref{fig:phase-trans-full}
       and~\ref{fig:phase-trans-psd}.
     } \label{fig:phase-trans-gaussian}
\end{figure}

\section{Discussion}
\label{sec:discussion}

\subsection{Improvements}

In this paper, we have shown that under suitable conditions, one can
reconstruct an $n \times n$ matrix of rank $r$ from a small number
of its sampled entries provided that this number is on the order of
$n^{1.2} r \log n$, at least for moderate values of the rank. One
would like to know whether better results hold in the sense that
exact matrix recovery would be guaranteed with a reduced number of
measurements. In particular, recall that an $n \times n$ matrix of
rank $r$ depends on $(2n-r)r$ degrees of freedom; is it true then
that it is possible to recover most low-rank matrices from on the
order of $nr$---up to logarithmic multiplicative factors---randomly
selected entries? Can the sample size be merely proportional to the
true complexity of the low-rank object we wish to recover?

In this direction, we would like to emphasize that there is nothing
in our approach that apparently prevents us from getting stronger
results.  Indeed, we developed a bound on the spectral norm of each
of the first four terms $(\PTc \PO \PT)\mathcal{H}^k(E)$ in the
series \eqref{eq:size} (corresponding to values of $k$ equal to $0,
1, 2, 3$) and used a general argument to bound the remainder of the
series. Presumably, one could bound higher order terms by the same
techniques. Getting an appropriate bound on $\|(\PTc \PO \PT){\cal
  H}^4(E)\|$ would lower the exponent of $n$ from $6/5$ to $7/6$.  The
appropriate bound on $\|(\PTc \PO \PT)\mathcal{H}^5(E)\|$ would
further lower the exponent to $8/7$, and so on. To obtain an optimal
result, one would need to reach $k$ of size about $\log n$. In doing
so, however, one would have to pay special attention to the size of
the decoupling constants (the constant $C_D$ for two variables in
Lemma \ref{teo:decoupling}) which depend on $k$---the number of
decoupled variables. These constants grow with $k$ and upper bounds
are known \cite{delaPena1,delaPena2}.

\subsection{Further directions}

It would be of interest to extend our results to the case where the
unknown matrix is approximately low-rank. Suppose we write the SVD
of a matrix $\mtx{M}$ as
\[
\mtx{M} = \sum_{1 \le k \le n} \sigma_k \vct{u}_k \vct{v}_k^*,
\]
where $\sigma_1 \ge \sigma_2 \ge \ldots \ge \sigma_n \ge 0$ and assume for
simplicity that none of the $\sigma_k$'s vanish. In general, it is
impossible to complete such a matrix exactly from a partial subset
of its entries. However, one might hope to be able to recover a good
approximation if, for example, most of the singular values are small
or negligible. For instance, consider the truncated SVD of the
matrix $\mtx{M}$,
\[
\mtx{M}_r = \sum_{1 \le k \le r} \sigma_k \vct{u}_k \vct{v}_k^*,
\]
where the sum extends over the $r$ largest singular values and let
$\mtx{M}_\star$ be the solution to \eqref{eq:sdp}. Then one would
not expect to have $\mtx{M}_\star = \mtx{M}$ but it would be of
great interest to determine whether the size of $\mtx{M}_\star -
\mtx{M}$ is comparable to that of $\mtx{M} - \mtx{M}_r$ provided
that the number of sampled entries is sufficiently large. For
example, one would like to know whether it is reasonable to expect
that $\|\mtx{M}_\star - \mtx{M}\|_*$ is on the same order as
$\|\mtx{M} - \mtx{M}_r\|_*$ (one could ask for a similar comparison
with a different norm). If the answer is positive, then this would
say that approximately low-rank matrices can be accurately recovered
from a small set of sampled entries.

Another important direction is to determine whether the reconstruction
is robust to noise as in some applications, one would presumably
observe
\[
Y_{ij} = M_{ij} + z_{ij}, \quad (i,j) \in \Omega,
\]
where $z$ is a deterministic or stochastic perturbation. In this
setup, one would perhaps want to minimize the nuclear norm subject
to $\|\PO(\mtx{X}-\mtx{Y})\|_F \le \epsilon$ where $\epsilon$ is an
upper bound on the noise level instead of enforcing the equality
constraint $\PO(\mtx{X}) = \PO(\mtx{Y})$. Can one expect that this
algorithm or a variation thereof provides accurate answers?  That
is, can one expect that the error between the recovered and the true
data matrix be proportional to the noise level?

\section{Appendix}

\subsection{Proof of Theorem \ref{teo:rudelson2}}

The proof of \eqref{eq:largeRudel} follows that in \cite{CR07} but we
shall use slightly more precise estimates.

Let $Y_1, \ldots, Y_n$ be a sequence of independent random variables
taking values in a Banach space and let $Y_\star$ be the supremum
defined as
\begin{equation}
  \label{eq:Z}
  Y_\star = \sup_{f \in {\cal F}} \, \sum_{i = 1}^n f(Y_i),
\end{equation}
where ${\cal F}$ is a countable family of real-valued functions such
that if $f \in {\cal F}$, then $-f \in {\cal F}$.  Talagrand
\cite{talagrand96ne} proved a concentration inequality about $Y_\star$,
see also \cite[Corollary 7.8]{ledoux01co}.
\begin{theorem}
  \label{teo:LT}
  Assume that $|f| \le B$ and $\E f(Y_i) = 0$ for every $f$ in ${\cal
    F}$ and $i = 1, \ldots, n$. Then for all $t \ge 0$,
  \begin{equation}
    \label{eq:Talagrand}
    \P(|Y_\star - \E Y_\star| > t) \le 3 \exp\left(-\frac{t}{K B} \log \left(1 + \frac{B t}
        {\sigma^2 + B \E Y_\star}\right)\right),
  \end{equation}
  where $\sigma^2 = \sup_{f \in {\cal F}} \, \sum_{i = 1}^n \E
  f^2(Y_i)$, and $K$ is a numerical constant.
\end{theorem}
We note that very precise values of the numerical constant $K$ are
known and are small, see \cite{klein05co}.

We will apply this theorem to the random variable $Z$ defined in the
statement of Theorem \ref{teo:rudelson2}.  Put $\mathcal{Y}_{ab} = p^{-1}
(\delta_{ab} - p) \, \PT(\eab) \otimes \PT(\eab)$ and $\mathcal{Y} =
\sum_{ab} \mathcal{Y}_{ab}$. By definition,
\begin{align*}
  Z = \sup \,\, \<\mtx{X}_1, \mathcal{Y}(\mtx{X}_2)\> & =
  \sup \,\, \sum_{ab} \<\mtx{X}_1, \mathcal{Y}_{ab} (\mtx{X}_2)\>\\
  & = \sup \,\, p^{-1} \sum_{ab} (\delta_{ab} - p) \<\mtx{X}_1,
  \PT(\eab)\> \, \<\PT(\eab), \mtx{X}_2\>,
\end{align*}
where the supremum is over a countable collection of matrices
$\mtx{X}_1$ and $\mtx{X}_2$ obeying $\|\mtx{X}_1\|_F \le 1$ and
$\|\mtx{X}_2\|_F \le 1$. Note that it follows from \eqref{eq:PTeabF}
\begin{align*}
  |\<\mtx{X}_1, \mathcal{Y}_{ab}(\mtx{X}_2)\>| &= p^{-1}\, |\delta_{ab}-p|\,
  |\<\mtx{X}_1,
  \PT(\eab)\>| \, |\<\PT(\eab), \mtx{X}_2\>| \\
  &\le p^{-1} \, \|\PT(\eab)\|_F^2 \le 2\mu_0 r/(\min(n_1,n_2) p) =
  2\mu_0 \, nr/m
\end{align*}
(recall that $n = \max(n_1,n_2)$).  Hence, we can apply Theorem
\ref{teo:LT} with $B = 2\mu_0 (nr/m)$. Also
\begin{align*}
  \E |\<\mtx{X}_1, \mathcal{Y}_{ab}(\mtx{X}_2)\>|^2 & = p^{-1}(1-p) \, |\<\mtx{X}_1,
  \PT(\eab)\>|^2 \,
  |\<\mtx{X}_2, \PT(\eab)\>|^2 \\
  & \le p^{-1}\, \|\PT(\eab)\|_F^2 \, |\<\PT(\mtx{X}_2), \eab\>|^2
\end{align*}
so that
\begin{align*}
  \sum_{ab} \E |\<\mtx{X}_1, \mathcal{Y}_{ab}(\mtx{X}_2)\>|^2 & \le (2\mu_0 \, nr/m) \sum_{ab}   |\<\PT(\mtx{X}_2), \eab\>|^2\\
  & = (2 \mu_0 \, nr/m) \, \|\PT(\mtx{X}_2)\|_F^2 \le 2\mu_0 nr/m.
\end{align*}
Since $\E Z \le 1$, Theorem \ref{teo:LT} gives
\[
P(|Z - \E Z| > t) \le 3 \exp\left(-\frac{t}{KB} \log(1+ t/2)\right) \le 3
\exp\left(-\frac{t \log 2}{K B} \min(1, t/2)\right),
\]
where we have used the fact that $\log(1+u) \ge (\log 2)\, \min(1,u)$
for $u \ge 0$. Plugging $t = \lambda \sqrt{\frac{\mu_0\, nr \log
    n}{m}}$ and $B = 2\mu_0\, nr/m$ establishes the claim.

\subsection{Proof of Lemma \ref{teo:E-max-indep2}}

We shall make use of the following lemma which is an application of
well-known deviation bounds about binomial variables.
\begin{lemma}
  Let 
  $\{\delta_i\}_{1 \le i \le n}$ be a sequence of i.i.d.~Bernoulli
  variables with $\P(\delta_i = 1) = p$ and $Y = \sum_{i = 1}^n
  \delta_i$.  Then for each $\lambda > 0$,
\begin{equation}
\label{eq:binomial}
\P(Y > \lambda \, \E Y) \le \exp\left(-\frac{\lambda^2}{2 + 2\lambda/3}
  \, \E Y\right).
\end{equation}
\end{lemma}

The random variable $\sum_{b} \delta_{ab} E_{ab}^2$ is bounded by
$\|\mtx{E}\|_\infty^2 \, \sum_b \delta_{ab}$ and it thus suffices to
estimate the $q$th moment of $Y_* = \max Y_a$ where $Y_a = \sum_{b}
\delta_{ab}$. The inequality \eqref{eq:binomial} implies that
\[
\P(Y_* > \lambda np) \le n \, \exp\left(-\frac{\lambda^2}{2+2\lambda/3}
    \, np\right),
\]
and for $\lambda \ge 2$, this gives $\P(Y_* > \lambda np) \le n \,
e^{-\lambda np/2}$. Hence
\[
  \E Y_*^q = \int_0^\infty \P(Y_* > t) \, q t^{q-1} \, dt  \le
  (2np)^q + \int_{2np}^\infty
  n \, e^{-t/2} \,  q t^{q-1} \, dt.
\]
By integrating by parts, one can check that when $q \le np$, we have
\[
\int_{2np}^\infty
  n \, e^{-t/2} \,  q t^{q-1} \, dt \le nq \, (2np)^q \, e^{-np}.
\]
Under the assumptions of the lemma, we have $nq \, e^{-np} \le 1$ and,
therefore,
\[
 \E Y_*^q  \le 2 \, (2np)^q.
\]
The conclusion follows.

\subsection*{Acknowledgments}
E.~C. was partially supported by a National Science Foundation grant
CCF-515362, by the 2006 Waterman Award (NSF) and by an ONR grant. The
authors would like to thank Ali Jadbabaie, Pablo Parrilo, Ali Rahimi,
Terence Tao, and Joel Tropp for fruitful discussions about parts of
this paper. E.~C. would like to thank Arnaud Durand for his careful
proof-reading and comments.

\small
\bibliographystyle{plain}
\bibliography{MatrixCompletionFinal3}

\end{document}